\documentclass[11pt]{article}
\newcommand{\contentmode}{1}   % 1 = full
\newcommand{\classmode}{1}     % 1 = article

\newif\ifincludeAppPolicy   \includeAppPolicytrue   % app_01_iterative_policy_update_algorithm
\newif\ifincludeAppSolver   \includeAppSolvertrue   % app_02_implicit_solver_algorithm
\newif\ifincludeAppProofs   \includeAppProofstrue   % app_03_deferred_proofs
\newif\ifincludeAppPerf     \includeAppPerftrue     % app_04_performance_convergence

\newif\iffullcontent
\newif\ifshortcontent
\newif\ifarticleclass

\ifnum\contentmode=1 \fullcontenttrue  \fi
\ifnum\contentmode=2 \shortcontenttrue \fi
\ifnum\classmode=1   \articleclasstrue \fi

\ifarticleclass
  \usepackage[T1]{fontenc}
  \usepackage{lineno}
  \usepackage{hyperref}\hypersetup{hidelinks}
  \modulolinenumbers[5]
  \usepackage{amsmath,amssymb,mathtools,amsthm}
  \usepackage{xcolor}
  \DeclareMathOperator*{\argmin}{arg\,min}
\usepackage{graphicx,wrapfig,float,epstopdf}
  \usepackage{caption,subcaption}
  \usepackage{authblk}
  \usepackage{booktabs,array}
  \usepackage{placeins}
  \usepackage{setspace}
  \usepackage{fullpage,etoolbox}
  \usepackage{times}
  \usepackage{algorithm,algorithmic}
  \usepackage{amsthm}
\else
  \usepackage{times}
  \usepackage{float,epstopdf}
  \usepackage{mathtools}
  \usepackage{wrapfig,booktabs,array}
  \usepackage{algorithm,algorithmic}
\DeclareMathOperator*{\argmin}{arg\,min}

  \makeatletter

  \let\c@subfigure\undefined
  \let\c@subtable\undefined
  \let\p@subfigure\undefined
  \let\p@subtable\undefined
  \makeatother

  \usepackage[font=small,labelfont=bf,labelsep=colon,format=plain,justification=justified,singlelinecheck=false,aboveskip=4pt,belowskip=2pt]{caption}
  \usepackage{subcaption}
  \makeatletter
  \renewcommand*{\@jmlrenddoc}{%
    \FloatBarrier
    \global\let\@reprint\@empty
  }
  \newcommand{\jmlrendhere}{%
    \phantomsection
    \protected@edef\@currentlabelname{end of \@shorttitle}%
    \label{jmlrend}%
  }
  \makeatother
\fi

\expandafter\ifx\csname compactenvsloaded\endcsname\relax\else \fi
\expandafter\let\csname compactenvsloaded\endcsname\empty

\makeatletter

\newlength{\compactthmtopsep}
\newlength{\compactthmbotsep}
\newlength{\compactthmheadsep}

\newlength{\compactsectopsep}
\newlength{\compactsecbotsep}

\newcommand{\newshorttheorem}[4]{%
  \newenvironment{#1}[1][]{%
    \par
    \vspace{\compactthmtopsep}%
    \refstepcounter{#3}%
    \def\@compactthm@arg{##1}%
    \ifx\@compactthm@arg\@empty
      \def\@compactthm@head{#2~\csname the#3\endcsname}%
    \else
      \def\@compactthm@head{#2~\csname the#3\endcsname\space(##1)}%
    \fi
    \noindent\textbf{\@compactthm@head}\hspace{\compactthmheadsep}#4\ignorespaces
  }{\par\vspace{\compactthmbotsep}}%
}

\newshorttheorem{theoremshort}     {Theorem}     {theorem}    {\itshape}
\newshorttheorem{propositionshort} {Proposition} {proposition}{\itshape}
\newshorttheorem{lemmashort}       {Lemma}       {lemma}      {\itshape}
\newshorttheorem{corollaryshort}   {Corollary}   {corollary}  {\itshape}

\newshorttheorem{definitionshort}  {Definition}  {definition} {\upshape}
\newshorttheorem{assumptionshort}  {Assumption}  {assumption} {\upshape}
\newshorttheorem{exampleshort}     {Example}     {example}    {\upshape}
\newshorttheorem{problemshort}     {Problem}     {problem}    {\upshape}

\newshorttheorem{remarkshort}      {Remark}      {remark}     {\upshape}

\providecommand{\compactsecfont}{\large\bfseries}
\providecommand{\compactsubsecfont}{\normalsize\bfseries}
\providecommand{\compactsubsubsecfont}{\normalsize\bfseries}

\providecommand{\compactsectopsep}{6pt plus 2pt minus 1pt}
\providecommand{\compactsecbotsep}{4pt plus 1pt minus 1pt}

\newcommand{\sectionshort}[1]{%
  \par
  \vspace{\compactsectopsep}%
  \refstepcounter{section}%
  \setcounter{subsection}{0}%
  \setcounter{subsubsection}{0}%
  \addcontentsline{toc}{section}%
    {\protect\numberline{\thesection}#1}%
  \noindent{\compactsecfont\thesection\quad #1}\par
  \vspace{\compactsecbotsep}%
  \@afterheading\ignorespaces
}

\newcommand{\subsectionshort}[1]{%
  \par
  \vspace{\compactsectopsep}%
  \refstepcounter{subsection}%
  \setcounter{subsubsection}{0}%
  \addcontentsline{toc}{subsection}%
    {\protect\numberline{\thesubsection}#1}%
  \noindent{\compactsubsecfont\thesubsection\quad #1}\par
  \vspace{\compactsecbotsep}%
  \@afterheading\ignorespaces
}

\newcommand{\subsubsectionshort}[1]{%
  \par
  \vspace{\compactsectopsep}%
  \refstepcounter{subsubsection}%
  \addcontentsline{toc}{subsubsection}%
    {\protect\numberline{\thesubsubsection}#1}%
  \noindent{\compactsubsubsecfont\thesubsubsection\quad #1}\par
  \vspace{\compactsecbotsep}%
  \@afterheading\ignorespaces
}

\providecommand{\compactparfont}{\normalsize\bfseries}
\providecommand{\compactpartopsep}{3pt plus 1pt minus 0.5pt}

\newcommand{\sectionshortstar}[1]{%
  \par\vspace{\compactsectopsep}%
  {\noindent\textbf{#1}\par}%
  \vspace{\compactsecbotsep}%
  \@afterheading
}

\makeatother
\newcommand{\eqgentle}{%
  \setlength{\abovedisplayskip}{4pt plus 1pt minus 1pt}%
  \setlength{\belowdisplayskip}{4pt plus 1pt minus 1pt}%
  \setlength{\abovedisplayshortskip}{0pt plus 1pt}%
  \setlength{\belowdisplayshortskip}{4pt plus 1pt minus 1pt}%
  \setlength{\jot}{2pt}%
}
\newcommand{\eqtight}{%
  \setlength{\abovedisplayskip}{2pt plus 1pt}%
  \setlength{\belowdisplayskip}{2pt plus 1pt}%
  \setlength{\abovedisplayshortskip}{0pt plus 1pt}%
  \setlength{\belowdisplayshortskip}{2pt plus 1pt}%
  \setlength{\jot}{1pt}%
}

\newcommand{\eqgentleon}{\begingroup\eqgentle}
\newcommand{\eqtighton}  {\begingroup\eqtight}

\newcommand{\eqend}{\endgroup}

\newcommand{\subeqtighton}{\ifshortcontent\eqtighton\fi}
\newcommand{\subeqgentleon}{\ifshortcontent\eqgentleon\fi}

\newcommand{\subeqend}{\ifshortcontent\eqend\fi}

\ifarticleclass
  \newtheorem{theorem}{Theorem}
  \newtheorem{lemma}{Lemma}
  \newtheorem{proposition}{Proposition}
  \newtheorem{corollary}{Corollary}
  \newtheorem{definition}{Definition}
  \newtheorem{problem}{Problem}
  \newtheorem{remark}{Remark}
  \newtheorem{example}{Example}
  \newtheorem{assumption}{Assumption}
\else
  \makeatletter
        
  \let\c@theorem\undefined    
          
  \let\c@lemma\undefined      
    
  \let\c@proposition\undefined 
      
  \let\c@corollary\undefined  
     
  \let\c@definition\undefined 
         
  \let\c@remark\undefined     
        
  \let\c@example\undefined    
  \makeatother
  \newtheorem{theorem}{Theorem}
  \newtheorem{lemma}{Lemma}
  \newtheorem{proposition}{Proposition}
  \newtheorem{corollary}{Corollary}
  \newtheorem{definition}{Definition}
  \newtheorem{remark}{Remark}
  \newtheorem{example}{Example}
  \newtheorem{assumption}{Assumption}
  
  \numberwithin{equation}{section}
\fi

\newcommand{\R}{\mathbb{R}}
\newcommand{\calR}{\mathcal{R}}

\newcommand{\norm}[1]{\left\lVert#1\right\rVert}
\newcommand{\abs}[1]{\left\lvert#1\right\rvert}
\newcommand{\blue}[1]{#1}

\graphicspath{{figure/}{plots_homogenous/}}
\allowdisplaybreaks

\newcounter{pitem}[theorem]
\renewcommand{\thepitem}{P\arabic{pitem}}

\newcommand{\pitem}[2]{%
  \par\smallskip                         % close previous paragraph + spacing
  \noindent                              % start the new line
  \refstepcounter{pitem}\label{#1}%      % anchor lands HERE, on the (Pk) line
  \textbf{(\thepitem) #2}\ \ignorespaces
}

\newcommand{\pitemshort}[2]{%
  \par                                   % close previous paragraph
  \vspace{\compactpartopsep}%            % same topskip as \paragraphshort
  \noindent                              % start the new line
  \refstepcounter{pitem}\label{#1}%      % anchor lands HERE, on the (Pk) line
  {\compactparfont (\thepitem) #2}\hspace{0.5em}\ignorespaces
}

\expandafter\ifx\csname versiontoggleloaded\endcsname\relax\else \fi
\expandafter\let\csname versiontoggleloaded\endcsname\empty

\usepackage{comment}

\iffullcontent
  \includecomment{fullblock}
\else
  \excludecomment{fullblock}
\fi

\ifshortcontent
  \includecomment{shortblock}
\else
  \excludecomment{shortblock}
\fi

\newcommand{\fullonly}[1]{\iffullcontent #1\fi}
\newcommand{\shortonly}[1]{\ifshortcontent #1\fi}

\newcommand{\articleonly}[1]{\ifarticleclass #1\fi}
\newcommand{\ldconly}[1]{\ifarticleclass\else #1\fi}

\ifarticleclass
  \includecomment{articleblock}
  \excludecomment{l4dcblock}
\else
  \excludecomment{articleblock}
  \includecomment{l4dcblock}
\fi

\newcommand{\fulleq}[2]{%
  \iffullcontent
    \begin{equation}#1\end{equation}%
  \else
    #2%
  \fi
}

\ifarticleclass
  \title{Safe Planning in Interactive Environments via Iterative Policy Updates and Adversarially Robust Conformal Prediction}
  \author[1]{Omid Mirzaeedodangeh}
  \author[2]{Eliot Shekhtman}
  \author[2,3]{Nikolai Matni}
  \author[1]{Lars Lindemann}
  \affil[1]{Automatic Control Laboratory, ETH Z\"urich, Switzerland}
  \affil[2]{Computer and Information Science, University of Pennsylvania, USA}
  \affil[3]{Electrical and Systems Engineering, University of Pennsylvania, USA}
  \date{\vspace{-5ex}}
  \providecommand{\keywords}[1]{\textbf{\textit{Index terms---}} #1}
\else
  \title[Safe Planning in Interactive Environments]{Safe Planning in Interactive Environments via Iterative Policy Updates and Adversarially Robust Conformal Prediction}

\author{%
 \Name{Omid Mirzaeedodangeh} \Email{omirzaeedoda@ethz.ch}\\
 \addr Automatic Control Laboratory, ETH Zürich, Switzerland
 \AND
 \Name{Eliot Shekhtman} 
 \Email{shekhe@seas.upenn.edu}\\
 \addr Computer and Information Science, University of Pennsylvania, USA
 \AND
 \Name{Nikolai Matni}\thanks{Eliot Shekhtman and Nikolai Matni were partially supported by AFOSR Award FA9550-24-1-0102, SF Award SLES-2331880, and NSF CAREER award ECCS-2045834.}
 \Email{nmatni@seas.upenn.edu}\\
 \addr Electrical and Systems Engineering, University of Pennsylvania, USA\\
  \addr Computer and Information Science, University of Pennsylvania, USA
 \AND
 \Name{Lars Lindemann} \Email{llindemann@ethz.ch}\\
 \addr Automatic Control Laboratory, ETH Zürich, Switzerland
\vspace{-0.5cm}
}
\fi

\begin{document}
% =============================================================

\maketitle
\ifarticleclass
  \thispagestyle{plain}\pagestyle{plain}
  \renewcommand{\baselinestretch}{0.97}
\fi

% --- Abstract (shared file) ----------------------------------
\begin{abstract}
Safe planning of an autonomous agent in interactive environments -- such as the control of a self-driving vehicle among pedestrians -- poses a major challenge as the behavior of the environment is unknown and reactive to the behavior of the autonomous agent. This coupling gives rise to interaction-driven distribution shifts where the autonomous agent's control policy may change the environment's behavior, thereby invalidating safety guarantees in existing work. Indeed, recent works have used conformal prediction (CP) to generate distribution-free safety guarantees using observed data of the environment. However, CP's assumption on data exchangeability is violated in interactive settings due to a circular dependency where a control policy update changes the environment's behavior, and vice versa. To address this gap, we propose an iterative framework that robustly maintains safety guarantees across policy updates by  quantifying the potential impact of a planned policy update on the environment's behavior. We realize this via adversarially robust CP where we perform a regular CP step  in each episode using observed data under the current policy, but then transfer safety guarantees across policy updates by analytically adjusting the CP result to account for distribution shifts. This adjustment is performed based on a policy-to-trajectory sensitivity analysis, resulting in a safe, episodic open-loop planner. We further conduct a contraction analysis of the system providing conditions under which both the CP results and the policy updates are guaranteed to converge. We empirically demonstrate these safety and convergence guarantees on a two-dimensional car-pedestrian and a high-dimensional quadcopter case study. To the best of our knowledge, these are the first results that
provide valid safety guarantees in such interactive settings.
\end{abstract}

% \keywords{
% Safe planning in interactive environments, distribution shifts, adversarially robust conformal prediction, iterative control policy updates.}

% --- Keywords (per-class syntax) -----------------------------
\ifarticleclass
\else
\fi

% =============================================================
% SHARED BODY SECTIONS
% =============================================================
% These four sections appear in all three builds. Use
% \begin{fullblock} / \begin{shortblock} INSIDE the section
% files to express content differences between full and short.
% =============================================================
\section{Introduction} \label{sec:introduction}

Autonomous agents, e.g., self-driving vehicles and service robots, are increasingly deployed in human-centric, multi-agent environments \cite{mavrogiannis2023core,feng2025multiagentembodiedaiadvances}. 
The safe control of an autonomous agent  is challenging as the behavior of uncontrollable agents, e.g., pedestrians, are unknown and interactive, i.e., they may react to the behavior of the autonomous agent. This creates an intricate coupling and interaction-driven distribution shift where the distribution of uncontrollable agent behaviors changes with the control policy of the autonomous agent. In this paper, we address this   “chicken-and-egg” problem where changing the control policy changes the behavior of uncontrollable agents, and vice versa, see Figure \ref{fig:motivation}.

\ifarticleclass
  % BUILD 1 — arxiv extended
    \begin{wrapfigure}[17]{r}{0.4\textwidth}
    \centering
    \vspace{-0.2cm}
    \includegraphics[width=1\linewidth]     {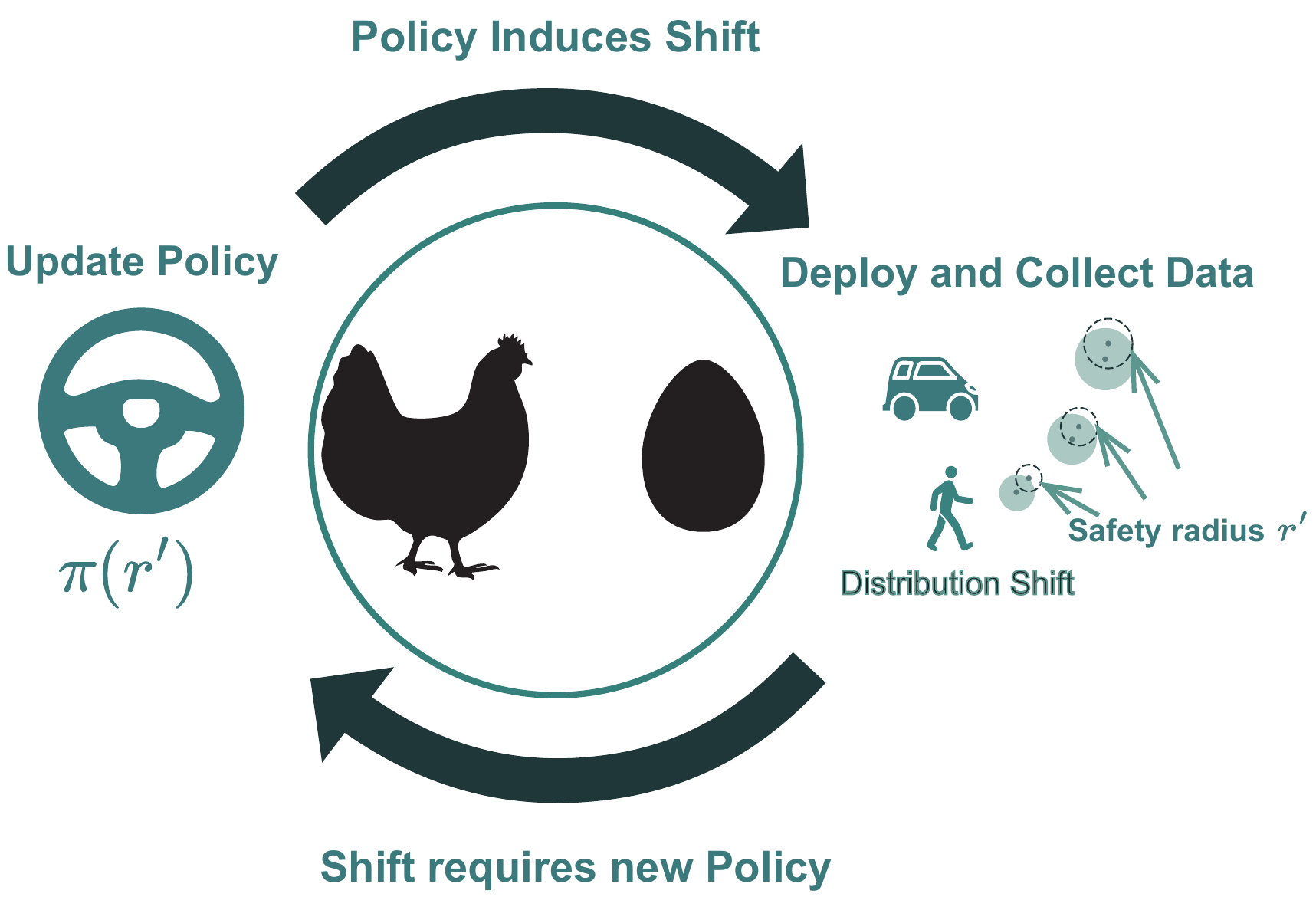}
    \caption{“Chicken-and-egg” problem: A change in policy can induce a distribution shift in the environment, here a pedestrian. This shift results in a modified safety radius $r'$ that captures the pedestrian's behavior under this  policy. This modified safety radius in turn requires a policy update $\pi(r')$.}
    \label{fig:motivation}
    \end{wrapfigure}
\else
  \iffullcontent
    % BUILD 2 — L4DC extended
    \begin{wrapfigure}[15]{r}{0.315\textwidth}
      \centering
      \vspace{-0.1cm}%
      \includegraphics[width=\linewidth]{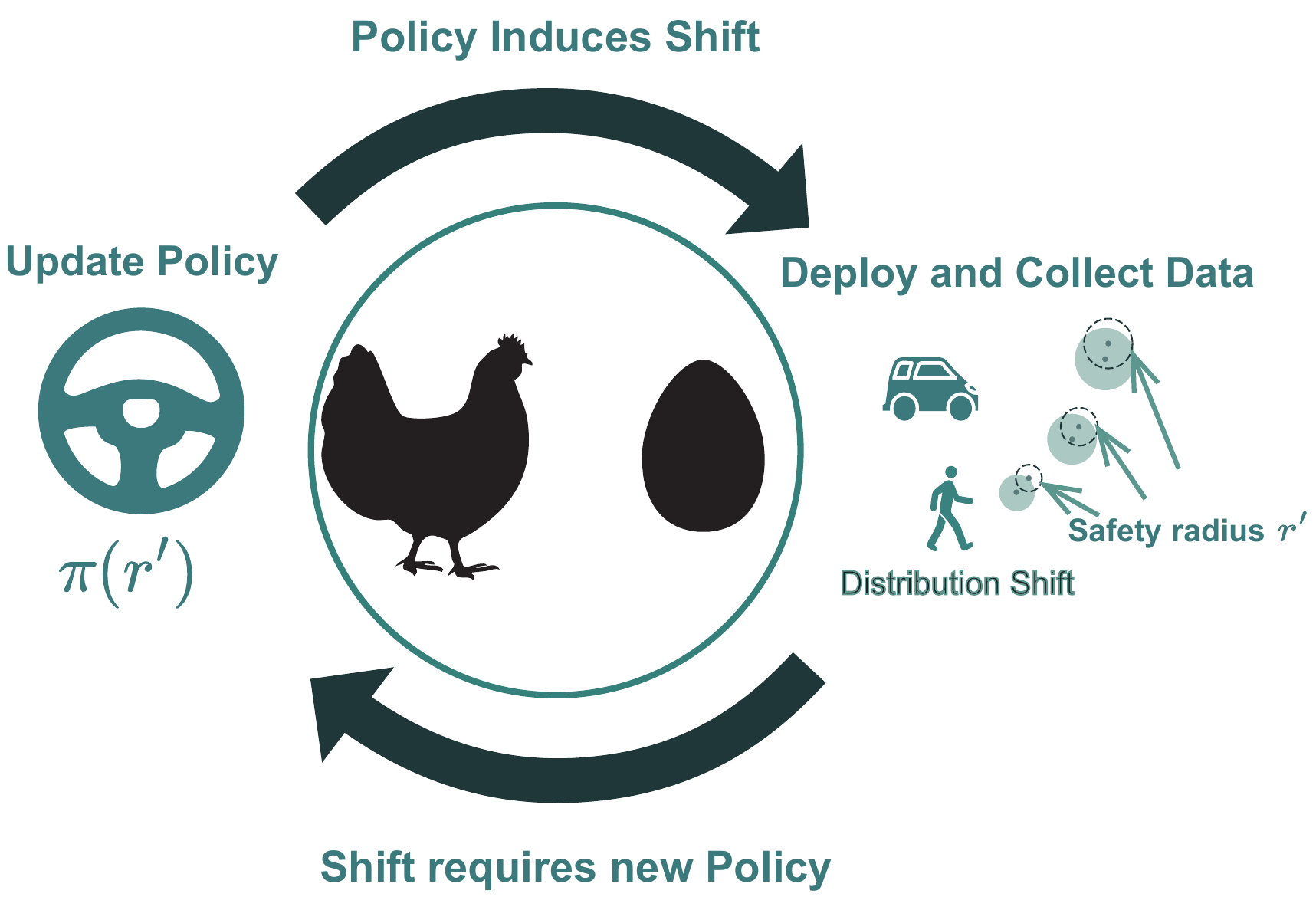}
      \caption{``Chicken-and-egg'' problem: A policy change can induce a distribution shift in the environment, here a pedestrian. This shift results in a modified safety radius $r'$ that captures the pedestrian's behavior under this policy. This modified safety radius in turn requires a policy update $\pi(r')$.}
      \label{fig:motivation}
    \end{wrapfigure}
  \else
    % BUILD 3 — L4DC submission
    \begin{wrapfigure}[15]{r}{0.315\textwidth}
      \centering
      \vspace{-0.1cm}%
      \includegraphics[width=\linewidth]{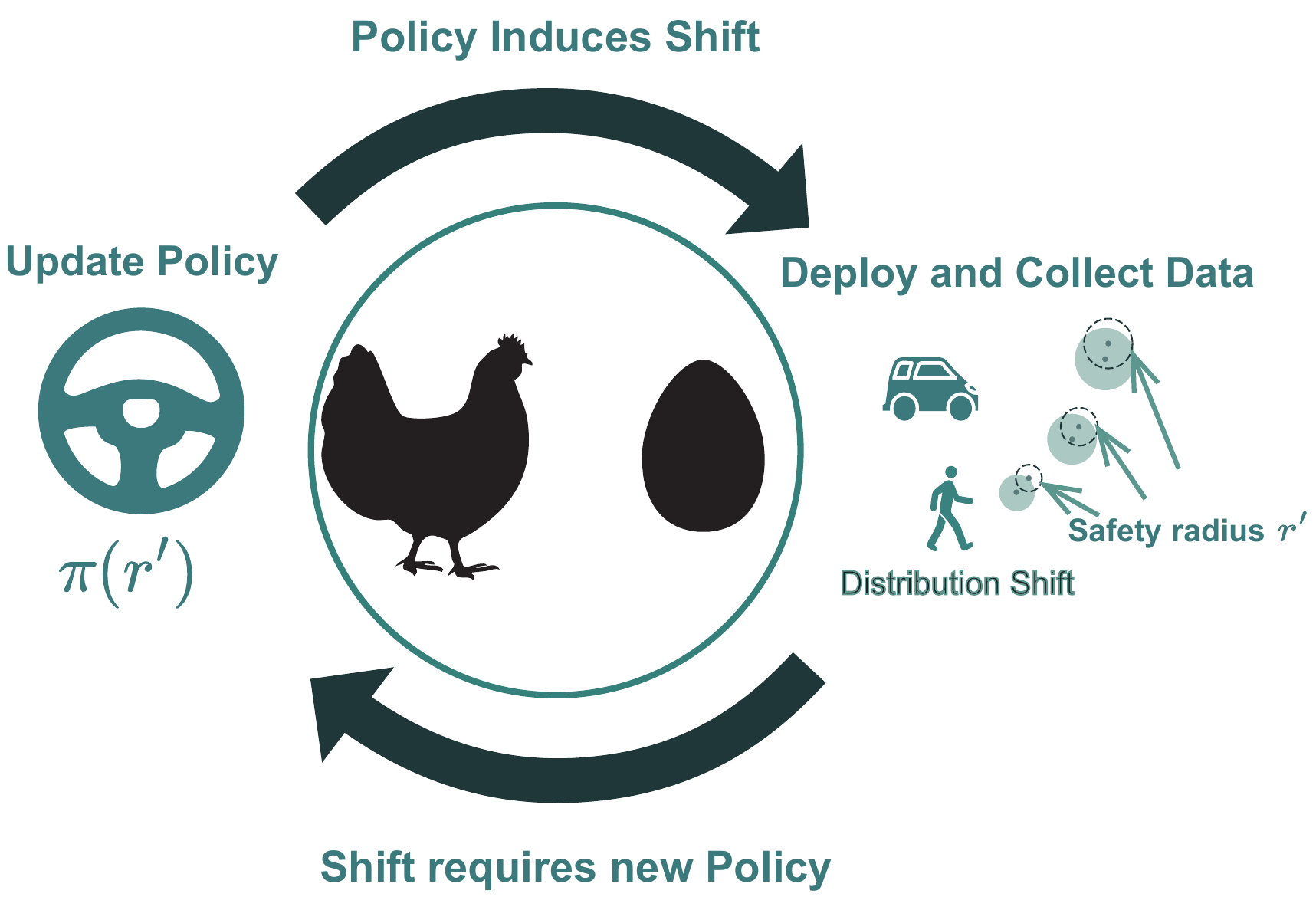}
      \caption{``Chicken-and-egg'' problem: A policy change can induce a distribution shift in the environment, here a pedestrian. This shift results in a modified safety radius $r'$ that captures the pedestrian's behavior under this policy. This modified safety radius in turn requires a policy update $\pi(r')$.}
      \label{fig:motivation}
    \end{wrapfigure}
  \fi
\fi
% \begin{wrapfigure}[15]{r}{0.315\textwidth}
%   \centering
%   \articleonly{\vspace{-0.2cm}}%
%   \ldconly{\vspace{-0.1cm}}%
%   \includegraphics[width=1\linewidth]{figure/OMID_SHAPE_X_2.pdf}
%   \caption{``Chicken-and-egg'' problem: A policy change can induce a distribution shift in the environment, here a pedestrian. This shift results in a modified safety radius $r'$ that captures the pedestrian's behavior under this policy. This modified safety radius in turn requires a policy update $\pi(r')$.}
%   \label{fig:motivation}
% \end{wrapfigure}
% \begin{wrapfigure}[17]{r}{0.4\textwidth}
% \centering
% \vspace{0.3cm}
% \includegraphics[width=1\linewidth]{figure/OMID_SHAPE_X_2.pdf}
%     \caption{“Chicken-and-egg” problem: A change in policy can induce a distribution shift in the environment, here a pedestrian. This shift results in a modified safety radius $r'$ that captures the pedestrian's behavior under this  policy. This modified safety radius in turn requires a policy update $\pi(r')$.}
%     \label{eq:motivation}
% \end{wrapfigure}
Non-interactive approaches sequentially integrate predictions of uncontrollable agents into the design of the control policy \cite{trautman2010unfreezing,du2011robot}, while interactive approaches simultaneously predict uncontrollable agents and design control policies to take interactions directly into account \cite{kretzschmar2016socially,everett2021collision}. Non-interactive approaches usually assume worst-case interaction bounds to provide safety guarantees,  tending to be conservative. Interactive approaches perform better in practice, but fail to provide any safety guarantees. We instead propose a planning framework that takes interaction into account and enjoys  safety guarantees via iterative policy updates and robust conformal prediction.

\textbf{Safe planning with conformal prediction.} Conformal prediction (CP) is an uncertainty quantification technique that provides  statistical guarantees for test data after a one-time calibration on held-out calibration data \cite{vovk2005algorithmic, angelopoulos2021gentle}. CP has been used to construct prediction sets for uncontrollable agents that contain their true but unknown behavior with a user-specified probability \cite{cleaveland2024conformal,sun2022copula,zecchin2024forking}. These works learn a trajectory predictor  for uncontrollable agents and perform calibration over a nonconformity score that measures the deviation between the predicted and the true motion on held-out trajectories. The resulting calibration result defines a set around new predictions that is probabilistically valid and thereby operational, i.e., suitable for motion planning. This has  been explored via model predictive control \cite{lindemann2022safe,dixit2023adaptive,shin2025egocentric}, reinforcement learning \cite{yao2024sonic,sun2023conformal}, control barrier functions \cite{zhang2025conformal,hsu2025statistical}, sampling-based search \cite{sheng2024safe}, and LLM planning \cite{doula2025safepath,wang2024probabilistically}, see \cite{lindemann2024formal} for a survey.

\textbf{Interaction breaks conformal prediction guarantees.} CP presumes that the data at test time  is exchangeable with the data at calibration time.\footnote{We note that independent and identically distributed data is automatically also exchangeable.} Indeed, simply using non-exchangeable calibration data can break conformal prediction guarantees, see e.g., \cite{tibshirani2019conformal,zhao2024robust}. In interactive settings the premise of exchangeability is violated due to a circular dependency induced by interaction-driven distribution shifts: changing the control policy changes the data distribution at test time, while updating the held-out calibration data changes the control policy. Many CP extensions aim to mitigate this issue, e.g., via adaptive CP \cite{gibbs2021adaptive, zaffran2022adaptive}, CP for covariate shifts \cite{tibshirani2019conformal, yang2024doubly, alijani2025wqlcp}, or robust CP \cite{cauchois2024robust, aolaritei2025conformal,GendlerEtAl2022ICLR}. However, these extensions either require continual recalibration or additional estimators, yield undersirable time-averaged or worst-case guarantees, or do not close the loop with a planner whose policy changes drive the distribution shift. Our interactive setting demands a mechanism that (i) directly addresses policy-induced distribution shifts, and (ii) closes the loop with the planner.

\textbf{Positioning and differences to prior work on planning with CP.} Prior work embeds calibrated CP sets into the planner by either assuming exchangeability \cite{lindemann2022safe,tonkens2023scalable,YU2026112616} or by using adaptive CP \cite{dixit2023adaptive,yao2024sonic,shin2025egocentric,sheng2024safe}.  Conceptually closest to our work are \cite{wang2025learning} and \cite{huang2025interaction}, which also address agent interactions with CP; \cite{wang2025learning} discretizes the joint agent space and uses CP to capture state-dependent interaction. However, data requirements grow exponentially with the state dimension without addressing the aforementioned core circular dependency. The work \cite{huang2025interaction} proposes iterative policy updates similar to our method, but without providing episodic safety guarantees as we enable via adversarially robust CP. Additionally, \cite{huang2025interaction} provide safety guarantees only after convergence, which is generally not guaranteed.  In contrast, we provide: (i) episodic safety guarantees by {analytically bounding} the distribution shift,  and (ii)  explicit conditions on episodic convergence via a contraction analysis.

\begin{fullblock}
\section{Problem Formulation}
\label{sec:pf}
\end{fullblock}
\begin{shortblock}
\sectionshort{Problem Formulation}
\label{sec:pf}
\end{shortblock}

%\subsection{System and Interaction Dynamics}
%\label{subsec:system}
We consider the state of an ego agent $x_t\in\R^{d_x}$, its control input $u_t\in\R^{d_u}$, and the state of uncontrollable agents' $y_t\in\R^{d_y}$ at discrete times $t=0,\ldots,T$. Their coupled dynamics are
\begin{fullblock}
\begin{subequations}\label{eq:pf:dyn}
\begin{align}
x_{t+1}&=f_X(x_t,u_t)\\
y_{t+1}&=f_Y(y_t,x_t,u_t,\nu_t)
\end{align}
\end{subequations}
\end{fullblock}
\begin{shortblock}
\subeqgentleon
\begin{equation}\label{eq:pf:dyn}
x_{t+1}=f_X(x_t,u_t),\qquad y_{t+1}=f_Y(y_t,x_t,u_t,\nu_t), \qquad t = 0,...,T-1.
\end{equation}
\eqend
\end{shortblock}
with exogenous noise $\nu_t$, which is a random variable that models the uncontrollable agents' intentions. The mapping $(x_t,u_t)\mapsto f_Y(\cdot,x_t,u_t,\cdot)$ captures agent interaction. Here, safety is encoded by a  safety function $H(x_{0:T},y_{0:T})$ and defined over the entire trajectories $x_{0:T}$ and $y_{0:T}$. The system is considered safe if $H(x_{0:T},y_{0:T}) \le 0$. We assume that the system \eqref{eq:pf:dyn} is Lipschitz continuous.
\begin{assumption}
\label{ass:lipschitz_dynamics}
There exist Lipschitz constants $L_{Xx}, L_{Xu}, L_{Yy}, L_{Yx},L_{Yu}\ge 0$ such that
\iffullcontent
\begin{align*}
 \|f_{X}(x,u)-f_{X}(x^{\prime},u^{\prime})\|_2 &\le L_{Xx}\|x-x^{\prime}\|_2+L_{Xu}\|u-u^{\prime}\|_2 \\
 \|f_{Y}(y,x,u,\nu)-f_{Y}(y^{\prime},x^{\prime},u^{\prime},\nu)\|_2 &\le L_{Yy}\|y-y^{\prime}\|_2+L_{Yx}\|x-x^{\prime}\|_2+L_{Yu}\|u-u^{\prime}\|_2
\end{align*}
\else
$\|f_{X}(x,u)-f_{X}(x^{\prime},u^{\prime})\|_2 \le L_{Xx}\|x-x^{\prime}\|_2+L_{Xu}\|u-u^{\prime}\|_2$ and $\|f_{Y}(y,x,u,\nu)-f_{Y}(y^{\prime},x^{\prime},u^{\prime},\nu)\|_2 \le L_{Yy}\|y-y^{\prime}\|_2+L_{Yx}\|x-x^{\prime}\|_2+L_{Yu}\|u-u^{\prime}\|_2$
\fi
for all $x,x'\in \R^{d_x}$, $u,u'\in\R^{d_u}$, and  $y,y'\in\R^{d_y}$.
\end{assumption}

\textbf{Running Example. }We use a running example of a {self-driving vehicle}  and a {pedestrian} with states $x_t$ and $y_t$ (e.g., position, velocity) and controls $u_t$ (e.g., acceleration, steering).  The dynamics $f_Y(y_t,x_t,u_t,\nu_t)$ capture interaction as the pedestrian's future position $y_{t+1}$ depends on their own state $y_t$, the vehicle's states $x_t$ and actions $u_t$, and their own intentions $\nu_t$.
 A common way to define $H$ is as the maximum violation of a separation constraint over time, e.g.,
$
H(x_{0:T},y_{0:T})\ :=\ \max_{0\le t\le T}\big\{\, b_t - c(x_t,y_t)\,\big\}\ \le\ 0$)
where $c(x_t,y_t)$ represents a separation measure such as the distance between the vehicle's position and the pedestrian's position and $b_t \ge 0$ is a safety margin. 
% The safety condition \eqref{eq:safety} is thus met if the separation $c(x_t,y_t)$ remains greater than or equal to the required margin $b_t$ at all times.

% \subsection{The Idealized Chance-Constrained Planning Problem}
% \label{subsec:ideal}
\begin{fullblock}
\subsection{The Idealized Chance-Constrained Planning Problem}\label{sec:idealized}
\end{fullblock}
\begin{shortblock}
\subsectionshort{The Idealized Chance-Constrained Planning Problem}\label{sec:idealized}
\end{shortblock}
Our ideal (unfortunately unattainable) objective is to find a control policy $u_{0:T-1}$ that minimizes a performance cost $J(x_{0:T},u_{0:T-1})$ while satisfying a probabilistic safety guarantee:
\begin{equation}
\min_{u_{0:T-1}}\,J(x_{0:T},u_{0:T-1})
\quad\text{s.t.}\quad
\begin{cases}
\mathbb{P}\!\left\{H(x_{0:T},y_{0:T})\le 0\right\}\ge 1-\alpha \\
\text{Dynamics in equation } \eqref{eq:pf:dyn}
\end{cases}
\label{eq:pf:chance}
\end{equation}
The problem \eqref{eq:pf:chance} is generally intractable. First, the uncontrollable agents' dynamics $f_Y$ and the noise distribution of $\nu_t$ are typically unknown, e.g., we cannot precisely model interaction and intentions, requiring a {model- and distribution-free approach}. Second, the {high-dimensionality of the problem} induces computational complexity: even if $f_Y$ and $\nu_t$ were known, the chance constraint \eqref{eq:pf:chance} would require solving a complex, high-dimensional integral over the  distribution of $y_{0:T}$. 

While the first two challenges have been addressed in the literature, the third and most difficult challenge arises due to \textbf{interaction-induced distribution shifts}. Effectively, the dynamics $f_Y$ and the distribution of $y_{0:T}$ are {policy-dependent}, meaning they change when the ego-agent's states $x_t$ and control inputs $u_t$ change, e.g., if the car accelerates aggressively, the pedestrian's decision to cross an intersection will change. Robust control, which enforces safety for all permissible trajectories, is overly conservative. Statistical uncertainty quantification techniques -- such as in \cite{lindemann2022safe} which use conformal prediction (CP) -- fail as the policy-driven distribution shift breaks exchangeability assumptions needed in CP, e.g., data of a pedestrian crossing in front of a {slow} car is not {exchangeable} with the new scenario where the car {accelerates}. 

% \subsection{Non-Interactive Planning with Distribution-Free Certificates via Conformal Prediction}
% \label{subsec:noninteractive}
\begin{fullblock}
\subsection{Non-Interactive Planning with Distribution-Free Certificates via Conformal Prediction}\label{subsec:noninteractive}
\end{fullblock}
\begin{shortblock}
\subsectionshort{Non-Interactive Planning with Distribution-Free Certificates via Conformal Prediction}\label{subsec:noninteractive}
\end{shortblock}
We now summarize existing work on non-interactive planning -- primarily following \cite{lindemann2022safe} -- in which case $f_Y$ does not depend on $x_t$ and $u_t$, i.e., not addressing the third of the aforementioned challenges. These techniques will serve as a starting point for our proposed method.

First, a pre-designed offline predictor is used to produce a single, nominal estimate $\hat y_{0:T}=(\hat y_0,\ldots,\hat y_T)$ of the environment trajectory $y_{0:T}=(y_0,\ldots,y_T)$. To assess the accuracy of this estimate, we define the nonconformity score and the induced trajectory tube, respectively, as
\begin{fullblock}
\begin{subequations}\label{eq:score-tube}
\begin{align}
s(\hat y_{0:T},y_{0:T})\ &:=\ \max_{0\le t\le T}\|\hat y_t-y_t\|_2,\\
\mathcal{C}_r(\hat y_{0:T})\ &:=\ \big\{\,y_{0:T}:\ s(\hat y_{0:T},y_{0:T})\le r\,\big\}.
\end{align}
\end{subequations}
\end{fullblock}
\begin{shortblock}
\subeqtighton
\begin{equation}\label{eq:score-tube}
s(\hat y_{0:T},y_{0:T}) \,:=\, \max_{0\le t\le T}\|\hat y_t-y_t\|_2,\qquad
\mathcal{C}_r(\hat y_{0:T}) \,:=\, \big\{\,y_{0:T}:\ s(\hat y_{0:T},y_{0:T})\le r\,\big\}.
\end{equation}
\eqend
\end{shortblock}
The threshold $r\ge 0$ is computed using CP to obtain probabilistic guarantees on the correctness of the set $\mathcal{C}_r(\hat y_{0:T})$. The main idea is simple: we use a set of $N$ held-out calibration trajectories $\{y_{0:T}^{(i)}\}_{i=1}^N$ generated by the non-interactive dynamics $y_{t+1}^{(i)}=f_Y(y_t^{(i)},\nu_t^{(i)})$. We then compute $r$ as the $(1-\alpha)$-quantile of the held-out nonconformity scores, denoted by $q_{1-\alpha}(\{s(\hat y_{0:T}, y_{0:T}^{(i)})\}_{i=1}^N\cup \{\infty\})$.\footnote{The split-conformal quantile $q_{1-\alpha}$ is computed from the set of nonconformity scores $\{s(\hat y_{0:T}, y_{0:T}^{(i)})\}_{i=1}^N\cup \{\infty\}$. Let $s_{(k)}$ be the $k$-th smallest nonconformity score (the $k$-th order statistic). Then, the quantile is $q_{1-\alpha} = s_{(k)}$ with $k = \lceil (N+1)(1-\alpha) \rceil$. }
Since the test trajectory $y_{0:T}$ and the held-out trajectories $\{y_{0:T}^{(i)}\}_{i=1}^N$ are exchangeable, we have that
\articleonly{
  \begin{align*}
\mathbb{P}_{N+1}\big\{\,y_{0:T}\in\mathcal{C}_{q_{1-\alpha}}(\hat y_{0:T})\,\big\}\ \ge\ 1-\alpha,
\end{align*}}
\ldconly{$\mathbb{P}_{N+1}\big\{\,y_{0:T}\in\mathcal{C}_{q_{1-\alpha}}(\hat y_{0:T})\,\big\}\ \ge\ 1-\alpha$,}
where now $\mathbb{P}_{N+1}\{\cdot\}$  is a product probability measure that captures the randomness in test $y_{0:T}$ and calibration data $\{y_{0:T}^{(i)}\}_{i=1}^N$, and as such can approximate the chance constraint in equation \eqref{eq:pf:chance}. This guarantee motivates the formulation of the following robust planning problem:
\begin{fullblock}
\begin{subequations}
\begin{align}
\min_{u_{0:T-1}} \quad & J(x_{0:T},u_{0:T-1})\\[-0.15em]
\text{s.t.}\quad & x_{t+1}=f_X(x_t,u_t),\quad t=0,\ldots,T-1,\\[-0.15em]
& H(x_{0:T},\zeta)\le 0\quad \forall\,\zeta\in \mathcal{C}_{q_{1-\alpha}}(\hat y_{0:T}),
\label{eq:robust-noninteractive}
\end{align}
\end{subequations}
\end{fullblock}
\begin{shortblock}
\subeqtighton
\begin{subequations}
\begin{align}
\min_{u_{0:T-1}} \quad & J(x_{0:T},u_{0:T-1})\\[-0.15em]
\text{s.t.}\quad & x_{t+1}=f_X(x_t,u_t),\quad t=0,\ldots,T-1,\\[0.15em]
& H(x_{0:T},\zeta)\le 0\quad \forall\,\zeta\in \mathcal{C}_{q_{1-\alpha}}(\hat y_{0:T}),
\label{eq:robust-noninteractive}
\end{align}
\end{subequations}
% \begin{subequations}
% \begin{align}
% \min_{u_{0:T-1}} \quad & J(x_{0:T},u_{0:T-1})\\[-0.15em]
% \text{s.t.}\quad & x_{t+1}=f_X(x_t,u_t),\ t=0,\ldots,T-1,\quad
% H(x_{0:T},\zeta)\le 0\ \ \forall\,\zeta\in \mathcal{C}_{q_{1-\alpha}}(\hat y_{0:T}),
% \label{eq:robust-noninteractive}
% \end{align}
% \end{subequations}
\eqend
\end{shortblock}
which ensures $\mathbb{P}_{N+1}\{H(x_{0:T},y_{0:T})\le 0\}\ge 1-\alpha$ in the non-interactive case whenever 
\eqref{eq:robust-noninteractive} is feasible \cite{lindemann2022safe}.
% \fullonly{\eqref{eq:robust-noninteractive} is feasible}\shortonly{the constraints in \eqref{eq:robust-noninteractive} are jointly feasible}
For instance, for the collision avoidance-type safety constraints $H(x_{0:T},y_{0:T})=\max_{0\le t\le T}\{\, b_t - c(x_t,y_t)\}$,
the constraint \eqref{eq:robust-noninteractive}
% \fullonly{the constraint \eqref{eq:robust-noninteractive}}\shortonly{the safety constraint in \eqref{eq:robust-noninteractive}} 
reduces to a step-wise tightening of the form $c(x_t,\hat y_t) \ge b_t +\ q_{1-\alpha}$ for all times $t=0,\ldots,T$, which can easily be implemented.
% \fullonly{\\}
% ==== Figure 2 anchor — BUILD 2 (L4DC extended) only ====
% L4DC ext works at this anchor (before Running Example);
% keep it here so the working build is not disturbed.
\ifarticleclass\else\iffullcontent
  \begin{wrapfigure}[16]{r}{0.5\textwidth}
    \centering
    \includegraphics[width=\linewidth]{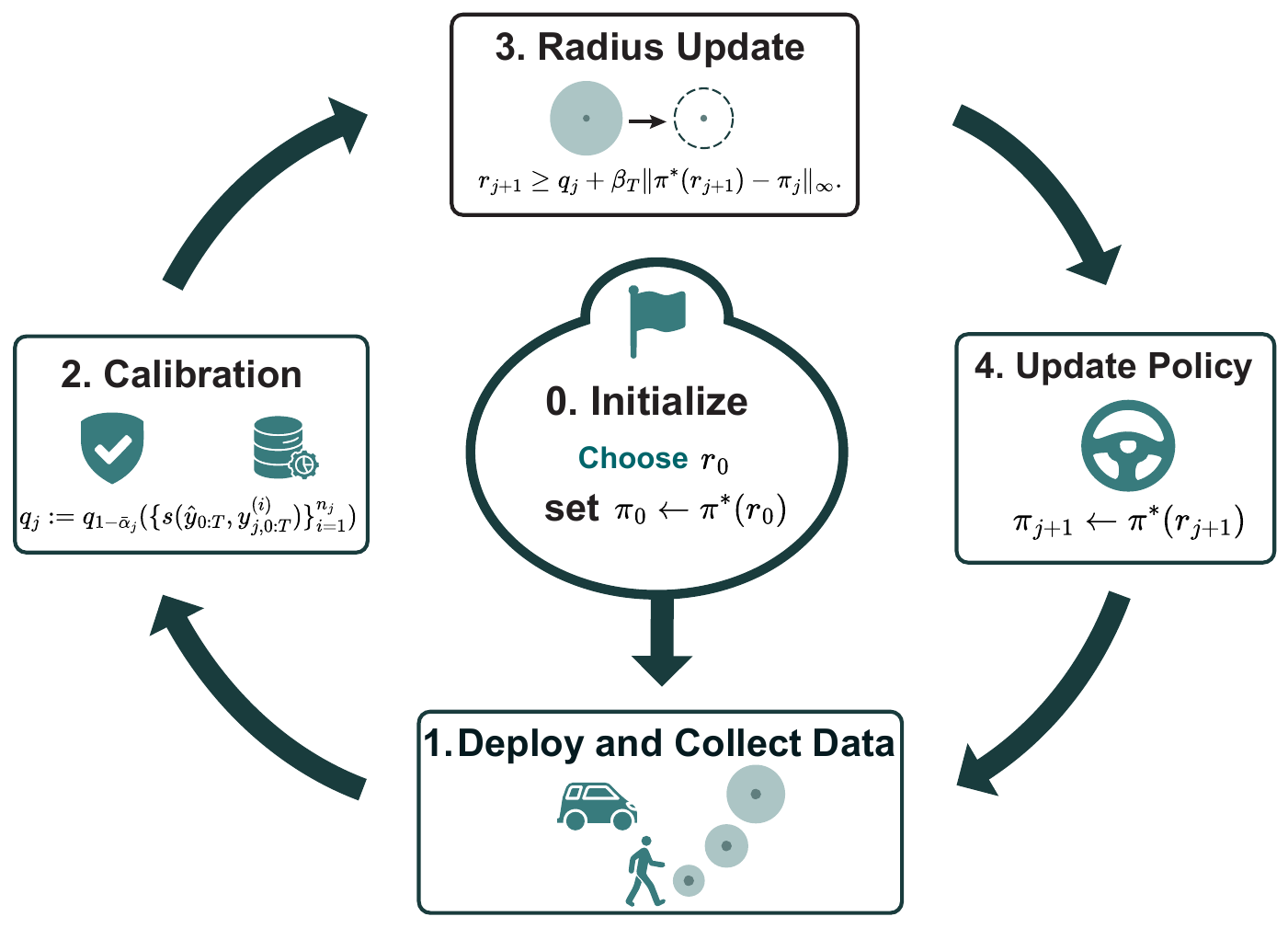}
    \caption{Our iterative algorithm: (0) initialize with a safe (yet conservative) radius $r_0$ and compute policy $\pi^\star(r_0)$; (1) deploy $\pi^\star(r_j)$ and collect data; (2) calibrate $q_j$; (3) compute new radius $r_{j+1}$ from $q_j$; (4) update policy to $\pi^\star(r_{j+1})$; (5) repeat.}
    \label{fig:overview}
  \end{wrapfigure}
\fi\fi

However, in the interactive case where the coupled agent dynamics are $y_{t+1}=f_Y(y_t,x_t,u_t,\nu_t)$ instead of $y_{t+1}=f_Y(y_t,\nu_t)$, a circular dependency -- or "chicken-and-egg" problem -- emerges. Changing the control inputs $u_t$ and ego-agent trajectory $x_t$ changes the distribution of $y_{0:T}$, violating the exchangeability assumption with the held-out trajectories $\{y_{0:T}^{(i)}\}_{i=1}^N$.
On the other hand, updating the held-out trajectories $\{y_{0:T}^{(i)}\}_{i=1}^N$ via $y_{t+1}^{(i)}=f_Y(y_t^{(i)},x_t,u_t,\nu_t^{(i)})$, changes the quantile $q_{1-\alpha}$ and thereby $u_t$ and $x_t$.
\articleonly{\\}
% ==== Figure 2 anchor — BUILDs 1 (arxiv) ====
\ifarticleclass
  % BUILD 1 — arxiv extended
  \begin{wrapfigure}[17]{r}{0.5\textwidth}
    \centering
    \includegraphics[width=\linewidth]{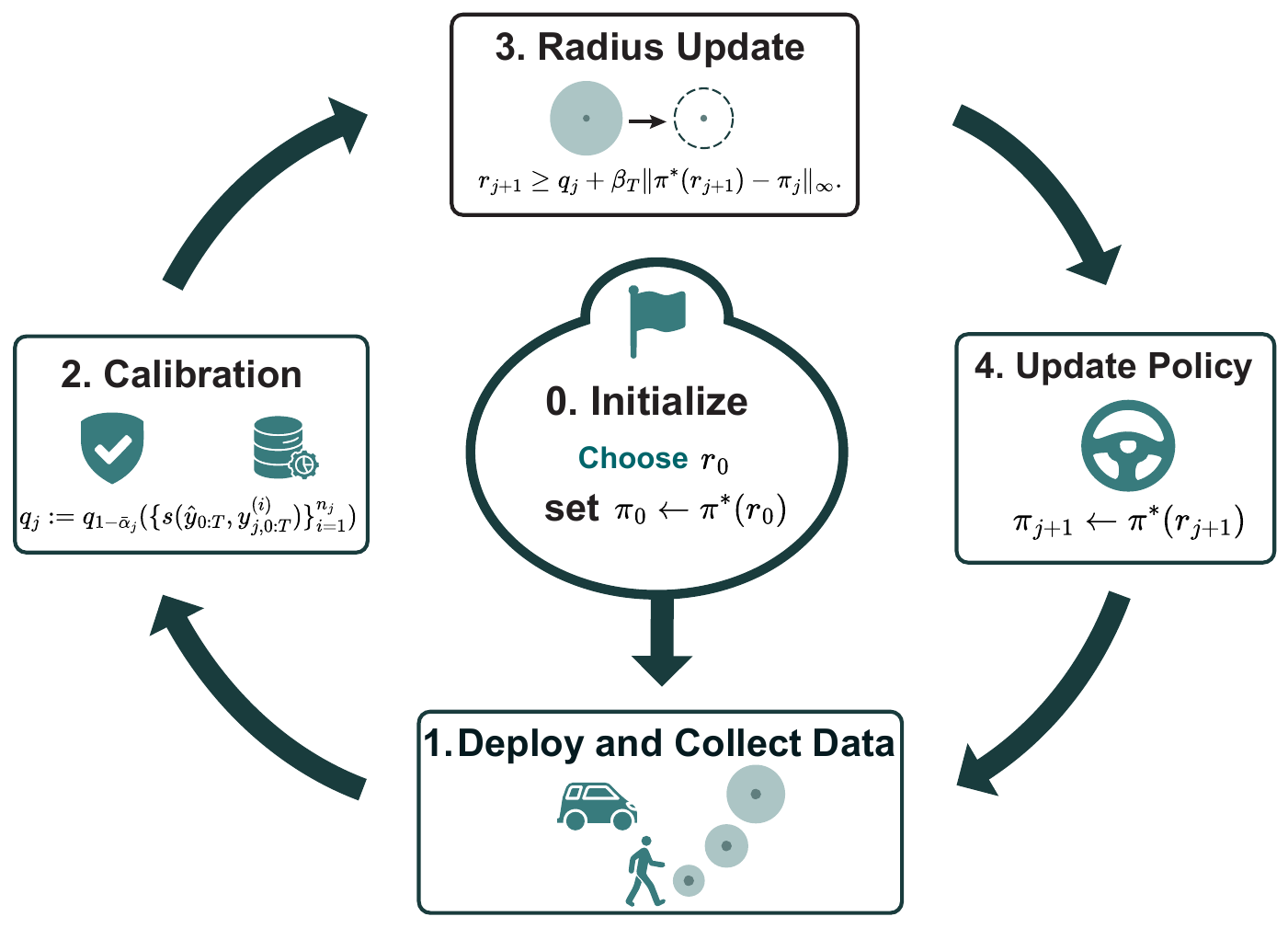}
    \caption{Our iterative algorithm: (0) initialize with a safe (yet conservative) radius $r_0$ and compute policy $\pi^\star(r_0)$; (1) deploy $\pi^\star(r_j)$ and collect data; (2) calibrate $q_j$; (3) compute new radius $r_{j+1}$ from $q_j$; (4) update policy to $\pi^\star(r_{j+1})$; (5) repeat.}
    \label{fig:overview}
  \end{wrapfigure}
\fi
To address this issue, we propose an episodic framework in which we iteratively compute new held-out trajectories while updating the control inputs.

\textbf{Running Example (cont.).}
In our running  example this non-interactive approach would first predict a {nominal path} $\hat y_{0:T}$ for the pedestrian (e.g., walking straight on the sidewalk). % ==== Figure 2 anchor — BUILD 3 (L4DC submission) only ====
% L4DC ext works at this anchor (before Running Example);
% keep it here so the working build is not disturbed.
\ifarticleclass\else\ifshortcontent
  % \begin{wrapfigure}[15]{r}{0.5\textwidth}
    \begin{wrapfigure}[14]{r}{0.4\textwidth}
    \centering
    \includegraphics[width=\linewidth]{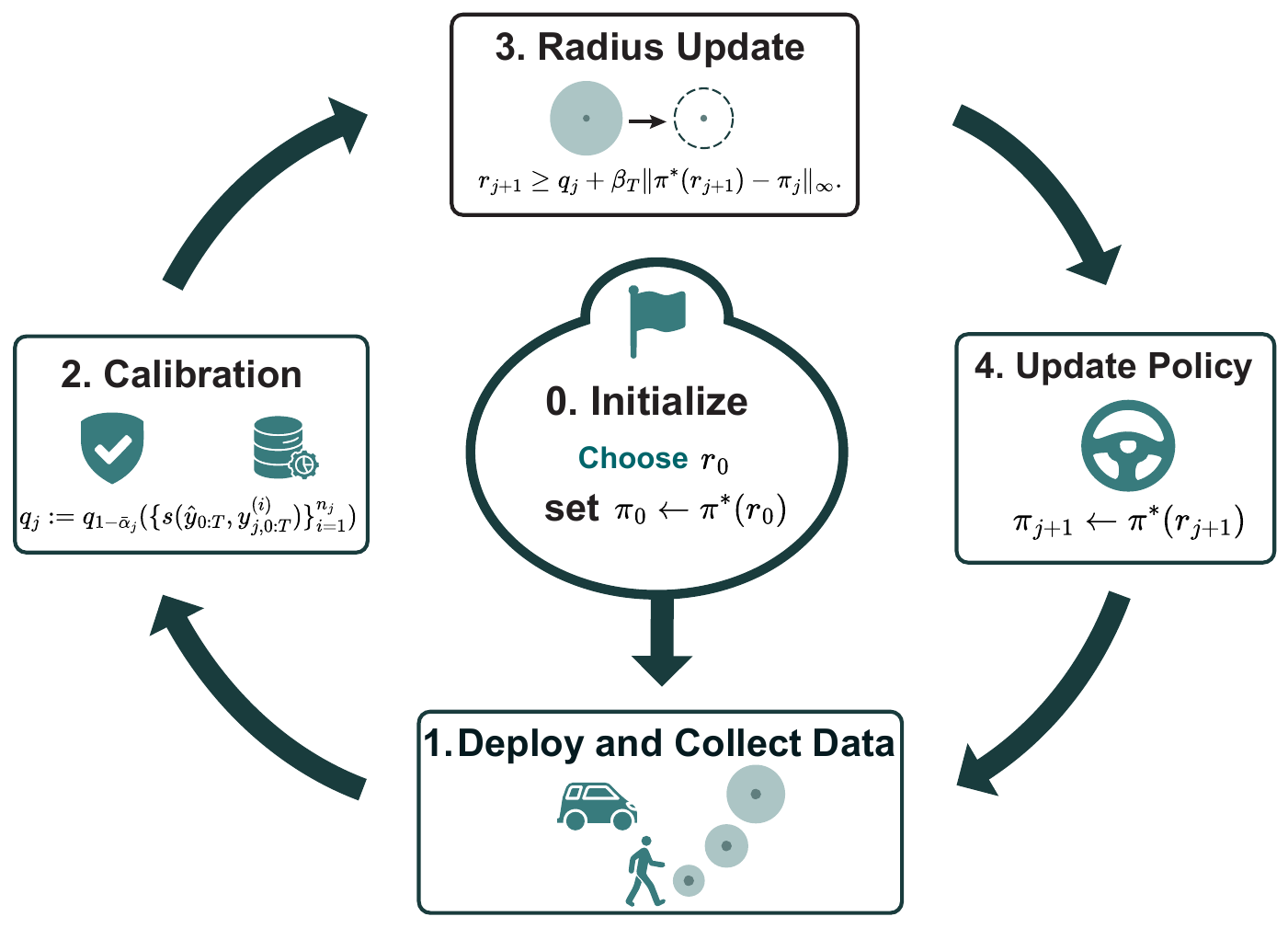}
    \caption{Our iterative algorithm: (0) initialize with a safe (yet conservative) radius $r_0$ and compute policy $\pi^\star(r_0)$; (1) deploy $\pi^\star(r_j)$ and collect data; (2) calibrate $q_j$; (3) compute new radius $r_{j+1}$ from $q_j$; (4) update policy to $\pi^\star(r_{j+1})$; (5) repeat.}
    \label{fig:overview}
  \end{wrapfigure}
   % \vspace{-1cm}
\fi\fi
The {trajectory tube} $\mathcal{C}_r$ represents a tube of radius $r = q_{1-\alpha}$ around this nominal path that guarantees the {actual} pedestrian path $y_{0:T}$ will be inside 
this tube with $1-\alpha$ probability, {assuming the car's actions do not influence the pedestrian}. The {"chicken-and-egg" problem} arises in the interactive case: the car's plan $u_t$ depends on the size of the pedestrian's uncertainty set (the radius $q_{1-\alpha}$), but the pedestrian's actual behavior, which determines the size of that tube, depends on the car's plan $u_t$.
\articleonly{\vspace{.2cm}}
% \subsection{Episodic Problem Formulation and Design Goals}
% \label{subsec:episodic-formulation-goals}
\begin{fullblock}
\subsection{Episodic Problem Formulation and Design Goals}
\label{subsec:episodic-formulation-goals}
\end{fullblock}
\begin{shortblock}
\subsectionshort{Episodic Problem Formulation and Design Goals}
\label{subsec:episodic-formulation-goals}
\end{shortblock}
To address the aforementioned circular dependency, we reframe the problem into an iterative, episodic framework which we explain below and summarize in Figure \ref{fig:overview}. Planning proceeds in episodes $j=0,1,2,\ldots$. We use a fixed nominal predictor $\hat y_{0:T} = (\hat y_0, \ldots, \hat y_T)$ to anchor the geometry of the uncertainty. Updating the predictor episodically is possible with minimal modifications, but omitted here for simplicity. At each episode $j$, we solve the following robust optimization problem, parameterized by an uncertainty radius $r_j \ge 0$:
\subeqtighton
\begin{equation}
\boxed{
\begin{aligned}
\textbf{P}[j;\,r_j]\qquad
\min_{u_{j,0:T-1}}\quad & J(x_{j,0:T},u_{j,0:T-1})\\[-0.2em]
\text{s.t.}\quad & x_{j,t+1}=f_X(x_{j,t},u_{j,t}),\quad t=0,\ldots,T-1,\\[-0.2em]
& H(x_{j,0:T},\zeta)\le 0\quad \forall\,\zeta\in \mathcal{C}_{r_j}(\hat y_{0:T}).
\end{aligned}}
\label{eq:episodic}
\end{equation}
\subeqend
Solving \eqref{eq:episodic} yields the policy $\pi_j:=u_{j,0:T-1}$ and the resulting optimal cost $J_j := J(x_{j,0:T},u_{j,0:T-1})$. This policy is then executed on the physical system $n_j$ times, which produces $n_j$ i.i.d. rollouts of the environment's trajectories $\left\{ y_{j, 0:T}^{(i)} \right\}_{i=1}^{n_j} \sim \mathcal{D}(\pi_j)$ under the ego-agent's policy, as governed by \eqref{eq:pf:dyn} and denoted by $\mathcal{D}(\pi_j)$.  
The design of this iterative process is driven by three high-level objectives:
% \iffullcontent
% \begin{align}
% &\textbf{Per-episode safety:}
% && \mathbb{P}\left\{\,H(x_{j,0:T},y_{j,0:T})\le 0\,\right\} \ \ge\ 1-\alpha
% \quad \text{for all } j,
% \label{eq:goal-safety}\\[-0.25em]
% &\textbf{Performance improvement:}
% && J_{j+1} \ \le\ J_j \ -\ \Delta_j
% \quad\text{with}\ \Delta_j\ge 0,
% \label{eq:goal-performance}\\[-0.25em]
% &\textbf{Stability of uncertainty:}
% && r_{j+1} \ \approx\ r_j \quad \text{eventually.}
% \label{eq:goal-convergence}
% \end{align}
% \else
% \refstepcounter{equation}\label{eq:goal-safety}%
% \refstepcounter{equation}\label{eq:goal-performance}%
% \refstepcounter{equation}\label{eq:goal-convergence}%
% \emph{per-episode safety}, $\mathbb{P}\{H(x_{j,0:T},y_{j,0:T})\le 0\}\ge 1-\alpha$ for all $j$; \emph{performance improvement}, $J_{j+1}\le J_j-\Delta_j$ with $\Delta_j\ge 0$; and \emph{stability of uncertainty}, $r_{j+1}\approx r_j$ eventually.
% \fi
\subeqtighton
\begin{align}
&\textbf{Per-episode safety:}
&& \mathbb{P}\left\{\,H(x_{j,0:T},y_{j,0:T})\le 0\,\right\} \ \ge\ 1-\alpha
\quad \text{for all } j,
\label{eq:goal-safety}\\[-0.25em]
&\textbf{Performance improvement:}
&& J_{j+1} \ \le\ J_j \ -\ \Delta_j
\quad\text{with}\ \Delta_j\ge 0,
\label{eq:goal-performance}\\[-0.25em]
&\textbf{Stability of uncertainty:}
&& r_{j+1} \ \approx\ r_j \quad \text{eventually.}
\label{eq:goal-convergence}
\end{align}
\subeqend
The key challenge is to determine how to update $r_j$ to $r_{j+1}$. 
This update must leverage the new data $\{ y_{j, 0:T}^{(i)} \}$ to reduce conservatism, while also guaranteeing that the {next} policy $\pi_{j+1}$ will be safe.
Because $r_j$ is implicitly {tied} to the policy $\pi_j$ through interaction, the policy-induced distribution shift from $j$ to $j+1$ must be accounted for to maintain \eqref{eq:goal-safety}. We make the following assumption.
\begin{fullblock}
\begin{assumption}\label{ass1}
At iteration $j=0$, we know $r_0$ s.t.\ $\pi_0$ satisfies $\mathbb{P}\{H(x_{0,0:T},y_{0,0:T})\leq 0\}\geq 1-\alpha$.
\end{assumption}
\end{fullblock}
\begin{shortblock}
\begin{assumptionshort}\label{ass1}
At iteration $j=0$, we know $r_0$ s.t.\ $\pi_0$ satisfies $\mathbb{P}\{H(x_{0,0:T},y_{0,0:T})\leq 0\}\geq 1-\alpha$.
\end{assumptionshort}
\end{shortblock}

 This assumption is natural to our problem formulation. In practice, we can typically find this constant by selecting a sufficiently large $r_0$ that is valid for any permissible ego agent policy.
 
\textbf{Running Example (cont.). } In episode $j=0$, assume that we have selected a valid {uncertainty radius $r_0$} for the ego agent's policy $\pi_0$, e.g., that the pedestrian will deviate at most $r_0=2$ meters from its predicted path under $\pi_0$ with probability of at least $1-\alpha$. We then observe realized pedestrian trajectories $\{ y_{0, 0:T}^{(i)} \}$ under $\pi_0$,  where we may now notice that the pedestrian may deviate less than $r_0$ meters from the predicted path. The key challenge is hence to iteratively update the next radii $r_{1}, r_2, \hdots$ and compute the next policies $\pi_{1},\pi_2$ while adhering to the objectives in \eqref{eq:goal-safety}-\eqref{eq:goal-convergence}.
% \section{An Iterative Planning Framework for Interactive Environments}
% \label{sec:iterative-framework}
\begin{fullblock}
\section{An Iterative Planning Framework for Interactive Environments}
\label{sec:iterative-framework}
\end{fullblock}
\begin{shortblock}
\sectionshort{An Iterative Planning Framework for Interactive Environments}
\label{sec:iterative-framework}
\end{shortblock}

In our interactive setting, the exchangeability assumption underlying CP is violated at every policy update.  When we solve \eqref{eq:episodic} and update the policy from $\pi_{j}$ to $\pi_{j+1}$, the realized trajectory $y_{j+1,0:T}$ is drawn from a different distribution. This re-introduces the "chicken-and-egg" problem.
% \subsection{Adversarial Conformal Prediction for Policy-Induced Shifts}
% \label{sec:acp_policy_shift}
\begin{fullblock}
\subsection{Adversarial Conformal Prediction for Policy-Induced Shifts}
\label{sec:acp_policy_shift}
\end{fullblock}
\begin{shortblock}
\subsectionshort{Adversarial Conformal Prediction for Policy-Induced Shifts}
\label{sec:acp_policy_shift}
\end{shortblock}
To carry safety certificates across such policy-induced shifts, our approach uses adversarial conformal prediction (ACP), which provides guarantees under bounded perturbations \cite{GendlerEtAl2022ICLR}. Let $z\sim\mathcal{D}$ be a random variable and  $\hat z$ be an estimate of $z$. Given a nonconformity score $s(\hat z, z)$ and a \textcolor{black}{perturbation radius} $\rho \ge 0$, ACP provides coverage for $s\big(\hat z, z + \Delta\big)$ where the random variable $z$ may be perturbed by any $\Delta$ with $\|\Delta\| \le \rho$. The next result follows  by combining ACP with conditional CP \cite{vovk2012conditional,duchi2025sample}, as summarized in \cite[Lemma 2]{lindemann2024formal}.

\begin{fullblock}
\begin{lemma}[\cite{GendlerEtAl2022ICLR} and \cite{vovk2012conditional,duchi2025sample}]\label{lemma_ccp}
    Let $z,z^{(1)},\hdots, z^{(N)}{\sim}\mathcal{D}$ be exchangeable random variables, $\rho>0$ be a perturbation radius, $\alpha,\delta\in (0,1)$ be failure probabilities, and
    \articleonly{\begin{align*}
        r := q_{1-\bar{\alpha}}\big(\{s(\hat z, z^{(i)})\}_{i=1}^N\big) := s_{(\lceil N(1-\bar{\alpha})\rceil)}
    \end{align*}}
    \ldconly{$r := q_{1-\bar{\alpha}}\big(\{s(\hat z, z^{(i)})\}_{i=1}^N\big) := s_{(\lceil N(1-\bar{\alpha})\rceil)}$}
    be the empirical $(1-\bar{\alpha})$-quantile of the unperturbed nonconformity scores $\{s(\hat z, z^{(i)})\}_{i=1}^N$, where $s_{(1)}\le\cdots\le s_{(N)}$ denote the associated order statistics and $\bar{\alpha}:= \alpha-\sqrt{\ln(1/\delta)/2N}$.\footnote{The empirical quantile $s_{(\lceil N(1-\bar{\alpha})\rceil)}$ here differs slightly from the empirical quantile $s_{(\lceil (N+1)(1-\alpha)\rceil)}$ from Section~\ref{sec:pf}.} Assume there exists a constant $M\ge 0$ such that for all $z$ it holds that
    $
    \sup_{\|\Delta\|\le \rho} s(\hat z, z+\Delta)\ \le\ s(\hat z, z)+M
    $.\footnote{If the nonconformity score $s$ is Lipschitz continuous in its second argument with constant $L$, then $M= L\rho$.} Then, we have that
\begin{align}
    \mathbb{P}_N\big\{\,\mathbb{P}\big\{\, \forall\,\Delta\ \text{with }\|\Delta\|\le\rho:\ z+\Delta \in \mathcal{C}^{\text{adv}}_{r}(\hat z;\rho)\big\} \ge 1-\alpha\big\}\ge 1-\delta.
    \label{eq:rob_pred}
\end{align}
where
$
\mathcal{C}^{\text{adv}}_{r}(\hat z;\rho)
:= \big\{ z : s(\hat z, z) \le r+ M \big\}
$ is a robustified prediction set.
\end{lemma}
\end{fullblock}

\begin{shortblock}
\begin{lemmashort}[\cite{GendlerEtAl2022ICLR} and \cite{vovk2012conditional,duchi2025sample}]\label{lemma_ccp}
    Let $z,z^{(1)},\hdots, z^{(N)}{\sim}\mathcal{D}$ be exchangeable random variables, $\rho>0$ be a perturbation radius, $\alpha,\delta\in (0,1)$ be failure probabilities, and
    $r := q_{1-\bar{\alpha}}\big(\{s(\hat z, z^{(i)})\}_{i=1}^N\big) := s_{(\lceil N(1-\bar{\alpha})\rceil)}$
    be the empirical $(1-\bar{\alpha})$-quantile of the unperturbed nonconformity scores $\{s(\hat z, z^{(i)})\}_{i=1}^N$, where $s_{(1)}\le\cdots\le s_{(N)}$ denote the associated order statistics and $\bar{\alpha}:= \alpha-\sqrt{\ln(1/\delta)/2N}$.\footnote{The empirical quantile $s_{(\lceil N(1-\bar{\alpha})\rceil)}$ here differs slightly from the empirical quantile $s_{(\lceil (N+1)(1-\alpha)\rceil)}$ from Section~\ref{sec:pf}.} Assume there exists a constant $M\ge 0$ such that for all $z$ it holds that
    $
    \sup_{\|\Delta\|\le \rho} s(\hat z, z+\Delta)\ \le\ s(\hat z, z)+M
    $.\footnote{If the nonconformity score $s$ is Lipschitz continuous in its second argument with constant $L$, then $M= L\rho$.} Then, we have that
\subeqgentleon
\begin{align}
    \mathbb{P}_N\big\{\,\mathbb{P}\big\{\, \forall\,\Delta\ \text{with }\|\Delta\|\le\rho:\ z+\Delta \in \mathcal{C}^{\text{adv}}_{r}(\hat z;\rho)\big\} \ge 1-\alpha\big\}\ge 1-\delta.
    \label{eq:rob_pred}
\end{align}
\eqend
where
$
\mathcal{C}^{\text{adv}}_{r}(\hat z;\rho)
:= \big\{ z : s(\hat z, z) \le r+ M \big\}
$ is a robustified prediction set.
\end{lemmashort}
\end{shortblock}

We note that the guarantees in \eqref{eq:rob_pred} are conditional in the sense that $\mathbb{P}_N\{\cdot\}$ is a product probability measure that captures randomness in the calibration data $\{z^{(i)}\}_{i=1}^N$, so that the inner probability measure $\mathbb{P}\{\cdot\}$ captures randomness over test data $z$ that holds with probability no less than $1-\delta$.
% \begin{remark}
% \textcolor{blue}{In this lemma, $q_{1-\bar{\alpha}}(\{s(\hat z,z^{(i)})\}_{i=1}^N)$ is the left empirical $(1-\bar{\alpha})$-quantile of the $N$ calibration scores: writing $s_{(1)} \le \cdots \le s_{(N)}$ for the sorted values of $\{s(\hat z,z^{(i)})\}_{i=1}^N$, we set $q_{1-\bar{\alpha}}=s_{(\lceil N(1-\bar{\alpha})\rceil)}$. This choice matches the sample-conditional high-probability guarantee used here, and should not be confused with the classical split conformal $(N+1)$-corrected threshold $s_{(\lceil (N+1)(1-\alpha)\rceil)}$ that is used to obtain exact finite-sample marginal coverage.}
% \end{remark}
To address the "chicken-and-egg" problem, we will apply robustification via ACP iteratively. At each episode $j$, we perform a recalibration step using the new data $\{ y_{j, 0:T}^{(i)} \}_{i=1}^{n_j} \sim \mathcal{D}(\pi_j)$ to compute the empirical quantile
\fulleq{
  q_j := q_{1-\bar{\alpha}_j}\!\big(\{s(\hat{y}_{0:T},\, y_{j,0:T}^{(i)})\}_{i=1}^{n_j}\big),
  \qquad
  \bar{\alpha}_j := \alpha-\sqrt{\tfrac{\ln(1/\delta_j)}{2n_j}},
  \label{eq:episodic_quantile}
}{$q_j := q_{1-\bar{\alpha}_j}\!\big(\{s(\hat{y}_{0:T},\, y_{j,0:T}^{(i)})\}_{i=1}^{n_j}\big)$ where $\delta_j\in(0,1)$ and $\bar{\alpha}_j := \alpha-\sqrt{\ln(1/\delta_j)/(2n_j)}$.} By Lemma~\ref{lemma_ccp}, this empirical quantile $q_j$ satisfies
$
\mathbb{P}_{n_j}\big\{\mathbb{P}\big\{
y_{j,0:T} \in \mathcal{C}_{q_j}(\hat y_{0:T})
\big\}
\ge 1-\alpha\big\}\ge 1-\delta_j,
$
where $y_{j,0:T} \sim \mathcal{D}(\pi_j)$ and $\mathcal{C}_{q_j}(\hat y_{0:T}) := \{ y_{0:T} : s(\hat y_{0:T}, y_{0:T}) \le q_j \}$. Thus, $q_j$ is a valid safety radius for policy $\pi_j$. To maintain safety for the next policy $\pi_{j+1}$, we must now account for policy-induced distribution shifts caused by the change from $\pi_j$ to $\pi_{j+1}$. We interpret and treat the policy update from $\pi_j$ to $\pi_{j+1}$ as the adversarial perturbation term $\rho_j$.

\textbf{High-level intuition.} Let us provide some intuition before we present the iterative algorithm and a detailed analysis in Sections~\ref{sec:iterative_planning} and \ref{sec:analysis}, respectively. Under Assumption~\ref{ass:lipschitz_dynamics}, we can bound the change in the uncontrollable agents trajectory by a constant $\beta_T\ge 0$ such that
\articleonly{\begin{align*}
\max_{0 \le t \le T}
\big\| y_t(\pi_{j+1}) - y_t(\pi_j) \big\|_2
\;\le\;
\beta_T \, \|\pi_{j+1} - \pi_j\|_\infty,
\end{align*}}
\ldconly{$\max_{0 \le t \le T}
\big\| y_t(\pi_{j+1}) - y_t(\pi_j) \big\|_2
\;\le\;
\beta_T \, \|\pi_{j+1} - \pi_j\|_\infty$,}
where $\|\pi_{j+1} - \pi_j\|_\infty := \max_{0 \le t \le T-1} \|u_{j+1,t} - u_{j,t}\|_2$.\footnote{We use the notation of $y_t(\pi_{j})$ instead of $y_{j,t}$ to highlight that the trajectory $y_t$ is a function of $\pi_{j}$.} We formalize the existence of $\beta_T$ next and provide the proof in Appendix~\ref{app:proof_beta_t}.

\begin{fullblock}
\begin{lemma}[Episode Coupling Sensitivity]
\label{prop:coupling_sensitivity}
Let the system in \eqref{eq:pf:dyn} be given and Assumption~\ref{ass:lipschitz_dynamics} hold. Then there exists a constant $\beta_{T}\ge0$ which depends on the horizon $T$ and the Lipschitz constants $L_{Xx}, L_{Xu}, L_{Yy}, L_{Yx},L_{Yu}$ such that for two policies $\pi = \{u_t\}_{t=1}^{T-1}$ and $\pi' = \{u'_t\}_{t=1}^{T-1}$:
\articleonly{\begin{align*}
\|y_{0:T}(\pi') - y_{0:T}(\pi)\|_{\infty} \le \beta_{T}\|\pi' - \pi\|_{\infty}.
\end{align*}}
\ldconly{$\|y_{0:T}(\pi') - y_{0:T}(\pi)\|_{\infty} \le \beta_{T}\|\pi' - \pi\|_{\infty}$.}
\end{lemma}
\end{fullblock}
\begin{shortblock}
\begin{lemmashort}[Episode Coupling Sensitivity]
\label{prop:coupling_sensitivity}
Let the system in \eqref{eq:pf:dyn} be given and Assumption~\ref{ass:lipschitz_dynamics} hold. Then there exists a constant $\beta_{T}\ge0$ which depends on the horizon $T$ and the Lipschitz constants $L_{Xx}, L_{Xu}, L_{Yy}, L_{Yx},L_{Yu}$ such that for two policies $\pi = \{u_t\}_{t=1}^{T-1}$ and $\pi' = \{u'_t\}_{t=1}^{T-1}$:
\articleonly{\begin{align*}
\|y_{0:T}(\pi') - y_{0:T}(\pi)\|_{\infty} \le \beta_{T}\|\pi' - \pi\|_{\infty}.
\end{align*}}
\ldconly{$\|y_{0:T}(\pi') - y_{0:T}(\pi)\|_{\infty} \le \beta_{T}\|\pi' - \pi\|_{\infty}$.}
\end{lemmashort}
\end{shortblock}
Next, note that our nonconformity score in \eqref{eq:score-tube} is $1$-Lipschitz continuous. Therefore, we have
\articleonly{\begin{align*}
|s(\hat y_{0:T},y_{j,0:T})-s(\hat y_{0:T},y_{j+1,0:T})|\le \beta_T \, \|\pi_{j+1} - \pi_j\|_\infty,
\end{align*}}
\ldconly{$|s(\hat y_{0:T},y_{j,0:T})-s(\hat y_{0:T},y_{j+1,0:T})|\le \beta_T \, \|\pi_{j+1} - \pi_j\|_\infty$,}
i.e., a policy change can increase the nonconformity score by at most
$
M_{j+1} := \beta_T \, \|\pi_{j+1} - \pi_j\|_\infty.
$

Using Lemma~\ref{lemma_ccp} and combining this term with the empirical quantile $q_j$ yields a high-probability upper bound on the required radius for the next policy $\pi_{j+1}$ that is of the form:
\subeqtighton
\begin{equation}
\mathbb{P}_{n_j}
\big\{\,\mathbb{P}
\big\{\,
s(\hat y_{0:T}, y_{j+1,0:T})
\;\le\;
q_j + M_{j+1}
\big\}
\;\ge\;
1-\alpha\big\} \ge 1-\delta_j.
\label{eq:episodic_transfer}
\end{equation}
\subeqend
This motivates our safe radius update rule $r_{j+1} := q_j + M_{j+1}$, which is constructed to be a high-probability upper bound on the true, unknown $(1-\alpha)$-quantile of the next distribution $\mathcal{D}(\pi_{j+1})$. Before deploying $\pi_{j+1}$, we solve the next planning problem $P[j{+}1; r_{j+1}]$ using $r_{j+1}$ as the required safety tube radius. By construction, any policy that is feasible for $P[j{+}1; r_{j+1}]$ is guaranteed, with high probability, to remain safe against all environment trajectories in that inflated tube.

% \subsection{The Iterative Planning Algorithm}
% \label{sec:iterative_planning}
\begin{fullblock}
\subsection{The Iterative Planning Algorithm}
\label{sec:iterative_planning}
\end{fullblock}
\begin{shortblock}
\subsectionshort{The Iterative Planning Algorithm}
\label{sec:iterative_planning}
\end{shortblock}
We follow the previously described ``\emph{recalibrate each episode, then transfer}'' approach. Algorithmically, the central challenge is to compute the next radius $r_{j+1}$, which must ensure safety for the next policy $\pi_{j+1}$ before we have even computed $\pi_{j+1}$. Formally, to compute $\pi_{j+1}$ we need to know $r_{j+1}$ for which we have to know $M_{j+1}$, which in turn depends on $\pi_{j+1}$, creating an implicit problem. We first define this problem and then present two ways to solve it: a computationally expensive but exact "implicit solver" and a computationally cheap "explicit solution" (used in our main analysis).

\textbf{Implicit Safety.}
From \eqref{eq:episodic_transfer}, $r_{j+1}$ must cover ${q}_j$ and the \textcolor{black}{robustification term $M_{j+1}$} such that $r_{j+1} \ge {q}_j + M_{j+1}$. However, the \textcolor{black}{robustification term} $M_{j+1}$ depends on the policy $\pi_{j+1}$ that we are trying to find, since $\pi_{j+1} = \pi^\star(r_{j+1})$. This gives the true, implicit safety requirement:

\begin{equation}
\label{eq:implicit_problem}
r_{j+1} \ge {q}_j + \beta_T \norm{\pi^\star(r_{j+1}) - \pi_j}_{\infty}.
\end{equation}
We are now looking for the smallest $r_{j+1} \in \calR:=[r_{\min}, r_{\max}]$ that satisfies this inequality. We next present details for both solutions, which are summarized in Algorithm~1 in Appendix~\ref{sec:iterative_algorithm}.

\textbf{Approach 1: The Implicit Solver.} We enforce the implicit inequality in \eqref{eq:implicit_problem} directly, treating $\pi^\star(r)$ as a black box. At episode $j$, the quantities $q_j$, $\pi_j$, and $\beta_T$ are fixed, so computing $r_{j+1}$ reduces to a scalar one-dimensional program: we search for the smallest $r\in[q_j,r_{\max}]$ satisfying \eqref{eq:implicit_problem}, which we solve by a bracketed line search and bisection or by a constrained solver. The  procedure is detailed in Appendix~\ref{sec:implicit_solver}. The theoretical guarantees of this solver use the global sensitivity bound $\beta_T$ from Lemma~\ref{prop:coupling_sensitivity}, but one may also substitute a data-driven estimate $\widehat\beta_T$ (see Appendix~\ref{app:practical_instantiations}).

\textbf{Approach 2: The Tractable Explicit Solver.} We here solve the implicit inequality \eqref{eq:implicit_problem} analytically. To get such an analytical bound, we assume that there exists a Lipschitz constant $L_U\ge 0$ that bounds how much the optimal policy $\pi^*(r)$ changes in response to a change in the radius $r\in\calR$.
\begin{fullblock}
\begin{assumption}\label{ass:planner_sensitivity}
\iffullcontent
   There exists a Lipschitz constant $L_U>0$ such that
   \begin{align*}
       \norm{\pi^*(r) - \pi^*(r')}_{\infty} \le L_U \abs{r - r'} \qquad \;\forall r,r'\in\calR
   \end{align*}
\else
   There exists a constant $L_U>0$ s.t.\ $\|\pi^*(r)-\pi^*(r')\|_\infty \le L_U|r-r'|, \;\forall r,r'\in\calR$.
\fi
\end{assumption}
\end{fullblock}
\begin{shortblock}
\begin{assumptionshort}\label{ass:planner_sensitivity}
There exists a constant $L_U>0$ s.t.\ $\|\pi^*(r)-\pi^*(r')\|_\infty \le L_U|r-r'|, \;\forall r,r'\in\calR$.
\end{assumptionshort}
\end{shortblock}
The Lipschitz continuity property in Assumption~\ref{ass:planner_sensitivity} holds for many optimization-based planners. We provide a detailed analysis in Appendix~\ref{app:proof_planner_sensitivity} for the common case in which the optimization problem $P[j; r]$ is feasible, convex, and has a unique, regular optimizer. From here, we now get
\subeqgentleon
\begin{align*}
\beta_T \norm{\pi^\star(r_{j+1}) - \pi_j}_{\infty} = \beta_T \norm{\pi^\star(r_{j+1}) - \pi^\star(r_j)}_{\infty} \le \beta_T L_U \abs{r_{j+1} - r_j} = \kappa \abs{r_{j+1} - r_j},
\end{align*}
\subeqend
where $\kappa := \beta_T L_U$ is a closed-loop gain that will help us analyze properties of our iterative planner. Instead of the implicit inequality \eqref{eq:implicit_problem}, we can use this upper bound to get the sufficient inequality:
\subeqtighton
% \subeqgentleon
\begin{equation}\label{eq:r_update}
r_{j+1} \ge {q}_j + \kappa \abs{r_{j+1} - r_j}.
\end{equation}
\subeqend
The inequality \eqref{eq:r_update} is a simple scalar inequality for $r_{j+1}$. As we will show in Section~\ref{sec:analysis}, this inequality has a unique, minimal (least conservative) solution, which is given by a closed-form expression. This is the algorithm we use in our analysis in the next section. \articleonly{It is computationally efficient (one-shot calculation) and is guaranteed to be safe (though it may be more conservative).} Lastly, we remark that we can estimate $L_U$ similarly to $\beta_T$ (recall Appendix~\ref{app:practical_instantiations}).
\begin{fullblock}
\subsection{Main Guarantees: Safety, Algorithmic Stability, Convergence, and Performance}
\label{sec:analysis}
\end{fullblock}
\begin{shortblock}
\subsectionshort{Main Guarantees: Safety, Algorithmic Stability, Convergence, and Performance}
\label{sec:analysis}
\end{shortblock}

Throughout this section, let the system in \eqref{eq:pf:dyn} be given and Assumptions~\ref{ass:lipschitz_dynamics} and \ref{ass1} hold. We start our analysis by stating our episodic safety guarantees which are proven in Appendix~\ref{app:thm-pe-proof}.

\begin{fullblock}
\begin{theorem}[Per-Episode Safety Guarantee]\label{thm:per_episode_safety}
Let $r_{j+1}$ follow the implicit safety requirement \eqref{eq:implicit_problem}. Then
\iffullcontent
\begin{equation}\label{eq:thm-pe:main-prob}
\mathbb{P}_{n_j}\!\left\{\,\mathbb{P}\!\left\{\,s(\hat y_{0:T}, y_{j+1,0:T}) \le r_{j+1}\,\right\} \ge 1-\alpha\,\right\} \;\ge\; 1-\delta_j.
\end{equation}
Furthermore, if $\mathbf{P}[j{+}1;\,r_{j+1}]$ is feasible, then
\begin{equation}\label{eq:thm-pe:chance-H}
\mathbb{P}_{n_j}\!\left\{\,\mathbb{P}\!\left\{\,H(x_{j+1,0:T},y_{j+1,0:T})\le 0\,\right\} \ge 1-\alpha \right\} \;\ge\; 1-\delta_j.
\end{equation}
\else
$\mathbb{P}_{n_j}\{\mathbb{P}\{s(\hat y_{0:T}, y_{j+1,0:T}) \le r_{j+1}\} \ge 1-\alpha\} \ge 1-\delta_j$. Furthermore, if $\mathbf{P}[j{+}1;\,r_{j+1}]$ is feasible, then $\mathbb{P}_{n_j}\{\mathbb{P}\{H(x_{j+1,0:T},y_{j+1,0:T})\le 0\} \ge 1-\alpha\} \ge 1-\delta_j$.
\fi
\end{theorem}
\end{fullblock}

\begin{shortblock}
\begin{theoremshort}[Per-Episode Safety Guarantee]\label{thm:per_episode_safety}
Let $r_{j+1}$ follow the implicit safety requirement \eqref{eq:implicit_problem}. Then $\mathbb{P}_{n_j}\{\mathbb{P}\{s(\hat y_{0:T}, y_{j+1,0:T}) \le r_{j+1}\} \ge 1-\alpha\} \ge 1-\delta_j$. Furthermore, if $\mathbf{P}[j{+}1;\,r_{j+1}]$ is feasible, then $\mathbb{P}_{n_j}\{\mathbb{P}\{H(x_{j+1,0:T},y_{j+1,0:T})\le 0\} \ge 1-\alpha\} \ge 1-\delta_j$.
\end{theoremshort}
\end{shortblock}
% \begin{theorem}[Per-Episode Safety Guarantee]\label{thm:per_episode_safety}
% Let the system in \eqref{eq:pf:dyn} be given, Assumptions~\ref{ass:lipschitz_dynamics} and \ref{ass1} hold, and $r_{j+1}$ follow the implicit safety requirement \eqref{eq:implicit_problem}. Then
% \iffullcontent
% \begin{equation}\label{eq:thm-pe:main-prob}
% \mathbb{P}_{n_j}\!\left\{\,\mathbb{P}\!\left\{\,s(\hat y_{0:T}, y_{j+1,0:T}) \le r_{j+1}\,\right\} \ge 1-\alpha\,\right\} \;\ge\; 1-\delta.
% \end{equation}
% Furthermore, if $\mathbf{P}[j{+}1;\,r_{j+1}]$ is feasible, then
% \begin{equation}\label{eq:thm-pe:chance-H}
% \mathbb{P}_{n_j}\!\left\{\,\mathbb{P}\!\left\{\,H(x_{j+1,0:T},y_{j+1,0:T})\le 0\,\right\} \ge 1-\alpha \right\} \;\ge\; 1-\delta.
% \end{equation}
% \else
% $\mathbb{P}_{n_j}\{\mathbb{P}\{s(\hat y_{0:T}, y_{j+1,0:T}) \le r_{j+1}\} \ge 1-\alpha\} \ge 1-\delta$. Furthermore, if $\mathbf{P}[j{+}1;\,r_{j+1}]$ is feasible, then $\mathbb{P}_{n_j}\{\mathbb{P}\{H(x_{j+1,0:T},y_{j+1,0:T})\le 0\} \ge 1-\alpha\} \ge 1-\delta$.
% \fi
% \end{theorem}
If $H$ is Lipschitz continuous with Lipschitz constant $L_H$, we note that a sufficient condition for the feasibility of $\mathbf{P}[j{+}1;\,r_{j+1}]$ is that $H(x_{j+1,0:T},\hat y_{0:T}) \le -\,L_H\,r_{j+1}$, see Appendix~\ref{app:sufficiency}.

In the remainder of this section, let additionally Assumption \ref{ass:planner_sensitivity} hold. We proceed our analysis using the explicit solver for the implicit inequality \eqref{eq:implicit_problem} which resulted in the tractable inequality in \eqref{eq:r_update}. We show its solution next and present the proof in Appendix~\ref{app:proof_explicit_update}.
% \begin{lemma}[Safe Explicit Radius Update]
% \label{thm:explicit_update_rule}
% Let $\kappa < 1$. The minimal (i.e., least conservative) radius ${r}_{j+1}$ that satisfies the tractable inequality \eqref{eq:r_update} is given by the following closed-form solution:
% \begin{equation}\label{r_update_explicit}
% {r}_{j+1} =
% \begin{cases}
%       \dfrac{{q}_j + \kappa r_j}{1+\kappa} & \text{if } {q}_j \le r_j \quad (\text{Shrinkage Branch}) \\
%       \dfrac{{q}_j - \kappa r_j}{1-\kappa} & \text{if } {q}_j > r_j \quad (\text{Expansion Branch})
% \end{cases}
% \end{equation}
% \end{lemma}
\begin{fullblock}
\begin{lemma}[Safe Explicit Radius Update]
\label{thm:explicit_update_rule}
Let $\kappa < 1$. The minimal (i.e., least conservative) radius ${r}_{j+1}$ that satisfies the tractable inequality \eqref{eq:r_update} is given by the following closed-form solution:
\begin{equation}\label{r_update_explicit}
{r}_{j+1} =
\begin{cases}
      \dfrac{{q}_j + \kappa r_j}{1+\kappa} & \text{if } {q}_j \le r_j \quad (\text{Shrinkage Branch}) \\
      \dfrac{{q}_j - \kappa r_j}{1-\kappa} & \text{if } {q}_j > r_j \quad (\text{Expansion Branch})
\end{cases}
\end{equation}
\end{lemma}
\end{fullblock}
\begin{shortblock}
\eqtighton
\begin{lemmashort}[Safe Explicit Radius Update]
\label{thm:explicit_update_rule}
Let $\kappa < 1$. The minimal (i.e., least conservative) radius ${r}_{j+1}$ that satisfies the tractable inequality \eqref{eq:r_update} is given by the following closed-form solution:
\begin{equation}\label{r_update_explicit}
{r}_{j+1} =
\begin{cases}
      \dfrac{{q}_j + \kappa r_j}{1+\kappa} & \text{if } {q}_j \le r_j \quad (\text{Shrinkage Branch}) \\
      \dfrac{{q}_j - \kappa r_j}{1-\kappa} & \text{if } {q}_j > r_j \quad (\text{Expansion Branch})
\end{cases}
\end{equation}
\end{lemmashort}
\eqend
\end{shortblock}
% \begin{lemma}[Safe Explicit Radius Update]
% \label{thm:explicit_update_rule}
% Let $\kappa < 1$. The minimal (i.e., least conservative) radius ${r}_{j+1}$ that satisfies the tractable inequality \eqref{eq:r_update} is given by the following closed-form solution:
% \begin{equation}\label{r_update_explicit}
% {r}_{j+1} =
% \begin{cases}
%       \dfrac{{q}_j + \kappa r_j}{1+\kappa} & \text{if } {q}_j \le r_j \quad (\text{Shrinkage Branch}) \\
%       \dfrac{{q}_j - \kappa r_j}{1-\kappa} & \text{if } {q}_j > r_j \quad (\text{Expansion Branch})
% \end{cases}
% \end{equation}
% \end{lemma}
If we limit $r_{j+1}$ to an admissible interval $\calR = [r_{\min}, r_{\max}]$ (e.g., as needed in Appendix~\ref{app:proof_planner_sensitivity} to derive $L_U$), then we can project $r_{j+1}$ onto $\calR$ via the operator $\Pi_{\calR}({r}_{j+1}) = \min(r_{\max}, \max(r_{\min}, {r}_{j+1}))$. The rule \eqref{r_update_explicit} leads to the following guarantees, for which we provide the proof in Appendix~\ref{app:proof_stability_shrinkage}.

\begin{fullblock}
\begin{theorem}[Episode-to-Episode Stability and Shrinkage]
\label{thm:stability_shrinkage}
Let $\kappa < 1$ and the radius $r_{j+1}$ satisfy \eqref{r_update_explicit}. Then, the per-episode change in the radius $r_{j+1}$ is bounded by:
\articleonly{\begin{align*}
\abs{{r}_{j+1} - r_j} \le \frac{1}{1-\kappa} \abs{{q}_j - r_j}.
\end{align*}}
\ldconly{$\abs{{r}_{j+1} - r_j} \le \frac{1}{1-\kappa} \abs{{q}_j - r_j}$.}
Furthermore, if ${q}_j < r_j$, then it holds that ${r}_{j+1} < r_j$.
\end{theorem}
\end{fullblock}

\begin{shortblock}
\begin{theoremshort}[Episode-to-Episode Stability and Shrinkage]
\label{thm:stability_shrinkage}
Let $\kappa < 1$ and the radius $r_{j+1}$ satisfy \eqref{r_update_explicit}. Then, the per-episode change in the radius $r_{j+1}$ is bounded by:
\articleonly{\begin{align*}
\abs{{r}_{j+1} - r_j} \le \frac{1}{1-\kappa} \abs{{q}_j - r_j}.
\end{align*}}
\ldconly{$\abs{{r}_{j+1} - r_j} \le \frac{1}{1-\kappa} \abs{{q}_j - r_j}$.}
Furthermore, if ${q}_j < r_j$, then it holds that ${r}_{j+1} < r_j$.
\end{theoremshort}
\end{shortblock}

Finally, we analyze conditions under which our algorithms converge. We study the true quantile of the nonconformity score $s(\hat y_{0:T},y_{0:T})$ for $y_{0:T}\sim\mathcal D(\pi^\star(r))$. For simplicity, let
\articleonly{\begin{align*}
F_r(t):=\mathbb{P}\{s(\hat y_{0:T}, y_{j+1,0:T})\le t\}
\end{align*}}
\ldconly{$F_r(t):=\mathbb{P}\{s(\hat y_{0:T}, y_{j+1,0:T})\le t\}$}
denote the cumulative distribution function (CDF) of the score and define the true quantile as
\articleonly{\begin{align*}
Q_{1-\alpha}(\pi^\star(r)) \;\coloneqq\; \inf\{\,t\in\mathbb R:\ F_s(t)\ge 1-\alpha\,\}.
\end{align*}}
\ldconly{$Q_{1-\alpha}(\pi^\star(r)) \;\coloneqq\; \inf\{\,t\in\mathbb R:\ F_s(t)\ge 1-\alpha\,\}$.}
We now provide conditions for the convergence of $r_j$ and  $q_j$  and provide the proof in Appendix~\ref{app:thm8-proof}.

\begin{fullblock}
\begin{theorem}[Convergence with Explicit Error Summation]\label{thm:convergence-main}
Let
$
T(r):=Q_{1-\alpha}\!\big(\pi^\star(r)\big)
$
be the $(1-\alpha)$ true quantile of the nonconformity score. Assume that $T(r)$ is Lipschitz continuous so that
\begin{align*}
|T(r)-T(r')|\le \kappa\,|r-r'|\quad\text{for all }r,r'\in \mathbb{R}_{\ge 0},\qquad \kappa=\beta_TL_U\in(0,1).
\end{align*}
Suppose further that there is a fixed point $r^\star\in \mathbb{R}_{\ge 0}$ with $T(r^\star)=r^\star$. At episode $j$, let the radius $r_{j+1}$ satisfy \eqref{r_update_explicit}. Define
$
e_j=|r_j-r^\star|
$ and $
\eta_j=q_j-T(r_j).
$
Then, the following hold:
\pitem{thm:conv:P1}{Finite-horizon error bound.}
At each episode $j$, we are guaranteed that
\begin{equation}\label{eq:explicit-sum-main}
e_{j+1}\le \gamma_\kappa^{\,j+1}e_0
+B_\kappa\sum_{m=0}^{j}\gamma_\kappa^{\,j-m}\,|\eta_m|,\
\gamma_\kappa=\frac{2\kappa}{1-\kappa},\ 
B_\kappa=\frac{1}{1-\kappa},
\ \ 
\gamma_\kappa<1\iff \kappa<\tfrac13.
\end{equation}
\pitem{thm:conv:P2}{Closed form under a uniform perturbation bound.}
If $|\eta_m|\le C$ for all $m\in\{0,\dots,j\}$, then
\begin{equation}\label{eq:closed-form-const-eta-main}
e_{j+1}\le \gamma_\kappa^{\,j+1}e_0+B_\kappa C\,\frac{1-\gamma_\kappa^{\,j+1}}{1-\gamma_\kappa},
\
\text{and if }\kappa<\tfrac13:\
\limsup_{j\to\infty}e_j\le \frac{B_\kappa}{1-\gamma_\kappa}\,C=\frac{1}{1-3\kappa}\,C.
\end{equation}
\pitem{thm:conv:P3}{High-probability control of $\eta_j$.}
If the CDF $F_{r_j}(s)$ is differentiable and has density $f_{r_j}(s)$ no less than $f_\star>0$ for all $s$ in a sufficiently large neighborhood of its $(1-\alpha)$ quantile $Q_{1-\alpha}\!\big(\pi^\star(r_j)\big)$, then we have
\begin{equation}\label{eq:eta-bound-main}
\mathbb{P}_{n_j}\!\left\{|\eta_j|
\;\le\;
\frac{|\alpha-\bar\alpha_j|}{f_\star}
\;+\;
\frac{1}{f_\star}\sqrt{\frac{\ln(2/\delta_j)}{2 n_j}}\right\}
\;\ge\; 1-\delta_j.
\end{equation}
Furthermore, we get
\begin{equation}\label{eq:explicit-sum-split-main}
\mathbb{P}_{\sum_{m=0}^{j}n_m}\!\Big\{e_{j+1}\le \gamma_\kappa^{\,j+1}e_0+B_\kappa\sum_{m=0}^{j}\gamma_\kappa^{\,j-m}A_m\Big\}\ge 1-\sum_{m=0}^{j}\delta_m,
\end{equation}
where $A_m:=\tfrac{|\alpha-\bar\alpha_m|}{f_\star}+\tfrac{1}{f_\star}\sqrt{\tfrac{\ln(2/\delta_m)}{2n_m}}$.
\pitem{thm:conv:P4}{Asymptotics.}
If $n_j\to\infty$ and $\bar\alpha_j\to\alpha$, then $\mathbb{P}_{n_j}\{|\eta_j|\to 0\}\ge 1-\delta_j$, and for $\kappa<\tfrac13$,
\begin{equation}\label{eq:e-asymptotic-zero}
\mathbb{P}_{\sum_{m=0}^{\infty}n_m}\{e_j\to 0\}\;\ge\;1-\sum_{m=0}^{\infty}\delta_m.
\end{equation}
 If instead $\bar\alpha_j\equiv\bar\alpha<\alpha$ and $n_j\to\infty$, then $\mathbb{P}_{n_j}\{|\eta_j|\to (\alpha-\bar\alpha)/f_\star\}\ge 1-\delta_j$, and
\begin{equation}\label{eq:e-asymptotic-limsup}
\mathbb{P}_{\sum_{m=0}^{\infty}n_m}\!\Big\{\limsup_{j\to\infty}e_j\le \tfrac{1}{1-3\kappa}\cdot\tfrac{\alpha-\bar\alpha}{f_\star}\Big\}\;\ge\;1-\sum_{m=0}^{\infty}\delta_m.
\end{equation}
\end{theorem}
\end{fullblock}

\begin{shortblock}
\begin{theoremshort}[Convergence with Explicit Error Summation]\label{thm:convergence-main}
Let
$
T(r):=Q_{1-\alpha}\!\big(\pi^\star(r)\big)
$
be the $(1-\alpha)$ true quantile of the nonconformity score. Assume that $T(r)$ is Lipschitz continuous so that
\subeqtighton
\begin{align*}
|T(r)-T(r')|\le \kappa\,|r-r'|\quad\text{for all }r,r'\in \mathbb{R}_{\ge 0},\qquad \kappa=\beta_TL_U\in(0,1).
\end{align*}
\eqend
Suppose further that there is a fixed point $r^\star\in \mathbb{R}_{\ge 0}$ with $T(r^\star)=r^\star$. At episode $j$, let the radius $r_{j+1}$ satisfy \eqref{r_update_explicit}. Define
$
e_j=|r_j-r^\star|
$ and $
\eta_j=q_j-T(r_j).
$
Then, the following hold:
\pitemshort{thm:conv:P1}{Finite-horizon error bound.}
At each episode $j$, we are guaranteed that
\subeqtighton
\begin{equation}\label{eq:explicit-sum-main}
e_{j+1}\le \gamma_\kappa^{\,j+1}e_0
+B_\kappa\sum_{m=0}^{j}\gamma_\kappa^{\,j-m}\,|\eta_m|,\
\gamma_\kappa=\frac{2\kappa}{1-\kappa},\ 
B_\kappa=\frac{1}{1-\kappa},
\ \ 
\gamma_\kappa<1\iff \kappa<\tfrac13.
\end{equation}

\pitemshort{thm:conv:P2}{Closed form under a uniform perturbation bound.}
If $|\eta_m|\le C$ for all $m\in\{0,\dots,j\}$, then
\begin{equation}\label{eq:closed-form-const-eta-main}
e_{j+1}\le \gamma_\kappa^{\,j+1}e_0+B_\kappa C\,\frac{1-\gamma_\kappa^{\,j+1}}{1-\gamma_\kappa},
\
\text{and if }\kappa<\tfrac13:\
\limsup_{j\to\infty}e_j\le \frac{B_\kappa}{1-\gamma_\kappa}\,C=\frac{1}{1-3\kappa}\,C.
\end{equation}

\pitemshort{thm:conv:P3}{High-probability control of $\eta_j$.}
If the CDF $F_{r_j}(s)$ is differentiable and has density $f_{r_j}(s)$ no less than $f_\star>0$ for all $s$ in a sufficiently large neighborhood of its $(1-\alpha)$ quantile $Q_{1-\alpha}\!\big(\pi^\star(r_j)\big)$, then  we have $\mathbb{P}_{n_j}\{|\eta_j|\le A_j\}\ge 1-\delta_j$ and $\mathbb{P}_{\sum_{m=0}^{j}n_m}\big\{e_{j+1}\le \gamma_\kappa^{\,j+1}e_0+B_\kappa\sum_{m=0}^{j}\gamma_\kappa^{\,j-m}A_m\big\}\ge 1-\sum_{m=0}^{j}\delta_m$ where $A_m:=\tfrac{|\alpha-\bar\alpha_m|}{f_\star}+\tfrac{1}{f_\star}\sqrt{\tfrac{\ln(2/\delta_m)}{2n_m}}$.
\pitemshort{thm:conv:P4}{Asymptotics.}
If $n_j\to\infty$ and $\bar\alpha_j\to\alpha$, then $\mathbb{P}_{n_j}\{|\eta_j|\to 0\}\ge 1-\delta_j$, and for $\kappa<\tfrac13$, $\mathbb{P}_{\sum_{m=0}^{\infty}n_m}\{e_j\to 0\}\ge 1-\sum_{m=0}^{\infty}\delta_m$.
If instead $\bar\alpha_j\equiv\bar\alpha<\alpha$ and $n_j\to\infty$, then $\mathbb{P}_{n_j}\{|\eta_j|\to (\alpha-\bar\alpha)/f_\star\}\ge 1-\delta_j$, and $\mathbb{P}_{\sum_{m=0}^{\infty}n_m}\{\limsup_{j\to\infty}e_j\le \tfrac{1}{1-3\kappa}\cdot\tfrac{\alpha-\bar\alpha}{f_\star}\}\ge 1-\sum_{m=0}^{\infty}\delta_m$.
\eqend
\end{theoremshort}
\end{shortblock}

Lastly, we study study convergence of the performance as measured via the function $J(\cdot)$. We note that $\mathbf{P}[j;r]$ in~\eqref{eq:episodic} depends on $j$ only through $r$. Hence, let
\fulleq{
  V(r):=\min_{\pi}\,J\big(x_{0:T}(\pi),\pi\big)\ \ \text{s.t.}\ \ H\big(x_{0:T}(\pi),\zeta\big)\le 0\ \ \forall\,\zeta\in\mathcal{C}_r(\hat y_{0:T}),\qquad J_j:=V(r_j),
  \label{eq:V_definition}
}{$V(r):=\min_{\pi}J(x_{0:T}(\pi),\pi)$ s.t.\ $H(x_{0:T}(\pi),\zeta)\le 0\ \forall\,\zeta\in\mathcal{C}_r(\hat y_{0:T})$ and $J_j:=V(r_j)$.}
The next result is proven in Appendix~\ref{app:proof-performance-convergence}.
\begin{fullblock}
\vspace{0\baselineskip}
\begin{theorem}[Performance convergence and improvement]
\label{thm:performance_convergence}
Let $V(r)$ have Lipschitz constant $L_V$ on $\calR$ so that $|V(r)-V(r')|\le L_V|r-r'|$ for all $r,r'\in \calR$. Then, it holds that $|J_{j+1}-J^\star|\le L_V\,e_{j+1}$ where $J^\star:=V(r^\star)$ and $r^\star$, $e_j$, $\eta_j$, $A_j$ are defined in Theorem~\ref{thm:convergence-main}. Furthermore, the following hold:
\pitem{thm:perf:P1}{Improvement over the initial policy.}
If $r_0>r^\star$ and $V(r)$ is strictly increasing on $[r^\star,r_0]$, then $J_j<J_0$ for $r_j\in[r^\star,r_0)$. If additionally $V(r)-V(r')\ge m_V(r-r')$ for some $m_V>0$ and all $r\ge r'$ in $[r^\star,r_0]$, then every $r_j\le r_0$ satisfies
\fulleq{
  J_0-J_j\;\ge\;m_V\,(r_0-r^\star-e_j).
  \label{eq:perf-improvement-initial-main}
}{$J_0-J_j\ge m_V(r_0-r^\star-e_j)$.}
\pitem{thm:perf:P2}{One-step improvement.}
If $r_j\ge r^\star$ and $|\eta_j|<(1-\kappa)(r_j-r^\star)$, then $q_j<r_j$ and $J_{j+1}\le J_j$. If the sufficient condition $A_j<(1-\kappa)(r_j-r^\star)$ holds for all $j\in\{J,\ldots,K\}$, then
\fulleq{
  \mathbb{P}_{\sum_{j=J}^{K}n_j}\!\left\{J_{j+1}\le J_j\ \forall\,j=J,\ldots,K\right\}
  \;\ge\;1-\sum_{j=J}^{K}\delta_j.
  \label{eq:perf-tail-monotone-hp}
}{$\mathbb{P}_{\sum_{j=J}^{K}n_j}\!\big\{J_{j+1}\le J_j\ , \forall\,j=J,\ldots,K\big\}\ge 1-\sum_{j=J}^{K}\delta_j$.}
\pitem{thm:perf:P3}{Asymptotics.}
Under the assumption in Theorem~\ref{thm:convergence-main}(\ref{thm:conv:P4}),
\fulleq{
  \mathbb{P}_{\sum_{m=0}^{\infty}n_m}\!\{J_j\to J^\star\}\;\ge\;1-\sum_{m=0}^{\infty}\delta_m.
  \label{eq:perf-asymptotic-J-conv}
}{$\mathbb{P}_{\sum_{m=0}^{\infty}n_m}\{J_j\to J^\star\}\ge 1-\sum_{m=0}^{\infty}\delta_m$.}
If $r_0>r^\star$, then there exists $j_0$ such that
\fulleq{
  \mathbb{P}_{\sum_{m=0}^{\infty}n_m}\!\{J_j\le J_0\ , \forall\,j\ge j_0\}\;\ge\;1-\sum_{m=0}^{\infty}\delta_m.
  \label{eq:perf-asymptotic-j0}
}{$\mathbb{P}_{\sum_{m=0}^{\infty}n_m}\{J_j\le J_0\ \forall\,j\ge j_0\}\ge 1-\sum_{m=0}^{\infty}\delta_m$.}
\end{theorem}
\end{fullblock}
\begin{shortblock}
\subeqtighton
\vspace{0\baselineskip}
\begin{theoremshort}[Performance convergence]\label{thm:performance_convergence}
Let $V(r)$ have Lipschitz constant $L_V$ on $\calR$ so that $|V(r)-V(r')|\le L_V|r-r'|$ for all $r,r'\in \calR$. Then, we have that $|J_{j+1}-J^\star|\le L_V\,e_{j+1}$ where $J^\star:=V(r^\star)$ and $r^\star$, $e_j$, $\eta_j$, $A_j$ are defined in Theorem~\ref{thm:convergence-main}. Furthermore, the following hold:
\pitemshort{thm:perf:P1}{Improvement over the initial policy.}
If $r_0>r^\star$ and $V(r)$ is strictly increasing on $[r^\star,r_0]$, then $J_j<J_0$ for any $r_j\in[r^\star,r_0)$. If there exists $m_V>0$ such that $V(r)-V(r')\ge m_V(r-r')$ for all $r,r'\in[r^\star,r_0]$, then every $r_j\le r_0$ satisfies $J_0-J_j\ge m_V(r_0-r^\star-e_j)$.
\pitemshort{thm:perf:P2}{One-step improvement.}
If $r_j\ge r^\star$ and $|\eta_j|<(1-\kappa)(r_j-r^\star)$, then $q_j<r_j$ and $J_{j+1}\le J_j$. If $A_j<(1-\kappa)(r_j-r^\star)$ for all $j\in\{J,\ldots,K\}$, then $\mathbb{P}_{\sum_{j=J}^{K}n_j}\!\big\{J_{j+1}\le J_j\ ,\forall\,j\big\}\ge 1-\sum_{j=J}^{K}\delta_j$.
\pitemshort{thm:perf:P3}{Asymptotics.}
Under the assumptions in Theorem~\ref{thm:convergence-main}(\ref{thm:conv:P4}), $\mathbb{P}_{\sum_{m=0}^{\infty}n_m}\{J_j\to J^\star\}\ge 1-\sum_{m=0}^{\infty}\delta_m$. If $r_0>r^\star$, then there exists $j_0$ such that $\mathbb{P}_{\sum_{m=0}^{\infty}n_m}\{J_j\le J_0\ , \forall\,j\ge j_0\}\ge 1-\sum_{m=0}^{\infty}\delta_m$.
\end{theoremshort}
\subeqend
\end{shortblock}

% \section{Case Studies}
% \label{sec:case_studies}
\begin{fullblock}
\section{Case Studies}
\label{sec:case_studies}
\end{fullblock}
\begin{shortblock}
\vspace{-.1cm}
\sectionshort{Case Studies}
\label{sec:case_studies}
\end{shortblock}

We consider two case studies: a two-dimensional car--pedestrian interaction (Section~\ref{sec:case-pedestrian}) and a multi-quadcopter navigation task in QuadSwarm~\cite{quadswarm} \iffullcontent(Section~\ref{sec:quadswarm})\else(in Appendix~\ref{app:quadswarm})\fi. Both instantiate the  program $\mathbf{P}[j;r_j]$ in~\eqref{eq:episodic} with  $\alpha:=0.1$, $\delta_j:=0.05$, $r_0=r_{\max}$, and $n_j:=1000$.
We report the radius $r_j$, the cost $J_j:=V(r_j)$, and the empirical coverages for $\mathbb{P}\{y_{j,0:T}\in\mathcal{C}_{r_j}(\hat y_{0:T})\}$ and $\mathbb{P}\{H(x_{j,0:T},y_{j,0:T})\le 0\}$ at target $1-\alpha=0.9$ to show Theorem~\ref{thm:per_episode_safety}.
We compare five rules that differ in (i) whether they recalibrate $q_j$ each episode, (ii) whether they transfer the certificate through the policy update via~\eqref{eq:implicit_problem} or~\eqref{r_update_explicit}, and (iii) how the transfer is resolved:
\emph{Robust CP -- implicit (ours)} enforces~\eqref{eq:implicit_problem} directly (see Appendix~\ref{sec:implicit_solver});
\emph{Robust CP -- explicit (ours)} uses~\eqref{r_update_explicit} with fixed $\kappa=\beta_T L_U$;
\emph{Naive CP} recalibrates but drops the transfer, setting $r_{j+1}=q_j$ and ignoring the perturbation $M_{j+1}$ in~\eqref{eq:episodic_transfer};
\emph{One-time CP} fixes $q_j\equiv q_{\mathrm{base}}$ from the first episode but keeps the transfer;
\emph{Split CP} computes $q_{\mathrm{base}}$ from non-interactive rollouts and sets $r_j\equiv q_{\mathrm{base}}$.
% \subsection{Two-Dimensional Car--Pedestrian Interaction}
% \label{sec:case-pedestrian}
\begin{fullblock}
\subsection{Two-Dimensional Car--Pedestrian Interaction}
\label{sec:case-pedestrian}
\end{fullblock}

\begin{shortblock}
\textbf{4.1. Two-Dimensional Car--Pedestrian Interaction}
\label{sec:case-pedestrian}
\end{shortblock}
The ego state $x_t\in\mathbb{R}^2$ follows  single-integrator dynamics $x_{t+1}=x_t+\Delta u_t$ with $\Delta:=0.1$ and $\|u_t\|_\infty\le 5$. The pedestrian state $y_t\in\mathbb{R}^2$ follows  interactive dynamics $y_{t+1}=y_t+\Delta(v_0+\phi(\|r_t\|_2)\,e_t)+w_t$ with relative position $r_t:=y_t-x_t$, direction $e_t:=r_t/\|r_t\|_2$, nominal velocity $v_0:=[-0.5,0]^\top$, noise $w_t\sim\mathcal{N}(0,\Delta\sigma^2 I)$ for $\sigma:=0.05$, and repulsion  $\phi(s):=v_{\max}\ell_c^2/(s^2+\ell_c^2)$ with $v_{\max}:=1$ and $\ell_c:=1$. The predictor
$\hat y_{t+1}=\hat y_t+\Delta v_0$ is intentionally misspecified: it omits $\phi(\cdot)$. The cost is $J(x,u):=\|x_T-x_{\mathrm{goal}}\|_2^2+10^{-3}\sum_{t=0}^{T-1}\|u_t\|_2^2$ with $x_{\mathrm{goal}}:=[6,0.5]^\top$ and $T:=50$; the safety function is $H(x,y):=\max_{t}(d_{\min}-\|x_t-y_t\|_2)$ with $d_{\min}:=0.8$. By the triangle inequality, the constraint $H(x,\zeta)\le 0\ \forall\,\zeta\in\mathcal{C}_{r_j}(\hat y_{0:T})$ reduces to $\|x_{j,t}-\hat y_t\|_2\ge d_{\min}+r_j$ for all $t=0,\ldots,T$. %which we solve in MATLAB via \texttt{fmincon}, warm-started from $\pi_{j-1}$.

\textbf{Observations.} Figure~\ref{fig:pedestrian-results} shows that under both \emph{Robust CP} variants the radius $r_j$ contracts from $r_0=r_{\max}$ and stabilizes near $r_j\approx 0.78$, the planner cost $J_j$ decreases from $J_0$ and stabilizes, and both tube and safety coverages remain at the target $1-\alpha$, in agreement with Theorem~\ref{thm:per_episode_safety} and Theorem~\ref{thm:performance_convergence}. Both solvers converge to the same steady-state radius and cost, with the implicit solver yielding a marginally larger radius during the transient ($r\approx 0.80$ vs.\ $r\approx 0.78$ at $j=3$, see Figure~\ref{fig:pedestrian-trajectories}). \emph{Naive CP} sets $r_{j+1}=q_j$, omitting $M_{j+1}$ in~\eqref{eq:episodic_transfer}: tube coverage drops to near zero in the early episodes before recovering once the radius stabilizes, precisely the tube coverage failure that~\eqref{eq:episodic_transfer} is designed to prevent. \emph{One-time CP} and \emph{Split CP} hold $q_j$ fixed and therefore cannot track $\mathcal{D}(\pi_j)$; their tube coverage collapses to near zero as the stale $q_{\mathrm{base}}$ does not account for the policy-induced distribution shift per~\eqref{eq:episodic_transfer}, and while their smaller radii ($r_j\approx 0.54$) yield lower planner costs, these lack valid coverage guarantees. Only the two \emph{Robust CP} variants simultaneously maintain both coverages and reduce $J_j$ from the conservative initialization, matching the shrinkage branch of~\eqref{r_update_explicit} under Theorem~\ref{thm:performance_convergence}. Figure~\ref{fig:pedestrian-trajectories} visualizes the trajectories and tube contraction at episodes $j=0,3,19$.
\begin{figure}[!t]
    \centering
    \captionsetup{font=small,skip=4pt}
    \setlength{\tabcolsep}{1pt}
    \begin{tabular}{@{}cccc@{}}
        \includegraphics[width=0.25\linewidth]{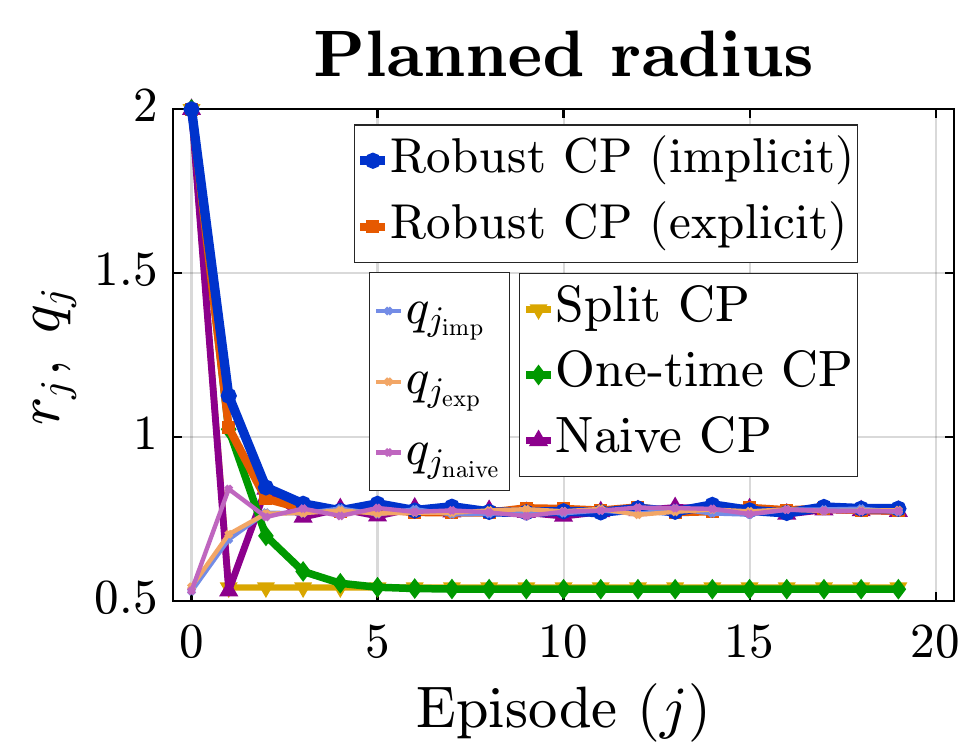} &
        \includegraphics[width=0.25\linewidth]{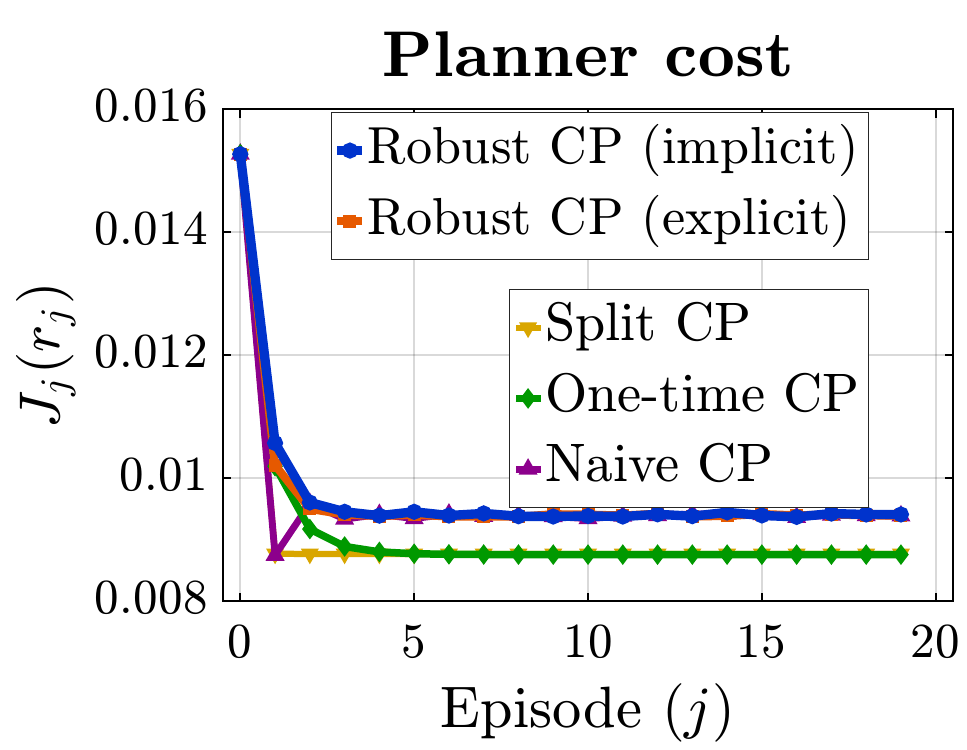} &
        \includegraphics[width=0.25\linewidth]
        {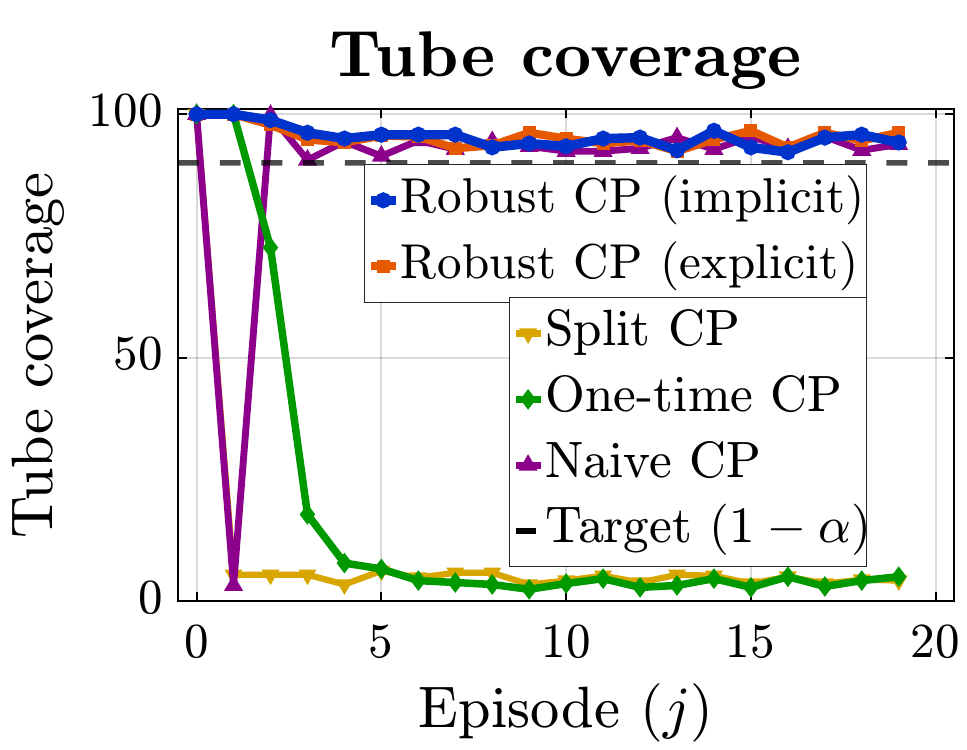} &
        \includegraphics[width=0.25\linewidth]{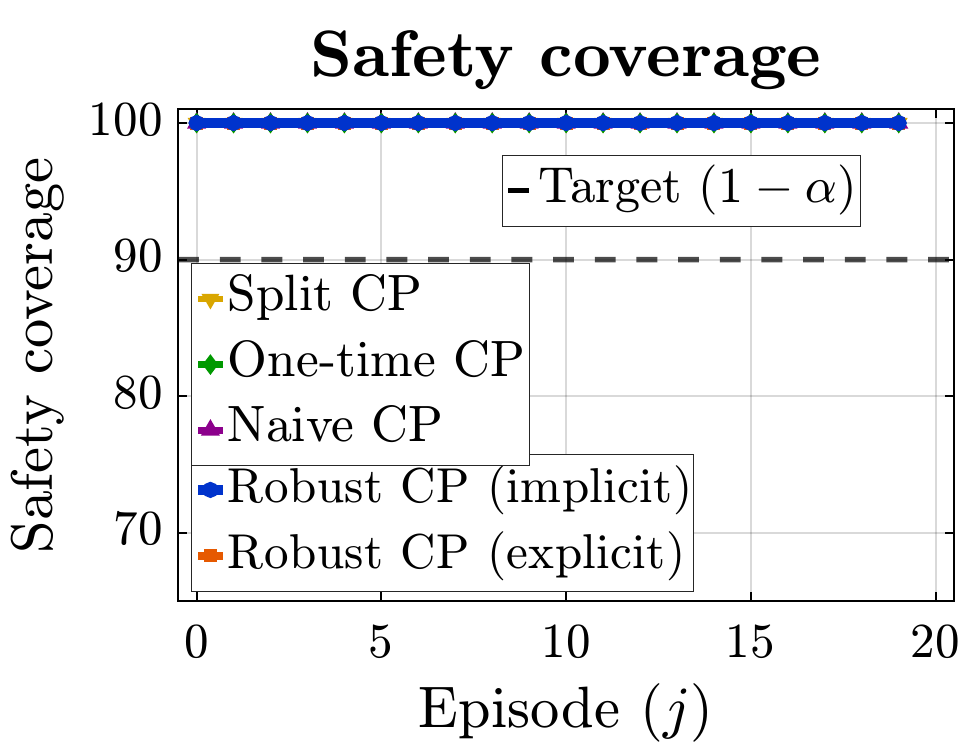}
    \end{tabular}
       \ldconly{ \vspace{-.3cm}}
        \caption{Car--pedestrian metrics. Left to right: $r_j$, $J_j$, empirical tube and safety coverages across episodes.}
    \label{fig:pedestrian-results}
    \ldconly{\vspace{-.55cm}}
\end{figure}

\begin{figure}[!t]
    \centering
    \captionsetup{font=small,skip=4pt}
    \setlength{\tabcolsep}{2pt}
    \begin{tabular}{@{}cc@{}}
        \includegraphics[width=0.28\linewidth]{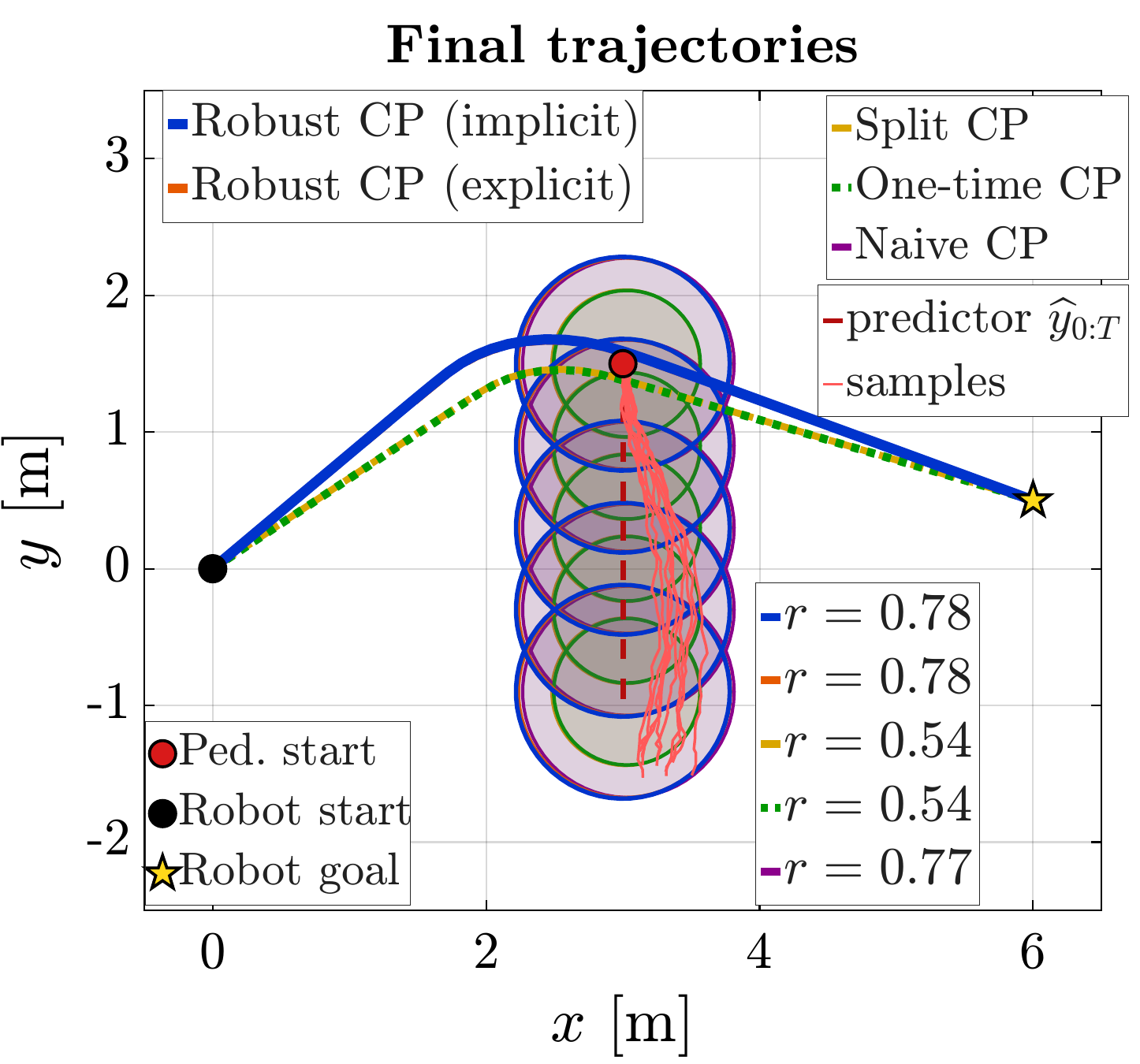} &
        \includegraphics[width=0.59\linewidth]{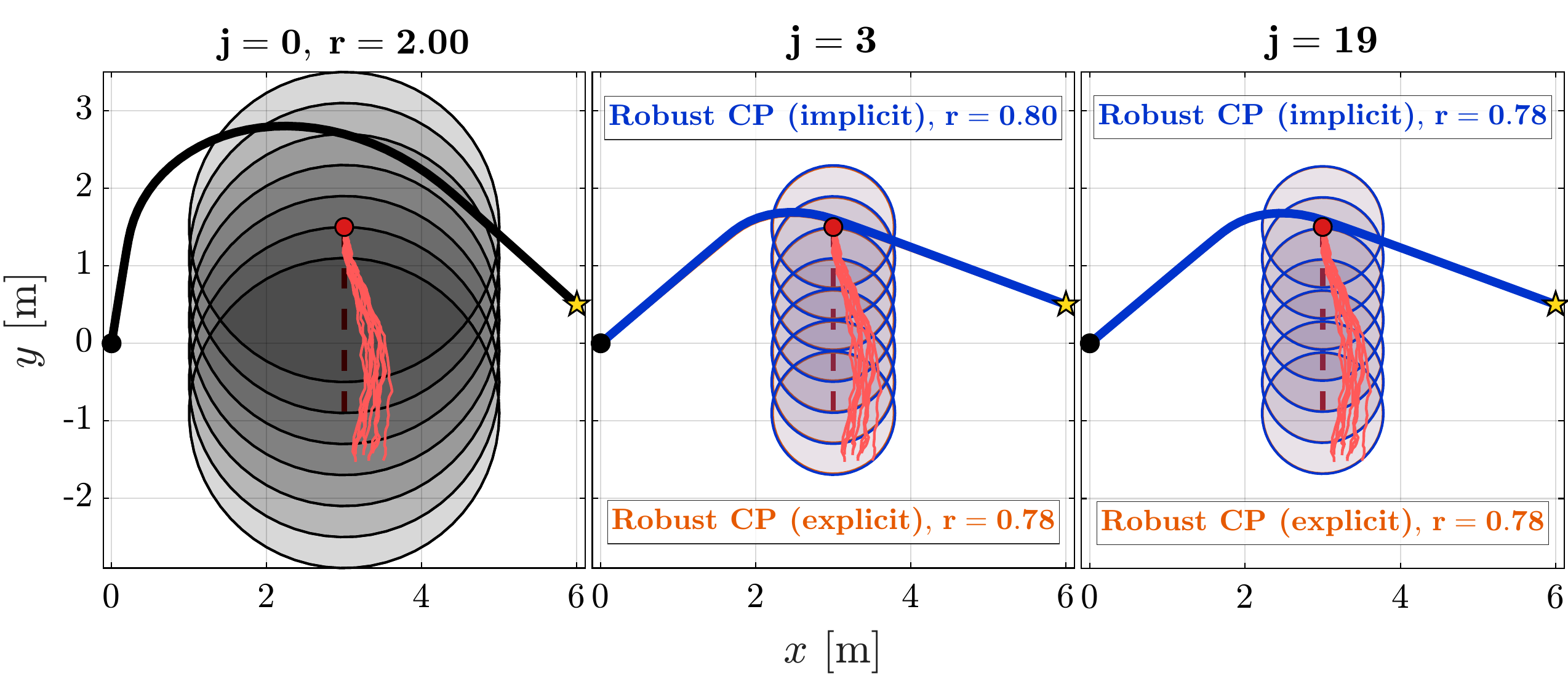}
    \end{tabular}
        \ldconly{\vspace{-.3cm}}
    \caption{Car–pedestrian trajectories. Left: final ego agent paths, all five rules. Right: tubes at $j = 0, 3, 19$.}
    \label{fig:pedestrian-trajectories}
    \ldconly{\vspace{-.7cm}}
\end{figure}

\begin{fullblock}

\section{Multi-Quadcopter Navigation}
\label{sec:quadswarm}

This appendix reports the results for the multi-quadcopter case study. All notation{,} targets ($1-\alpha=0.9$, {$\delta_j=0.05$}){,} {and baselines for comparison} are as defined in Section~\ref{sec:case_studies}. {Since the implicit solver requires multiple evaluations of $\pi^\star(r)$ per episode (see Appendix~\ref{sec:implicit_solver}), which is computationally expensive for the CBF-QP planner in~\eqref{eq:quadswarm-deployed-policy}, we employ only the explicit solver with update rule~\eqref{r_update_explicit} and report four rules: \emph{Robust CP -- explicit (ours)}, \emph{Naive CP}, \emph{One-time CP}, and \emph{Split CP}.}

\textbf{System and deployed policy.} We consider one ego agent and $N_{\mathrm{env}}=5$ environment quadrotors, each modeled as a double integrator system operating in $\mathbb{R}^3$ (an instantiation of~\eqref{eq:pf:dyn}), using the QuadSwarm simulator~\cite{quadswarm}. The environment agents execute learned navigation policies trained with collision penalties and do not deploy a test-time safety filter. {The ego agent uses nominal policy $\pi^{\mathrm{RL}}$ to track a goal position {at} every timestep {without knowledge of the environment agents}.  We approximate the cost minimization in $\mathbf{P}[j;r]$ by minimizing the distance between {the} robust action and the nominal action {at} every timestep, subject to the safety constraints.} {Concretely, we define for each environment agent $i\in[N_{\mathrm{env}}]$ and time $t$ the barrier function
\begin{equation}
h_i(x_t,t) := \|x_t - \hat y_t^{(i)}\|_2 - d_{\mathrm{coll}} - r_j,
\label{eq:quadswarm-barrier}
\end{equation}
which encodes separation from $\mathcal{C}_{r_j}(\hat y_{0:T})$. The ego agent then deploys the CBF-QP safety filter~\cite{ames2019cbf}
\begin{equation}
\pi_j(x_t) := \argmin_{u\in\mathcal{U}}\ \|u - \pi^{\mathrm{RL}}(x_t)\|_2^2 \;\;\; \text{s.t.}\;\; h_i\bigl(f_X(x_t,u),\,t{+}1\bigr) \ge (1{-}\gamma)\,h_i(x_t,t),\; \forall\, i\in[N_{\mathrm{env}}],
\label{eq:quadswarm-deployed-policy}
\end{equation}
where $\gamma\in(0,1]$ is the CBF decay rate and $\mathcal{U}$ is the admissible control set. Since $f_X$ is a double integrator (affine in $u$), \eqref{eq:quadswarm-deployed-policy} is a QP at each timestep.} Hence $\pi_j$ depends on $j$ only through $r_j$, playing the role of $\pi^\star(r_j)$ in Section~\ref{sec:iterative_planning}.
A collision is declared when $\|x_{j,t}-y_{j,t}^{(i)}\|_2<d_{\mathrm{coll}}:=2.5\,\ell_{\mathrm{arm}}$, with $\ell_{\mathrm{arm}} = 0.046$ being the quadrotor arm length.

\textbf{Calibration and update rule.} We use~\eqref{eq:score-tube} with the Euclidean norm in $\mathbb{R}^3$ and $n_j=200$. 
{We collect empirical quantiles $q_j^{(i)}$}%
\fullonly{ from~\eqref{eq:episodic_quantile}}%
\shortonly{ defined in Section~\ref{sec:acp_policy_shift}}
{for every agent $i \in [N_{\mathrm{env}}]$, and aggregate them {as}}
$\hat q_j:=\max_i q_j^{(i)}$ so that safety holds simultaneously for all agents{, with} $\hat q_j$ {replacing} $q_j$ in~\eqref{eq:implicit_problem} and~\eqref{r_update_explicit}. The admissible radius set is $\calR=[2\,\ell_{\mathrm{arm}},8\,\mathrm{m}]$ and the sensitivity gain is $\kappa=0.6$. All four update rules share $r_0=r_{\max}$ and the same predictor $\hat y_{0:T}$, initialized from a non-interactive rollout. An additional $100$ rollouts per episode estimate empirical tube coverage, empirical safety coverage, and the cumulative progress-to-goal reward.

\textbf{Observations.} Figure~\ref{fig:quadswarm-metrics} reports the four metrics. Under {\emph{Robust CP}}, $r_j$ contracts from $r_0=8\,\mathrm{m}$ toward $\approx 1\,\mathrm{m}$ (e.g.\ $r_1\approx 3.82\,\mathrm{m}$, $r_5\approx 1.02\,\mathrm{m}$); the first step is a strict convex combination of $r_0$ and $\hat q_0$ via the shrinkage branch of~\eqref{r_update_explicit}. 
Cumulative reward increases from $\approx 1.6$ to $\approx 2.68$, empirical tube coverage remains above $1-\alpha=0.9$, and empirical safety coverage stays at one, in agreement with Theorem~\ref{thm:per_episode_safety}. {\emph{Naive CP}} attains lower cumulative reward than the robust rule{{, as} the robust rule safely attain{s} an equilibrium with a smaller final radius}; {\emph{One-time CP}} and {\emph{Split CP}} hold $r_j$ fixed at their first certificate $q_{\mathrm{base}}\approx 1.31\,\mathrm{m}$ and likewise attain lower cumulative reward {at convergence}. Only {\emph{Robust CP}} converts the conservative initialization $r_0$ into a smaller stabilized operating radius without loss of either coverage, matching the shrinkage {branch of~\eqref{r_update_explicit} under} Theorem~\ref{thm:performance_convergence}. 
Figures~\ref{fig:quadswarm-qual} and~\ref{fig:quadswarm-qual-baselines} show the corresponding trajectories and tube evolution.

{We note that all four rules maintain empirical safety coverage at $100\%$ throughout (Figure~\ref{fig:quadswarm-metrics}$(a_4)$). This is attributable to the experimental setup:} {the environment agents were trained jointly in a homogeneous multi-agent environment in which all agents share{d} the same collision-avoidant objective and policy.} {Their learned policies are therefore conditioned on collision-averse behavior from all agents, including the ego agent, so that the policy-induced distribution shift $\mathcal{D}(\pi_j)\to\mathcal{D}(\pi_{j+1})$ per Lemma~\ref{prop:coupling_sensitivity} remains small even without the robustification term $M_{j+1}$ in~\eqref{eq:episodic_transfer}. This explains why \emph{Naive CP} does not exhibit the tube coverage collapse observed in the car--pedestrian study (Section~\ref{sec:case-pedestrian}).  Nevertheless, only \emph{Robust CP} systematically reduces $r_j$ while maintaining valid coverage guarantees per~\eqref{eq:episodic_transfer} and Theorem~\ref{thm:per_episode_safety}.}

{The qualitative panels in Figures~\ref{fig:quadswarm-qual} and~\ref{fig:quadswarm-qual-baselines} support this interpretation. As {$r_j$} contracts, the ego trajectory straightens and the tubes around representative environment agents shrink accordingly, yet {continue to contain} the realized trajectories.} {The baseline trajectories in Figure~\ref{fig:quadswarm-qual-baselines} remain visually similar across episodes, consistent with the small distribution shift in this setting, while operating at a fixed and larger safety margin than the converged \emph{Robust CP} radius.}

\begin{figure}[!htbp]
    \centering
    \captionsetup{font=small,skip=2pt}
    \begin{subfigure}[t]{0.24\linewidth}
        \centering\small $(a_1)$\\[-0.1em]
        \includegraphics[width=\linewidth]{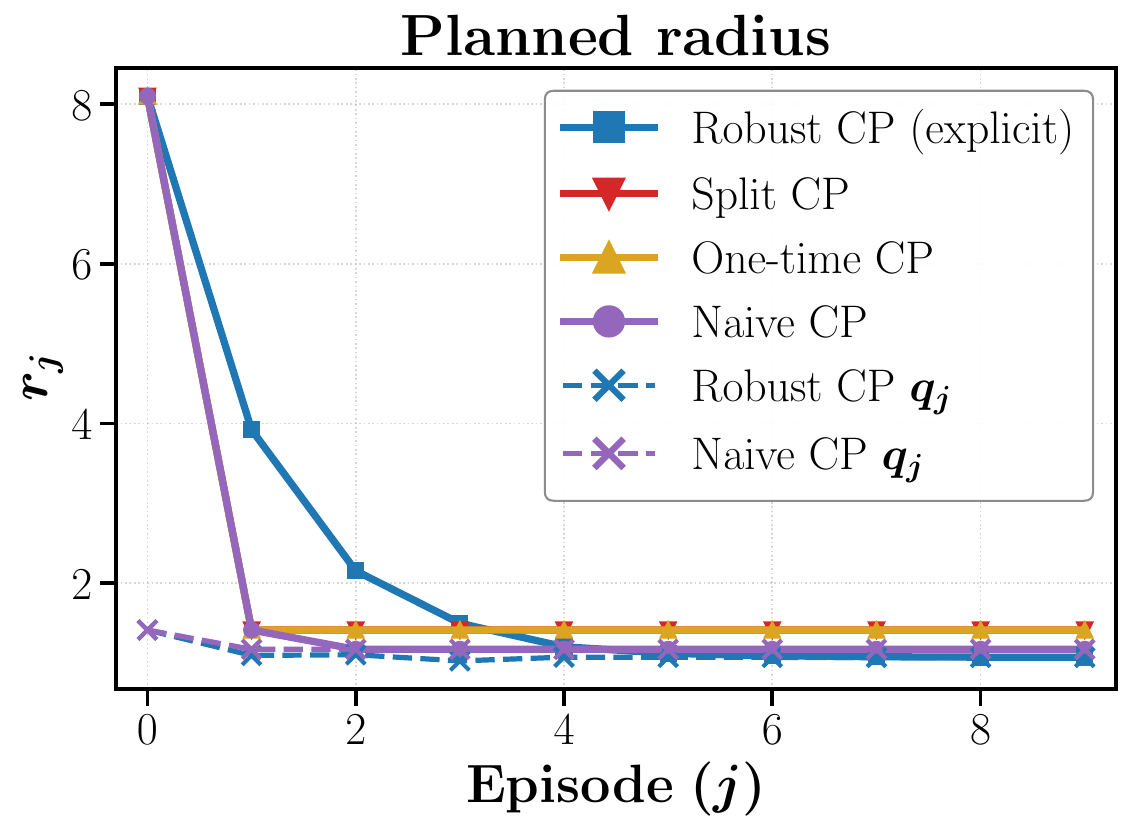}
    \end{subfigure}\hfill
    \begin{subfigure}[t]{0.24\linewidth}
        \centering\small $(a_2)$\\[-0.1em]
        \includegraphics[width=\linewidth]{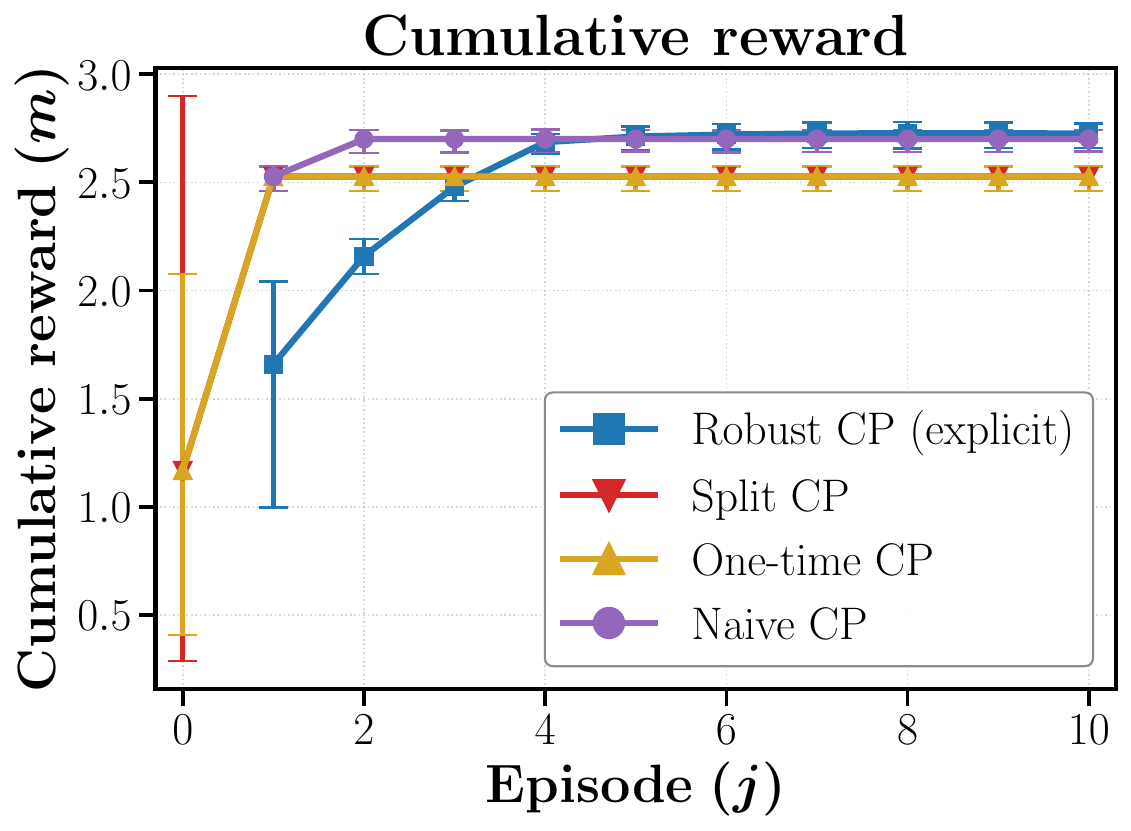}
    \end{subfigure}\hfill
    \begin{subfigure}[t]{0.24\linewidth}
        \centering\small $(a_3)$\\[-0.1em]
        \includegraphics[width=\linewidth]{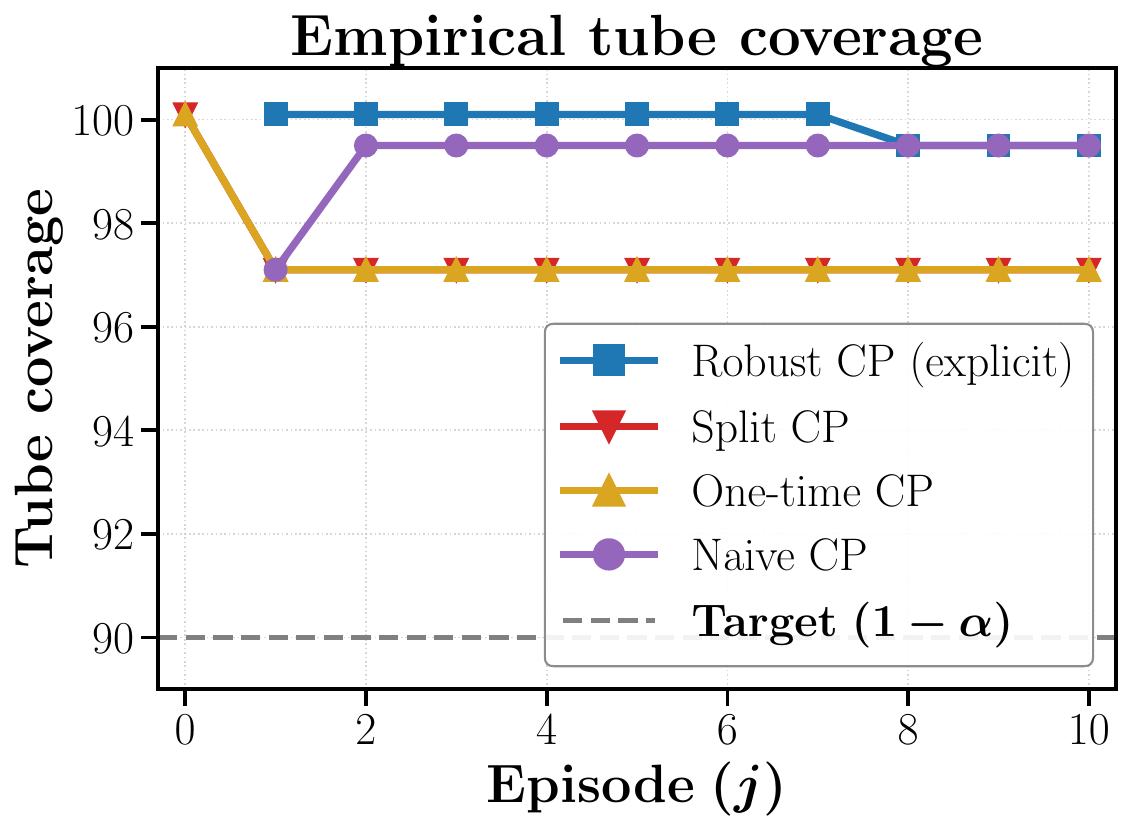}
    \end{subfigure}\hfill
    \begin{subfigure}[t]{0.24\linewidth}
        \centering\small $(a_4)$\\[-0.1em]
        \includegraphics[width=\linewidth]{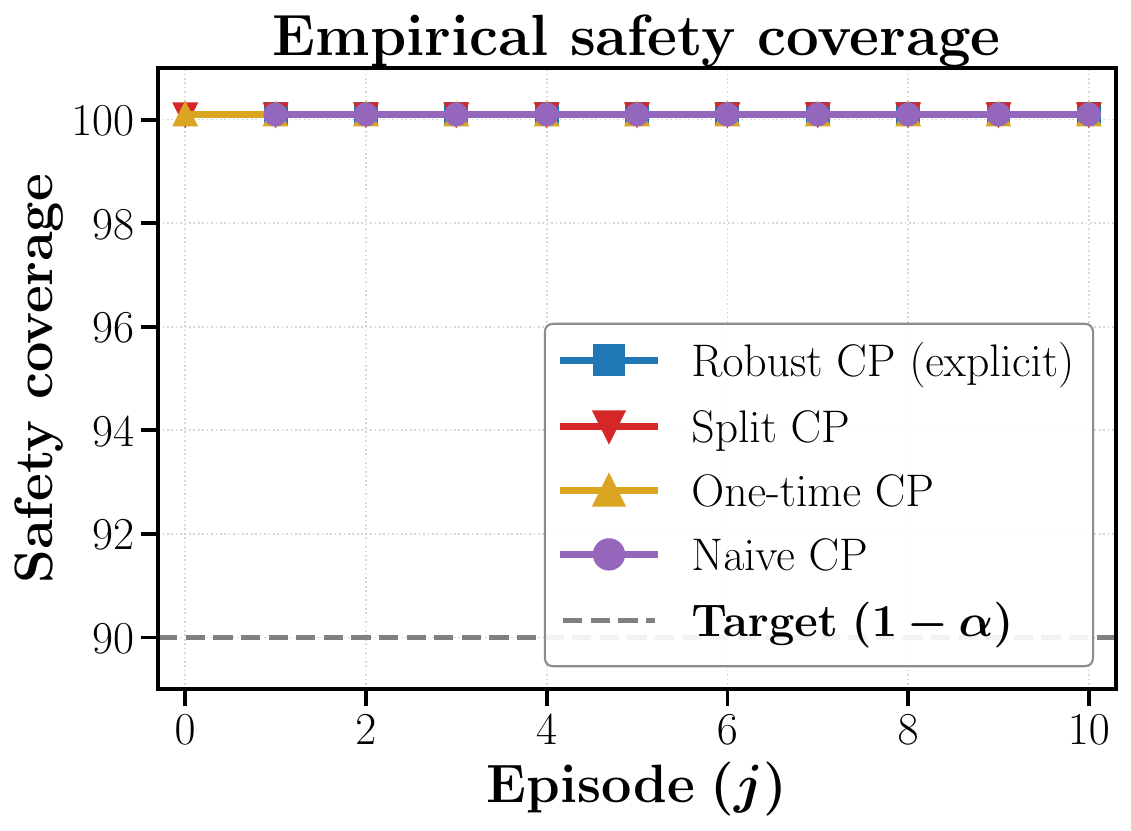}
    \end{subfigure}
    \caption{Multi-quadcopter quantitative results. Left to right: deployed radius $r_j$ (with the current aggregate quantile $\hat q_j$ overlaid for the robust rule), cumulative reward, empirical tube coverage (target $1-\alpha=0.9$), and empirical safety coverage. {\emph{Robust CP}} contracts $r_j$ and attains the highest cumulative reward while retaining both coverages at their target.}
    \label{fig:quadswarm-metrics}
\end{figure}

\begin{figure*}[!thbp]
    \centering
    \captionsetup{font=small,skip=2pt}
    \setlength{\tabcolsep}{0pt}
    \renewcommand{\arraystretch}{0.0}
    \begin{tabular}{@{}ccccc@{}}
        \includegraphics[width=0.20\linewidth]{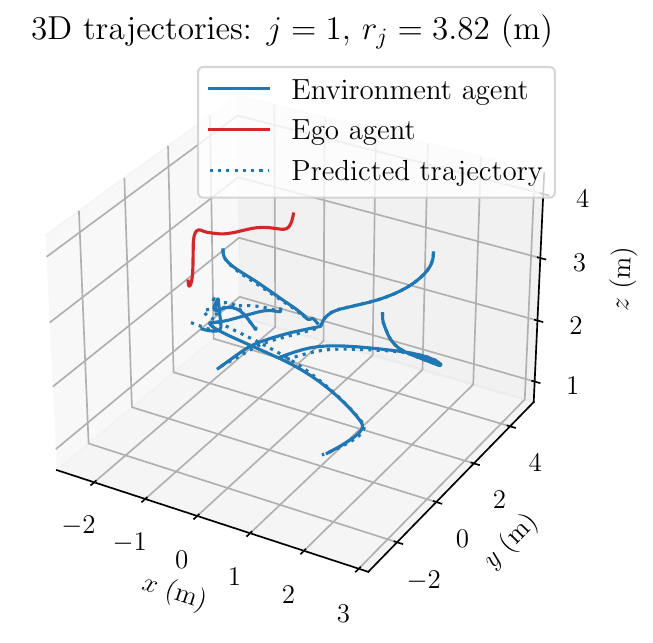} &
        \includegraphics[width=0.20\linewidth]{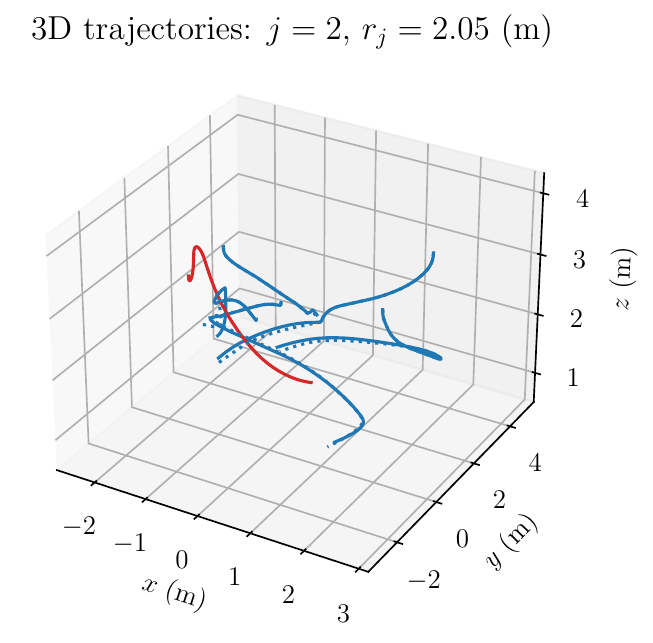} &
        \includegraphics[width=0.20\linewidth]{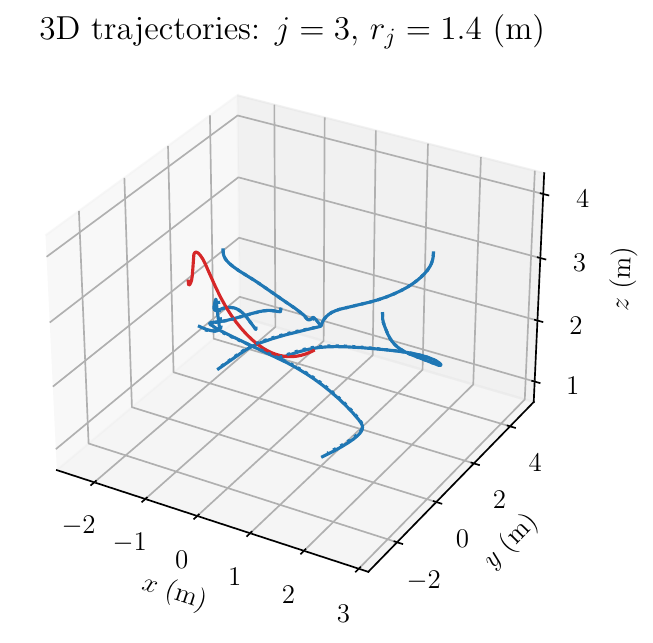} &
        \includegraphics[width=0.20\linewidth]{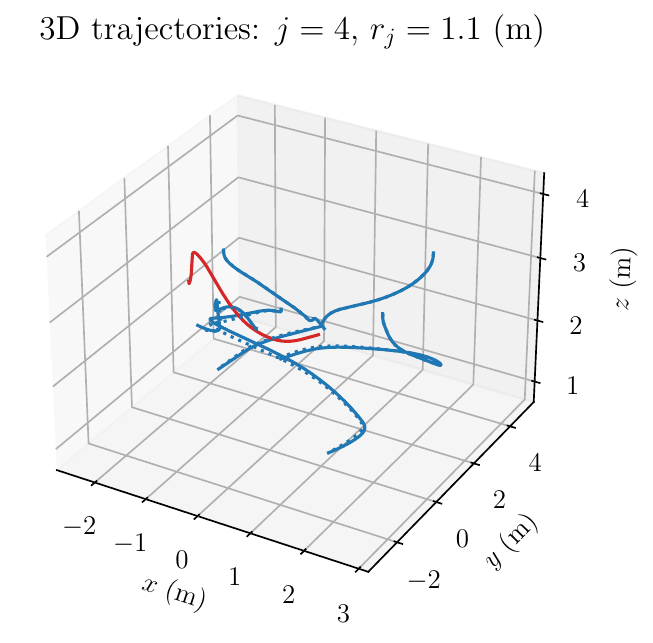} &
        \includegraphics[width=0.20\linewidth]{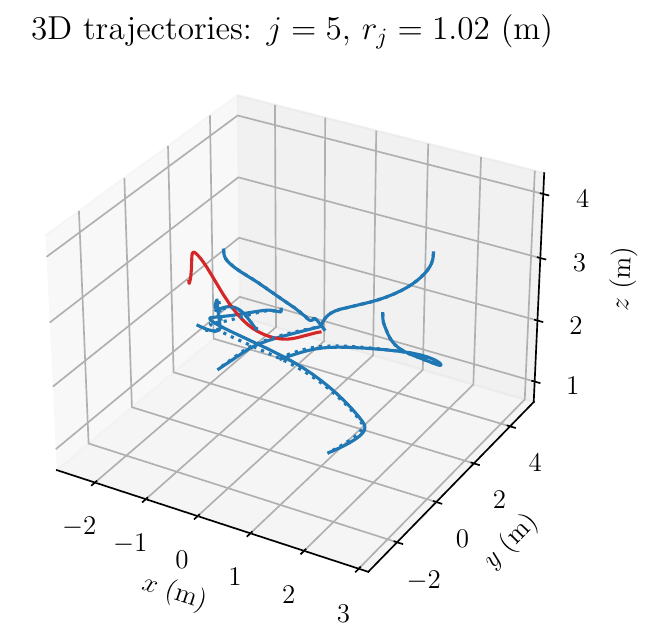}
    \end{tabular}
    \begin{tabular}{@{}ccccc@{}}
        \includegraphics[width=0.20\linewidth]{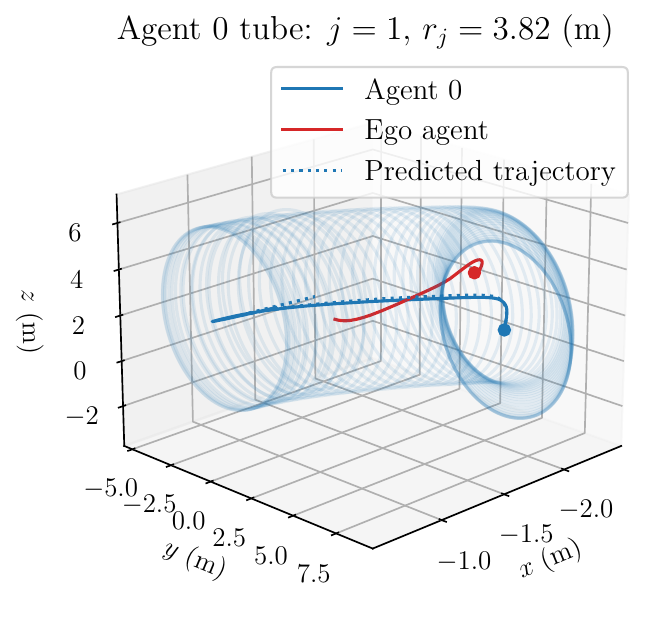} &
        \includegraphics[width=0.20\linewidth]{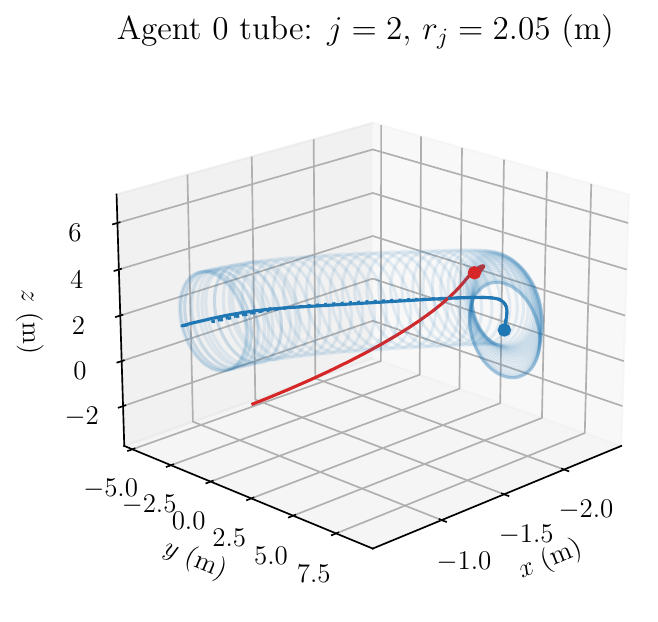} &
        \includegraphics[width=0.20\linewidth]{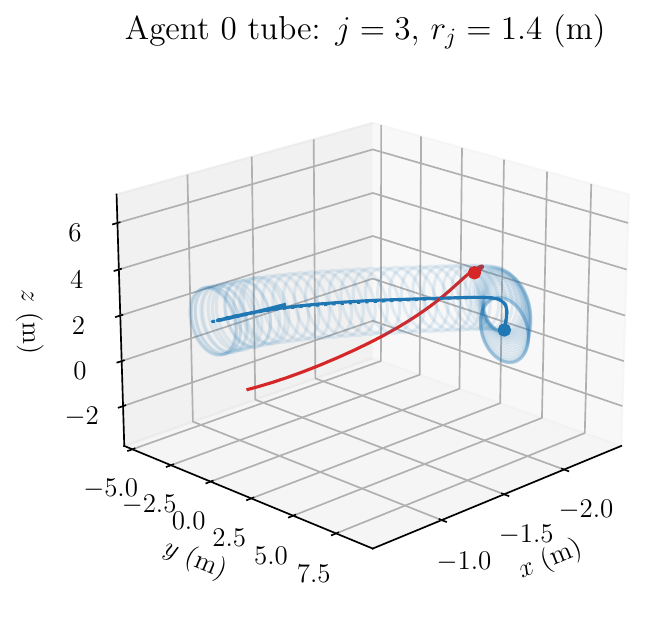} &
        \includegraphics[width=0.20\linewidth]{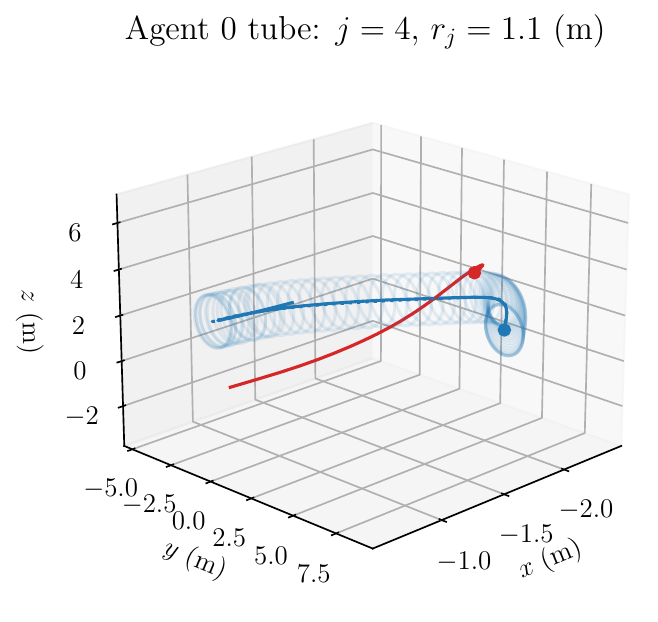} &
        \includegraphics[width=0.20\linewidth]{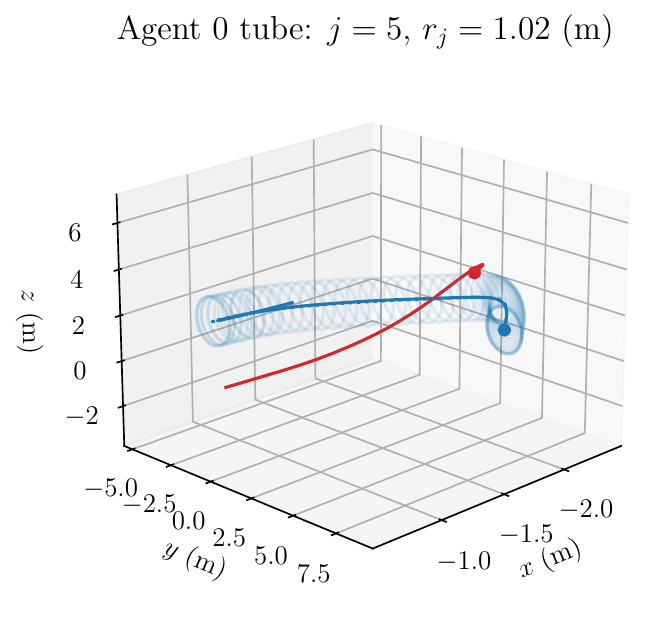}
    \end{tabular}
    \begin{tabular}{@{}ccccc@{}}
        \includegraphics[width=0.20\linewidth]{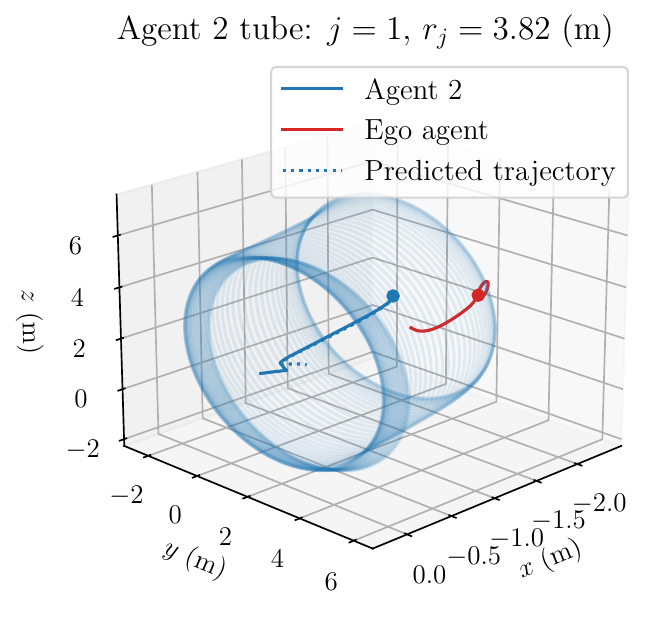} &
        \includegraphics[width=0.20\linewidth]{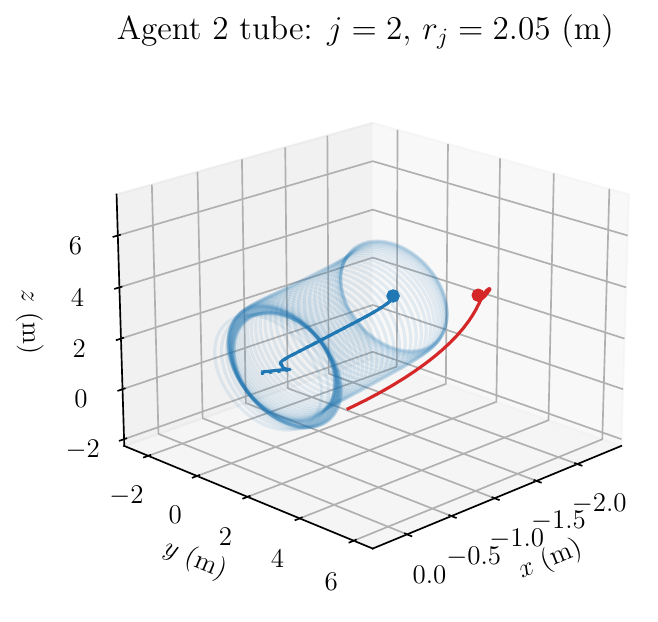} &
        \includegraphics[width=0.20\linewidth]{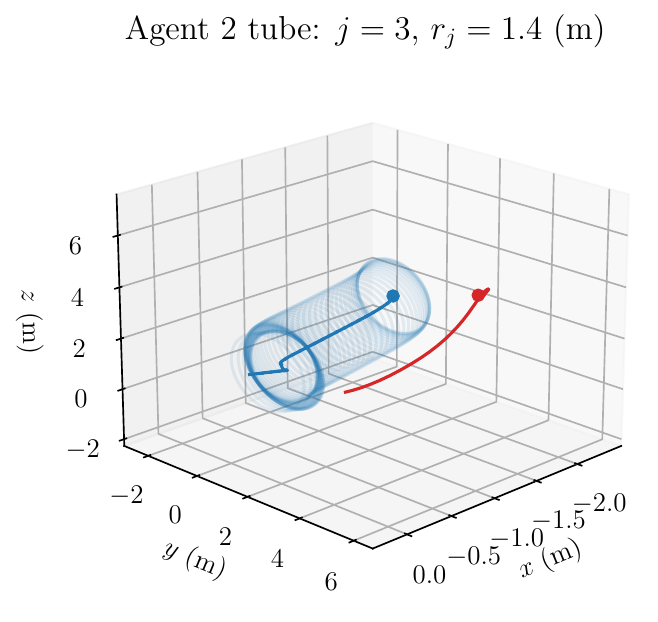} &
        \includegraphics[width=0.20\linewidth]{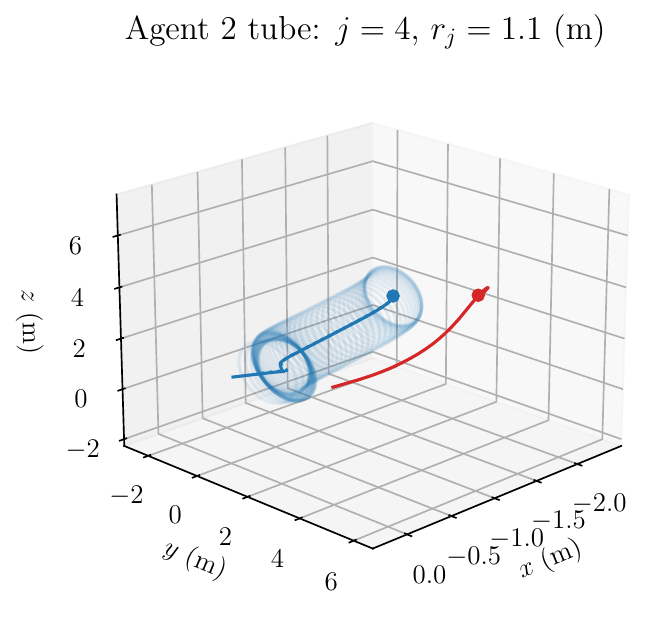} &
        \includegraphics[width=0.20\linewidth]{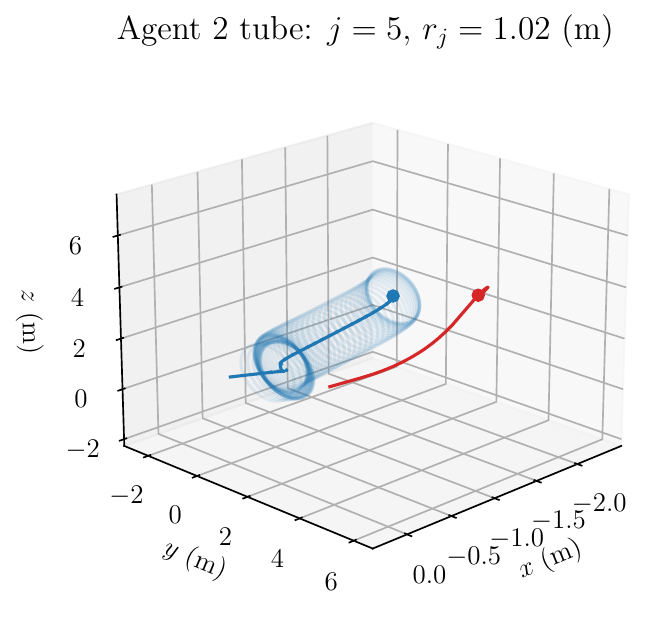}
    \end{tabular}
    \caption{Qualitative evolution under {\emph{Robust CP}} across episodes $j=1,\ldots,5$ (left to right). Row~1: realized trajectories (solid) and nominal predicted trajectories (dotted). Rows~2--3: conformal tubes $\mathcal{C}_{r_j}(\hat y_{0:T})$ for two representative environment agents. As $r_j$ contracts, the ego agent's path straightens while the tubes continue to contain the realized environment trajectories.}
    \label{fig:quadswarm-qual}
\end{figure*}
\begin{figure}[!htbp]
    \centering \captionsetup{font=small,skip=2pt}
    \setlength{\tabcolsep}{0pt}
    \renewcommand{\arraystretch}{0.0}
    \begin{tabular}{@{}ccccc@{}}
        \multicolumn{2}{c}{{\emph{Naive CP}}} &
        \multicolumn{2}{c}{{\emph{One-time CP}}} &
        {\emph{Split CP}} \\
        \includegraphics[width=0.19\linewidth]{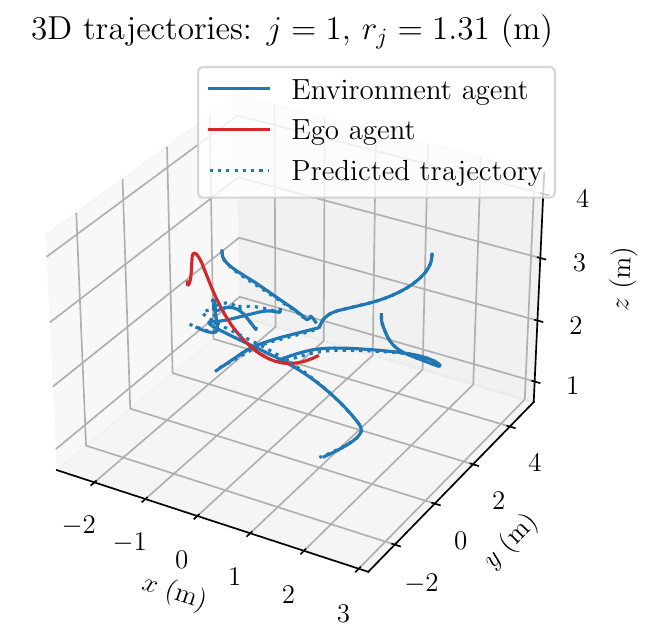} &
        \includegraphics[width=0.19\linewidth]{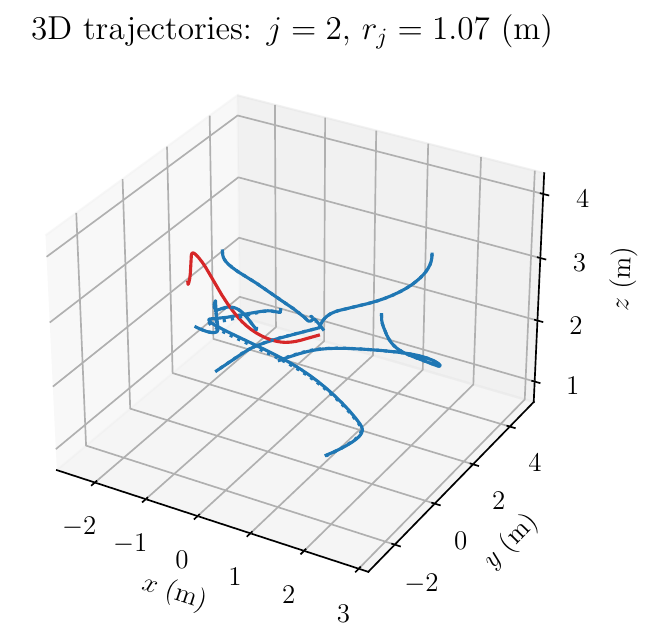} &
        \includegraphics[width=0.19\linewidth]{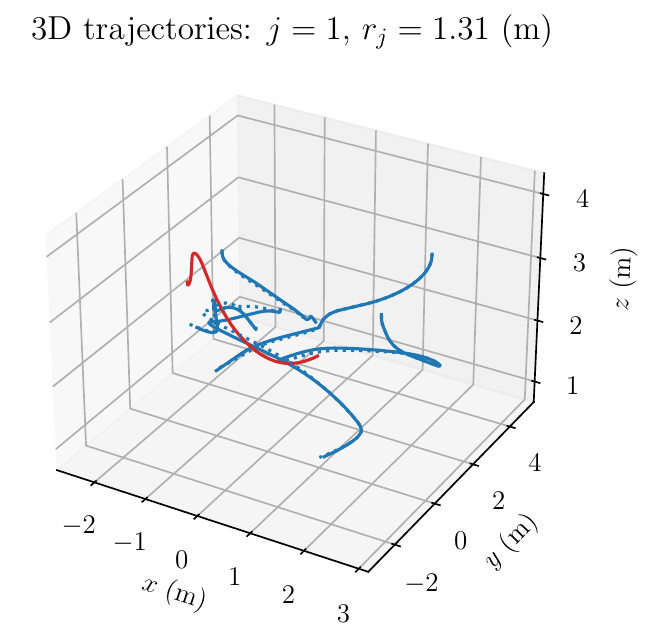} &
        \includegraphics[width=0.19\linewidth]{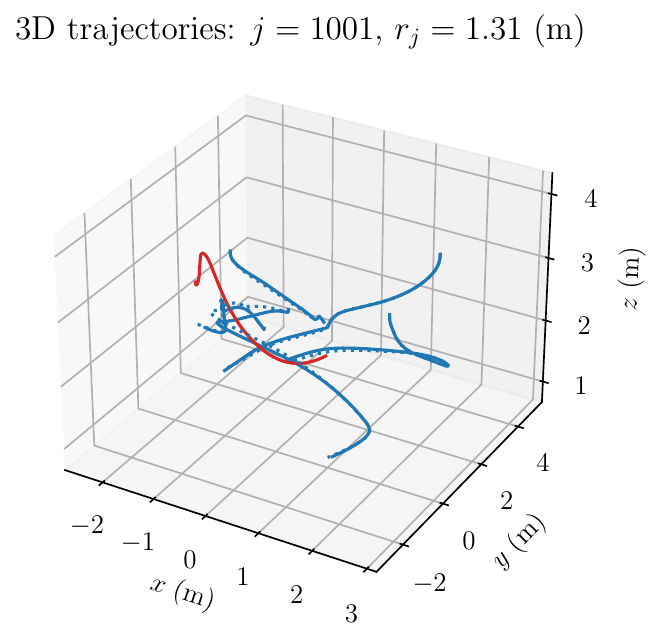} &
        \includegraphics[width=0.19\linewidth]{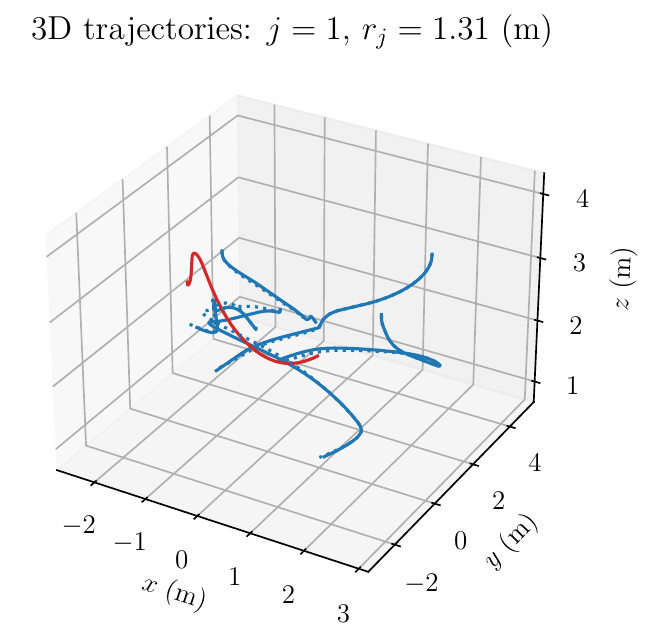} \\
        \includegraphics[width=0.19\linewidth]{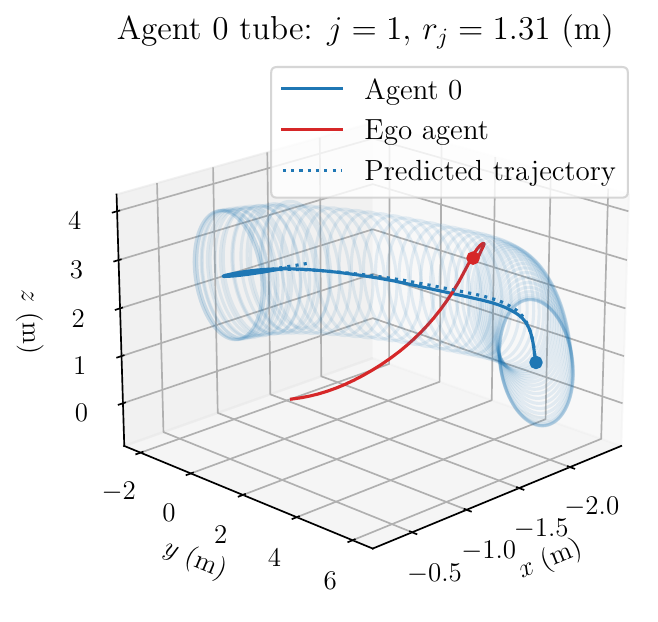} &
        \includegraphics[width=0.19\linewidth]{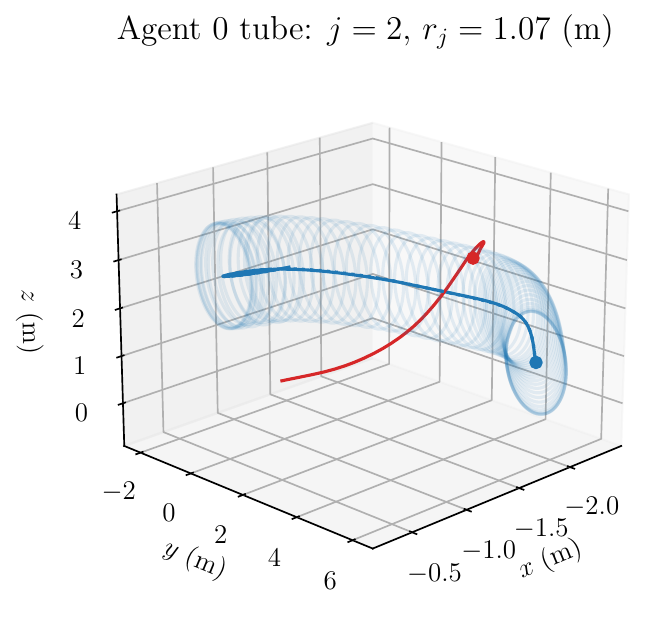} &
        \includegraphics[width=0.19\linewidth]{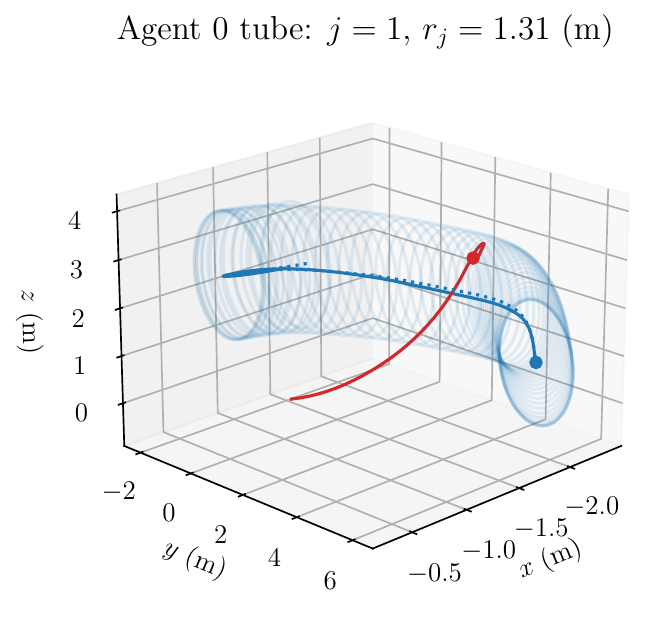} &
        \includegraphics[width=0.19\linewidth]{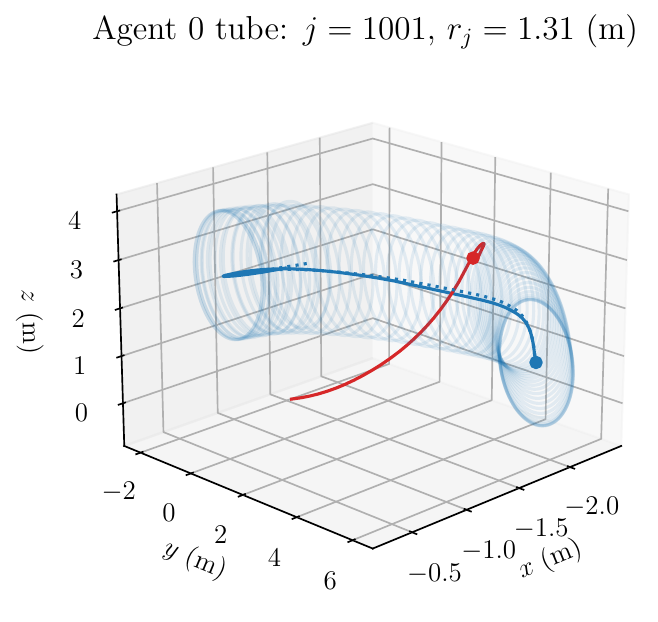} &
        \includegraphics[width=0.19\linewidth]{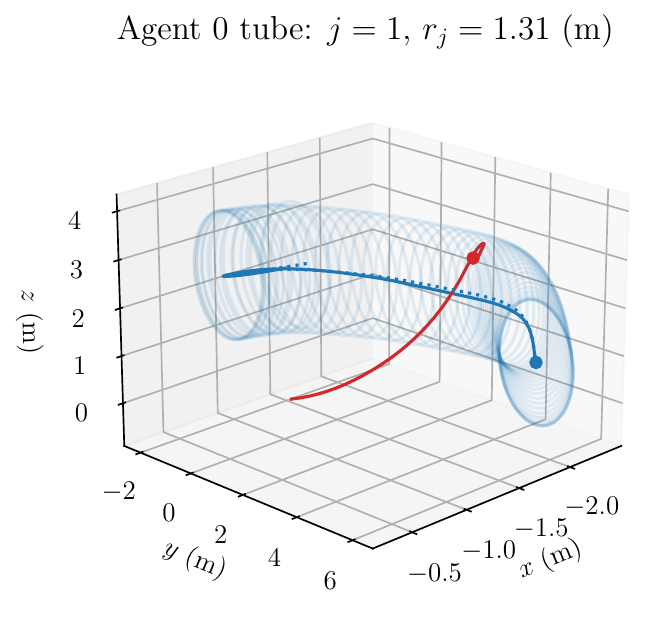} \\
        \includegraphics[width=0.19\linewidth]{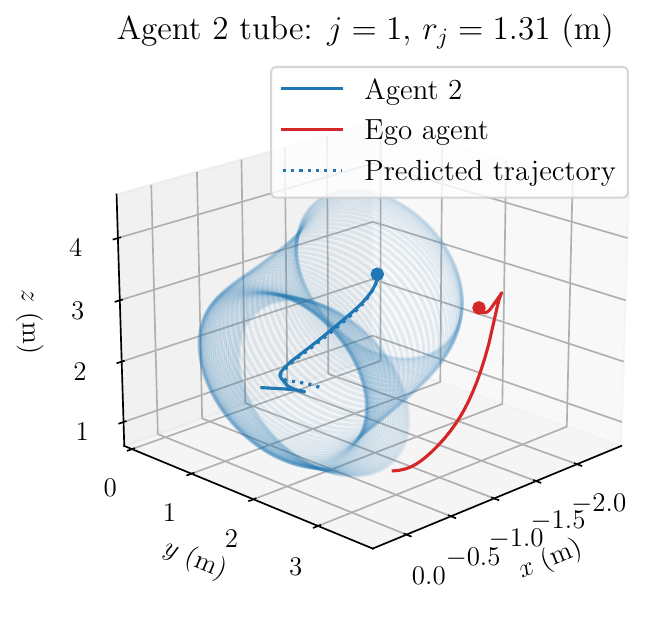} &
        \includegraphics[width=0.19\linewidth]{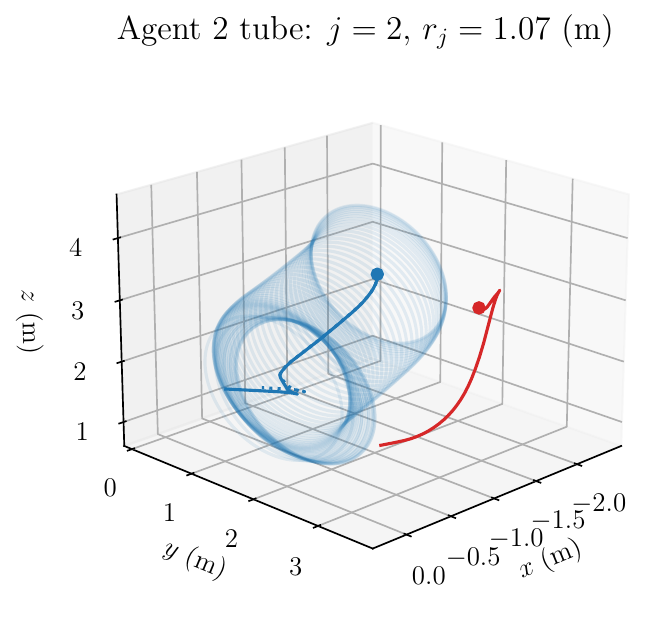} &
        \includegraphics[width=0.19\linewidth]{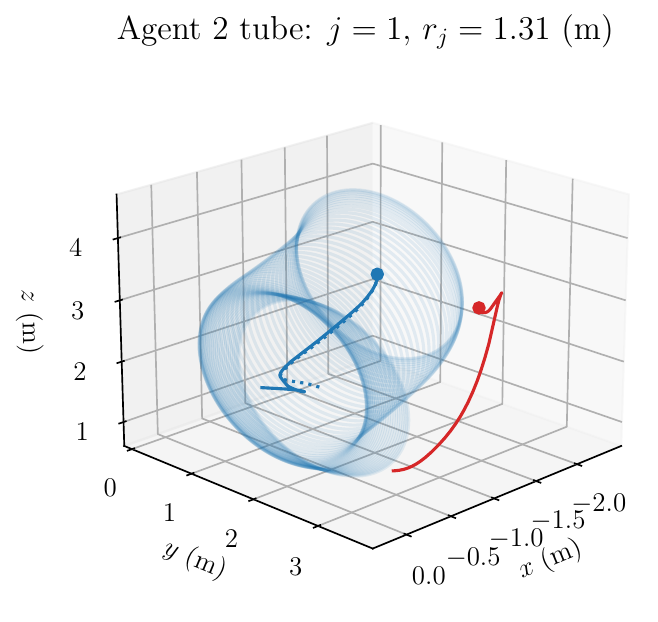} &
        \includegraphics[width=0.19\linewidth]{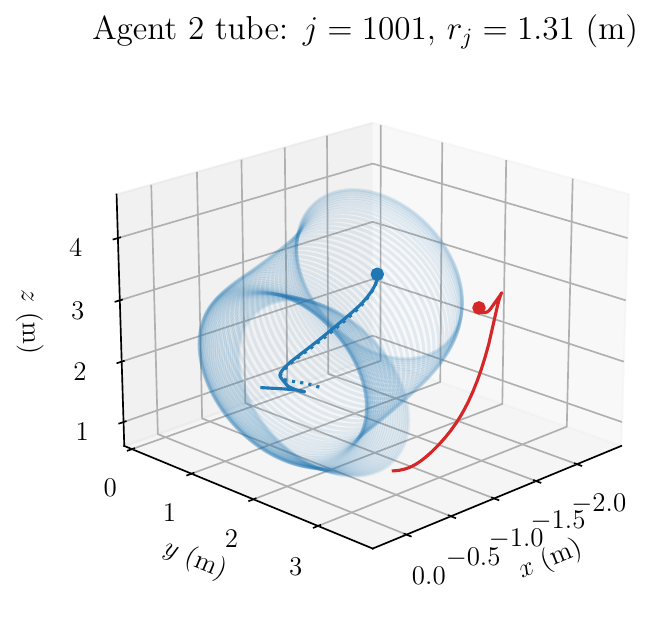} &
        \includegraphics[width=0.19\linewidth]{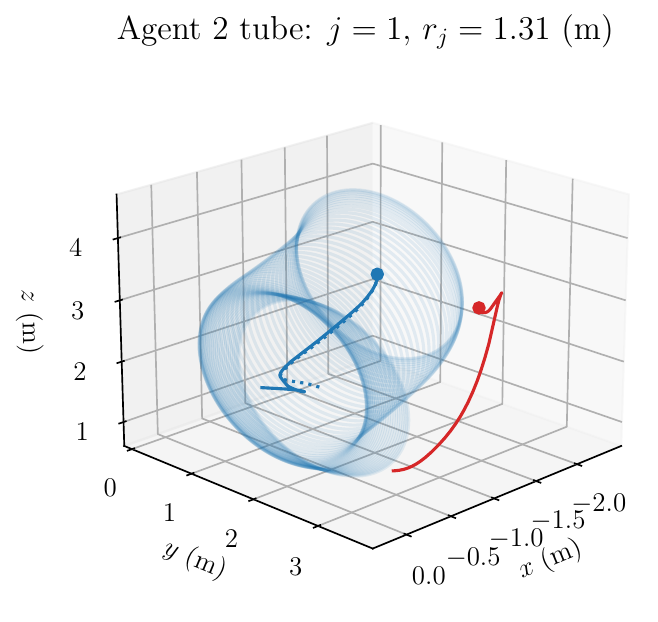}
    \end{tabular}
    \caption{Qualitative baseline comparison. {\emph{Naive CP}} at $j\in\{1,2\}$; {\emph{One-time CP}} at $j\in\{1,1001\}$; {\emph{Split CP}} at $j=1$. Row~1: realized trajectories; rows~2--3: conformal tubes for two representative environment agents.}
    \label{fig:quadswarm-qual-baselines}
    % \jmlrendhere
\end{figure}
\end{fullblock}

\begin{fullblock}
\subsection{Discussion}\label{sec:discussion}
\end{fullblock}
\begin{shortblock}
\textbf{4.2 Discussion}\label{sec:discussion}
\end{shortblock}
Across both case studies, the two \emph{Robust CP} variants are the only rules that simultaneously (i) maintain empirical tube and safety coverages at the target $1-\alpha$ and (ii) contract $r_j$ from $r_0=r_{\max}$, yielding a reduction in $J_j$ in the car--pedestrian study and an increase in cumulative reward in the multi-quadcopter study. \emph{Naive CP} omits $M_{j+1}$ in~\eqref{eq:episodic_transfer} and exhibits early-episode tube undercoverage in the car--pedestrian study, recovering only after the radius stabilizes---precisely the failure mode that~\eqref{eq:episodic_transfer} is designed to prevent. \emph{One-time CP} and \emph{Split CP} cannot track $\mathcal{D}(\pi_j)$ and lose tube coverage; while their smaller radii yield lower planner costs, these lack valid coverage guarantees per~\eqref{eq:episodic_transfer}. Per-episode safety is delivered by the transfer requirement~\eqref{eq:episodic_transfer} through Theorem~\ref{thm:per_episode_safety}, and the cost improvement is delivered by the shrinkage branch of~\eqref{r_update_explicit} through Theorem~\ref{thm:performance_convergence}.

% =============================================================
% ACKNOWLEDGEMENTS  (full content only)
% =============================================================
% Wrapped in \begin{fullblock} so they appear in both extended
% builds but NOT in the L4DC short submission.
\begin{fullblock}
\section*{Acknowledgements}
Nikolai Matni and Eliot Shekhtman are supported in part by AFOSR Award FA9550-24-1-0102, SF Award SLES-2331880, and NSF CAREER award ECCS-2045834.
\end{fullblock}

% =============================================================
% BIBLIOGRAPHY  (per-class)
% =============================================================
\ifarticleclass
  \bibliographystyle{IEEEtran}
\fi
\bibliography{main}

% =============================================================
% APPENDICES
% =============================================================
% Each appendix is gated by its own boolean from the user-
% editable settings block at the top of this file.
% =============================================================

\ifarticleclass
  \addtolength{\textheight}{-2cm}
  \newpage
\fi

\appendix
\ifincludeAppProofs \section{Deferred Proofs}
\label{app:proofs}

\subsection{Proof of Lemma \ref{prop:coupling_sensitivity}}
\label{app:proof_beta_t}
\subeqgentleon
Let us first recall the dynamics of ego and uncontrollable agents from \eqref{eq:pf:dyn} as
\begin{align*}
x_{t+1} = f_X(x_t,u_t), 
\qquad
y_{t+1} = f_Y(y_t,x_t,u_t,\nu_t),
\end{align*}
under the noise sequence $
\{\nu_t\}_{t=0}^{T-1}$. 
Under Assumption \ref{ass:lipschitz_dynamics}, i.e., Lipschitz continuity of $f_X$ and $f_Y$, our goal is now to bound the deviation in the environment trajectory from $y_{0:T}$ to $y'_{0:T}$ caused by a change in the policy from $\pi = \{u_t\}_{t=1}^{T-1}$ to $\pi' = \{u'_t\}_{t=1}^{T-1}$. Indeed, we want to show that
\begin{align}
\|x_{0:T}-x_{0:T}'\|_{\infty}&\le A_T\,\Delta u,\qquad 
A_T:=L_{Xu}\sum_{t=0}^{T-1}(L_{Xx})^t,
\label{eq:C-Xsens}\\
\|y_{0:T}-y_{0:T}'\|_{\infty}&\le \beta_T\,\Delta u,\qquad 
\beta_T:=(L_{Yx}A_T+L_{Yu})\sum_{t=0}^{T-1}(L_{Yy})^t .
\label{eq:C-Ysens}
\end{align}
where $\Delta u:=\|\pi-\pi'\|_\infty$ with $
\|\pi-\pi'\|_\infty := \max_{0\le t\le N-1}\|u_t-u'_t\|_2
$.
\begin{proof}
Let $\Delta x_t:=x_t-x'_t$, $\Delta y_t:=y_t-y'_t$, and $\Delta u_t:=u_t-u'_t$. From Assumption \ref{ass:lipschitz_dynamics},  we get
\begin{equation}
\|\Delta x_{t+1}\|_2\le L_{Xx}\|\Delta x_t\|_2+L_{Xu}\|\Delta u_t\|_2.
\label{eq:C-recX}
\end{equation}
Unrolling \eqref{eq:C-recX} until time $t$  with $\Delta x_0=0$ and $\|\Delta u_t\|_2\le\Delta u$ then gives
\begin{equation}
\|\Delta x_t\|_2\le L_{Xu}\sum_{i=0}^{t-1}(L_{Xx})^i\,\Delta u.
\label{eq:C-Xbound}
\end{equation}
Since all Lipschitz constants are nonnegative, the partial sum in \eqref{eq:C-Xbound} is monotone in $t$ so that
\[
\|x_{0:T}-x'_{0:T}\|_{\infty}=\max_{0\le t\le T}\|\Delta x_t\|_2=\|\Delta x_T\|_2
\le L_{Xu}\sum_{t=0}^{T-1}(L_{Xx})^t\,\Delta u,
\]
which proves \eqref{eq:C-Xsens}. Next, from Assumption \ref{ass:lipschitz_dynamics}, we get
\begin{align*}
\|\Delta y_{t+1}\|_2\le L_{Yy}\|\Delta y_t\|_2+L_{Yx}\|\Delta x_t\|_2+L_{Yu}\|\Delta u_t\|_2.
\end{align*}
Using \eqref{eq:C-Xbound} and $\|\Delta u_t\|_2\le\Delta u$, it the follows that
\begin{equation}
\|\Delta y_{t+1}\|_2\le L_{Yy}\|\Delta y_t\|_2+\big(L_{Yx}A_T+L_{Yu}\big)\Delta u.
\label{eq:C-recY2}
\end{equation}
Unrolling \eqref{eq:C-recY2} until time $t$ with $\Delta y_0=0$ then gives
\begin{align*}
\|\Delta y_t\|_2\le \big(L_{Yx}A_T+L_{Yu}\big)\sum_{i=0}^{t-1}(L_{Yy})^i\,\Delta u.
\end{align*}
Again by nonnegativity of the Lipschitz constants, this sum is monotone in $t$ so that
\[
\|y_{0:T}-y'_{0:T}\|_{\infty}=\|\Delta y_T\|_2
\le \big(L_{Yx}A_T+L_{Yu}\big)\sum_{t=0}^{T-1}(L_{Yy})^t\,\Delta u,
\]
which proves \eqref{eq:C-Ysens}, what was to be shown. 
\end{proof}

\subsection{Derivation of Planner Sensitivity \texorpdfstring{$L_U$}{L extunderscore U}}
\label{app:proof_planner_sensitivity}

\begin{proposition}[Planner Regularity \& Lipschitzness]
\label{ass:planner_lipschitz}
Assume there exists a compact interval of radii $\mathcal{R} := [r_{\min}, r_{\max}]$ such that for all $r \in \mathcal{R}$, the robust planning problem $\mathbf{P}[j;\,r_j]$ in \eqref{eq:episodic} is feasible, convex, and has a unique, regular optimizer $\pi^*(r)$ satisfying standard KKT conditions (LICQ, strict complementarity, and a nonsingular KKT matrix). Then the solution map $r \mapsto \pi^*(r)$ is locally Lipschitz on $\mathcal{R}$. In particular, there exists a uniform constant $L_U<\infty$ such that
\begin{equation}
\label{eq:planner_lip}
\|\pi^{*}(r)-\pi^{*}(r^{\prime})\|_\infty \;\le\; L_{U}\,|r-r^{\prime}|\,, \qquad \forall\, r, r' \in \mathcal{R}.
\end{equation}
\end{proposition}

\begin{proof}
We derive $L_U$ via parametric sensitivity of $\mathbf{P}[j;r]$ in  \eqref{eq:episodic}. Throughout, we assume the inequality constraints depend affinely on the radius,
\begin{align*}
a_k(\pi) + b_k\, r \le 0 \quad (k=1,\dots,m),    
\end{align*}
and any equality constraints $h(\pi)=0$ do not depend on $r$.\footnote{For instance, this is the case for the program in  \eqref{eq:episodic} when $H$ is Lipschitz continuous with Lipschitz constant $L_H$ so that $H(x_{j,0:T},\zeta)\le 0\quad \forall\,\zeta\in \mathcal{C}_{r_j}(\hat y_{0:T})$ can be rewritten as $H(x_{j,0:T},\hat y_{0:T}) \;\le\; -\,L_H\,r_{j}$ (see Appendix \ref{app:sufficiency}).}\footnote{If $h$ (or the gradients of $a_k$) depend on $r$, the derivative formula below includes additional right-hand-side terms; see Remark~\ref{rem:r_dep}. The Lipschitz conclusion \eqref{eq:planner_lip} is unchanged under strong regularity.} 

\paragraph{Active set and multipliers.}
Fix $r\in\mathcal{R}$ and let $\pi^*(r)$ denote the unique optimizer. Define the active set
\begin{align*}
 \mathcal{A}(r) \;:=\; \big\{\,k \;:\; a_k(\pi^*(r)) + b_k\,r = 0\,\big\},   
\end{align*}
and let $\lambda_k^*(r)>0$ be the corresponding Lagrange multipliers for $k\in\mathcal{A}(r)$ (strict complementarity). Let $\mu^*(r)$ be the multipliers for $h(\pi)=0$ (if present).

\paragraph{Second-order objects.}
Define the Hessian of the Lagrangian, the Jacobian of active constraints, and the vector of active radius coefficients:
\begin{align}
H(r)
&:= \nabla_{\!\pi\pi}^2 J(\pi^*(r)) + \sum_{k\in\mathcal{A}(r)} \lambda_k^*(r)\,\nabla_{\!\pi\pi}^2 a_k(\pi^*(r)),
\label{eq:H}
\\[1mm]
G_{\mathcal{A}}(r)
&:= \begin{bmatrix} \nabla_\pi a_k(\pi^*(r)) \end{bmatrix}_{k\in\mathcal{A}(r)} \in \mathbb{R}^{|\mathcal{A}(r)|\times n},
\label{eq:GA}
\\[1mm]
b_{\mathcal{A}}(r)
&:= \begin{bmatrix} b_k \end{bmatrix}_{k\in\mathcal{A}(r)} \in \mathbb{R}^{|\mathcal{A}(r)|}.
\label{eq:bA}
\end{align}
If equalities are present, let $E(r):=\nabla_\pi h(\pi^*(r))$.

\paragraph{KKT conditions.}
At $(\pi^*,\lambda^*,\mu^*;r)$ the KKT conditions read
\begin{align}
&\nabla_\pi J(\pi^*(r)) \;+\; \sum_{k\in\mathcal{A}(r)} \lambda_k^*(r)\,\nabla_\pi a_k(\pi^*(r)) \;+\; E(r)^\top \mu^*(r) \;=\; 0,
\label{eq:kkt_stat}
\\
& a_k(\pi^*(r)) + b_k\,r \;\le\; 0 \ \ (k=1,\dots,m), 
\qquad h(\pi^*(r)) \;=\; 0,
\label{eq:kkt_primal}
\\
& \lambda_k^*(r) \;\ge\; 0 \ \ (k=1,\dots,m), 
\qquad \lambda_k^*(r)\,\big(a_k(\pi^*(r)) + b_k\,r\big) \;=\; 0 \ \ (k=1,\dots,m).
\label{eq:kkt_comp}
\end{align}
Under LICQ, strict complementarity, and the second-order condition, the KKT matrix is nonsingular (strong regularity), and the active set $\mathcal{A}(r)$ is locally constant.

\paragraph{Linearized KKT system.}
Differentiate \eqref{eq:kkt_stat}–\eqref{eq:kkt_comp} with respect to $r$ over any subinterval where $\mathcal{A}(r)$ is fixed. With $E\equiv E(r)$ and dropping the explicit $r$ for brevity,
\begin{equation}
\label{eq:linKKT}
\begin{bmatrix}
H & G_{\mathcal{A}}^\top & E^\top\\
G_{\mathcal{A}} & 0 & 0\\
E & 0 & 0
\end{bmatrix}
\begin{bmatrix}
\dfrac{d\pi^*}{dr}\\[1mm] \dfrac{d\lambda^*_{\mathcal{A}}}{dr}\\[1mm] \dfrac{d\mu^*}{dr}
\end{bmatrix}
=
\begin{bmatrix}
0\\[1mm] -\,b_{\mathcal{A}}\\[1mm] 0
\end{bmatrix}.
\end{equation}
Nonsingularity of the KKT matrix (strong regularity) ensures a unique solution to \eqref{eq:linKKT}.

\paragraph{Sensitivity formula.}
Eliminate multipliers in \eqref{eq:linKKT} via the Schur complement. If no equalities are present ($E=0$), we obtain the classical expression
\begin{equation}
\label{eq:pi_prime}
\frac{d\pi^*(r)}{dr}
\;=\;
-\,H(r)^{-1}\,G_{\mathcal{A}}(r)^\top
\Big(G_{\mathcal{A}}(r)\,H(r)^{-1}\,G_{\mathcal{A}}(r)^\top\Big)^{-1}\,b_{\mathcal{A}}(r).
\end{equation}
With equalities, replace $G_{\mathcal{A}}$ by $W(r):=\begin{bmatrix}G_{\mathcal{A}}(r)\\[0.2em]E(r)\end{bmatrix}$ and $b_{\mathcal{A}}$ by $c_{\mathcal{A}}(r):=\begin{bmatrix}b_{\mathcal{A}}(r)\\[0.2em]0\end{bmatrix}$ in \eqref{eq:pi_prime}.

\paragraph{Uniform bound and Lipschitz continuity.}
Define
\begin{equation}
\label{eq:LU_def}
L_U \;:=\; \sup_{r\in\mathcal{R}}
\left\|
H(r)^{-1}G_{\mathcal{A}}(r)^\top
\Big(G_{\mathcal{A}}(r)H(r)^{-1}G_{\mathcal{A}}(r)^\top\Big)^{-1}
b_{\mathcal{A}}(r)
\right\|_\infty.
\end{equation}
Continuity of the matrices in \eqref{eq:H}–\eqref{eq:bA} and uniform nonsingularity on each fixed-active-set region imply the supremum in \eqref{eq:LU_def} is finite. The mean-value theorem then yields \eqref{eq:planner_lip}.
\end{proof}

\begin{remark}[Radius-dependent data]
\label{rem:r_dep}
If $h(\pi,r)$ depends on $r$ or if $\nabla_\pi a_k(\pi,r)$ depends on $r$, the right-hand side of \eqref{eq:linKKT} becomes $\big[\,g^\top,\,-b_{\mathcal{A}}^\top,\,-d_h^\top\big]^\top$, where
$g:=\partial_r\big(\nabla_\pi J + \sum_{k\in\mathcal{A}}\lambda_k^*\nabla_\pi a_k + E^\top \mu^*\big)$
and $d_h:=\partial_r h(\pi^*,r)$. The proof proceeds identically, and the Lipschitz property \eqref{eq:planner_lip} continues to hold under strong regularity.
\end{remark}

\subsection{Proof of Theorem~\ref{thm:per_episode_safety}}\label{app:thm-pe-proof}

Let $y_{j+1,0:T}$ and $y_{j,0:T}$ be trajectories generated by the control policies $\pi_{j+1}$ and $\pi_j$ in episodes $j+1$ and $j$, respectively. Recall from our previous discussion and Lemma \ref{prop:coupling_sensitivity} that 
\begin{align}
    \|y_{j+1,0:T} - y_{j,0:T}\|_\infty \;\le\; \beta_T\,\|\pi_{j+1}-\pi_j\|_\infty.
\end{align}
The nonconformity score $s$ defined in \eqref{eq:score-tube} has Lipschitz constant one so that we have
\begin{align*}
  |s(\hat y_{0:T},y_{j,0:T})-s(\hat y_{0:T},y_{j+1,0:T})|\le \beta_T \, \|\pi_{j+1} - \pi_j\|_\infty.
\end{align*}
Next, applying adversarial CP from Lemma \ref{lemma_ccp} immediately results in the guarantee
\begin{align*}
\mathbb{P}_{n_j}
\big\{\,\mathbb{P}
\big\{\,
s(\hat y_{0:T}, y_{j+1,0:T})
\;\le\;
q_j + \beta_T \, \|\pi_{j+1} - \pi_j\|_\infty
\big\}
\;\ge\;
1-\alpha\big\} \ge 1-\delta_j.
\end{align*}
 Since the implicit safety requirement in \eqref{eq:implicit_problem} holds by assumption, we know that $r_{j+1} \ge {q}_j + \beta_T \norm{\pi^\star(r_{j+1}) - \pi_j}_{\infty}$ is satisfied. This in turn implies that 
\begin{equation}\label{eq:safeeeeee}
\mathbb{P}_{n_j}\!\left\{\,\mathbb{P}\!\left\{\,s(\hat y_{0:T}, y_{j+1,0:T}) \le r_{j+1}\,\right\} \ge 1-\alpha\,\right\} \;\ge\; 1-\delta_j,
\end{equation} 
which is the first claim of Theorem~\ref{thm:per_episode_safety}.

Assume now that the optimization problem $\mathbf{P}[j{+}1;\,r_{j+1}]$ in \eqref{eq:episodic} is feasible. This means that 
\begin{align}\label{eq:safee}
    H(x_{j+1,0:T},\zeta)\le 0\quad \forall\,\zeta\in \mathcal{C}_{r_{j+1}}(\hat y_{0:T})
\end{align}
where we recall that the set $\mathcal{C}_{r_{j+1}}(\hat y_{0:T})$ is defined as $\mathcal{C}_{r_{j+1}}(\hat y_{0:T})\ :=\ \big\{\,y_{0:T}:\ s(\hat y_{0:T},y_{0:T})\le r_{j+1}\,\big\}$. Combining  \eqref{eq:safeeeeee} and \eqref{eq:safee} implies that 
\begin{align*}
\mathbb{P}_{n_j}\!\left\{\,\mathbb{P}\!\left\{\,H(x_{j+1,0:T},y_{j+1,0:T})\le 0\,\right\} \ge 1-\alpha \right\} \;\ge\; 1-\delta_j,
\end{align*}
which is the second claim of Theorem~\ref{thm:per_episode_safety}.

\subsection[Sufficient Condition for Feasibility of P(j+1; r) in Eq. (episodic)]{Sufficient Condition for Feasibility of \texorpdfstring{$\mathbf{P}[j{+}1;\,r_{j+1}]$}{P[j+1; r_{j+1}]} in \texorpdfstring{\eqref{eq:episodic}}{Eq. (episodic)}}
\label{app:sufficiency}

First note that feasibility of the optimization problem $\mathbf{P}[j;\,r_{j}]$ in \eqref{eq:episodic} is guaranteed if the constraint
\begin{align}\label{eq:safee__}
    H(x_{j,0:T},\zeta)\le 0\quad \forall\,\zeta\in \mathcal{C}_{r_{j}}(\hat y_{0:T})
\end{align}
is satisfied by the trajectory $x_{j,0:T}$, where we recall that $\mathcal C_{r_{j}}(\hat y_{0:T})=\{\,y_{0:T}:\ \|y_{0:T}-\hat y_{0:T}\|_\infty\le r_{j}\,\}$. 

Assume now that there exists a Lipschitz constant $L_H\ge 0$ such that 
\begin{align*}
  \big|\,H(x_{0:T},y_{0:T}) - H(x_{0:T},y'_{0:T})\,\big|
  \;\le\; L_H\,\|y_{0:T}-y'_{0:T}\|_\infty
\end{align*}
for all permissible $x_{0:T}$ and all $y_{0:T},y'_{0:T}$.
Consequently, if the trajectory $x_{j,0:T}$ satisfies
\begin{align*}
  H(x_{j,0:T},\hat y_{0:T}) \;\le\; -\,L_H\,r_{j},
\end{align*}
then, for any $y_{0:T}\in\mathcal C_{r_{j}}(\hat y_{0:T})$, we know that
\begin{align*}
  H(x_{j,0:T},y_{0:T})
  &\le H(x_{j,0:T},\hat y_{0:T}) + \big|H(x_{j,0:T},y_{0:T})-H(x_{j,0:T},\hat y_{0:T})\big|
  \nonumber\\
  &\le -\,L_H r_{j} + L_H\,\|y_{0:T}-\hat y_{0:T}\|_\infty
  \;\le\; -\,L_H r_{j} + L_H r_{j}
  \;=\; 0,
\end{align*}
by which we have shown that \eqref{eq:safee__} is satisfied so that $\mathbf{P}[j;\,r_{j}]$  is feasible.

\subsection{Proof of Lemma \ref{thm:explicit_update_rule}}
\label{app:proof_explicit_update}
Our goal  is to derive the solution to the tractable inequality in \eqref{eq:r_update}, which we recall for convenience:
\begin{equation}
  r_{j+1} \;\ge\; q_j \;+\; \kappa\,\big|r_{j+1}-r_j\big|.
  \label{eq:explicit_update_start}
\end{equation}
By assumption, we have that $\kappa<  1$. Since the right-hand side of~\eqref{eq:explicit_update_start} is nondecreasing in $r_{j+1}$, the minimal solution is obtained by saturating~\eqref{eq:explicit_update_start}. Due to the absolute value in $\lvert r_{j+1}-r_j\rvert$ in \eqref{eq:explicit_update_start}, we have to consider the two cases of shrinkage ($r_{j+1}\le r_j$) and expansion ($r_{j+1}> r_j$).

\paragraph{Case 1: Shrinkage ($r_{j+1}\le r_j$).}
We have that $\lvert r_{j+1}-r_j\rvert=r_j-r_{j+1}$ so that ~\eqref{eq:explicit_update_start} becomes
\begin{align*}
  r_{j+1} \;\ge\; q_j \;+\; \kappa\,(r_j-r_{j+1})
  \;\Longleftrightarrow\;
  (1+\kappa)\,r_{j+1} \;\ge\; q_j + \kappa r_j.
\end{align*}
Since $1+\kappa>0$, the minimal value of $r_{j+1} $ saturates the inequality, i.e.,
\begin{equation}
  r_{j+1} \;=\; \frac{q_j+\kappa r_j}{1+\kappa}.
  \label{eq:shrink_solution}
\end{equation}
Consistency with the shrinkage assumption requires
\begin{align*}
  \frac{q_j+\kappa r_j}{1+\kappa}\;\le\; r_j
  \;\Longleftrightarrow\;
  q_j \;\le\; r_j.
\end{align*}

\paragraph{Case 2: Expansion ($r_{j+1}> r_j$).}
We have that $\lvert r_{j+1}-r_j\rvert=r_{j+1}-r_j$ so that~\eqref{eq:explicit_update_start} becomes
\begin{align*}
  r_{j+1} \;\ge\; q_j \;+\; \kappa\,(r_{j+1}-r_j)
  \;\Longleftrightarrow\;
  (1-\kappa)\,r_{j+1} \;\ge\; q_j - \kappa r_j.
\end{align*}
Since $1-\kappa>0$, the minimal value of $r_{j+1} $ saturates the inequality, i.e.,
\begin{equation}
  r_{j+1} \;=\; \frac{q_j-\kappa r_j}{1-\kappa}.
  \label{eq:expand_solution}
\end{equation}
Consistency with the expansion assumption requires
\begin{align*}
  \frac{q_j-\kappa r_j}{1-\kappa}\;>\; r_j
  \;\Longleftrightarrow\;
  q_j \;>\; r_j.
\end{align*}

In conclusion, the above condition partition the line by $q_j\le r_j$ and $q_j>r_j$. Combining this with ~\eqref{eq:shrink_solution} and~\eqref{eq:expand_solution} yields the closed-form solution in \eqref{r_update_explicit}, which was to be shown.

\subsection{Proof of Theorem \ref{thm:stability_shrinkage}}
\label{app:proof_stability_shrinkage}
We analyze the magnitude $\abs{r_{j+1} - r_j}$ in the same two cases considered in Theorem \ref{thm:explicit_update_rule}.
\begin{itemize}
    \item \textbf{Case 1: Shrinkage (${q}_j \le r_j$).} We can derive that
    \begin{equation}\label{eq:addd}
    \abs{r_{j+1} - r_j} = r_j - r_{j+1} = r_j - \left(\frac{{q}_j + \kappa r_j}{1+\kappa}\right) = \frac{r_j(1+\kappa) - {q}_j - \kappa r_j}{1+\kappa} = \frac{r_j - {q}_j}{1+\kappa} = \frac{\abs{{q}_j - r_j}}{1+\kappa}.
    \end{equation}
    
    \item \textbf{Case 2: Expansion (${q}_j > r_j$).} We can derive that 
    \begin{align*}
    \abs{r_{j+1} - r_j} = r_{j+1} - r_j = \left(\frac{{q}_j - \kappa r_j}{1-\kappa}\right) - r_j = \frac{{q}_j - \kappa r_j - r_j(1-\kappa)}{1-\kappa} = \frac{{q}_j - r_j}{1-\kappa} = \frac{\abs{{q}_j - r_j}}{1-\kappa}.
    \end{align*}
\end{itemize}
Since $\kappa < 1$ while being positive, we have $1-\kappa < 1+\kappa$. The denominator is smaller in the expansion case, so the worst-case bound for $\abs{r_{j+1} - r_j}$ is given as
\begin{equation}
\abs{r_{j+1} - r_j} \le \frac{1}{1-\kappa} \abs{{q}_j - r_j},
\end{equation}
which was to be proven. Finally, in the case that ${q}_j < r_j$, we see that the numerator in \eqref{eq:addd} is positive, and since $1+\kappa > 0$, the change is positive, meaning $r_{j+1}< r_j $.

\subsection{Proof of Theorem~\ref{thm:convergence-main}}\label{app:thm8-proof}

\noindent\textbf{Step 1: Setup and notation.} For the convenience of the reader, we first recall our notation. Let
\begin{align*}
  T(r) \;=\; Q_{1-\alpha}\!\big(\pi^\star(r)\big)
\end{align*}
denote the $(1-\alpha)$ true {population} quantile under policy $\pi^\star(r)$. We assume the Lipschitz property
\begin{align*}
  \big|T(r)-T(r')\big| \;\le\; \kappa\,|r-r'|
  \qquad\text{for all }r,r'\in \mathbb{R}_{\ge 0},
\end{align*}
where $\kappa=\beta_T L_U\in(0,1)$ is the closed-loop gain. Suppose there exists a fixed point $r^\star\in \mathbb{R}_{\ge 0}$ with
\begin{align*}
  T(r^\star)\;=\;r^\star.
\end{align*}
Let the  error and calibration perturbation at episode $j$ be
\begin{equation}\label{eq:e-eta-def-app}
  e_j \;=\; |r_j - r^\star|,
  \qquad
  \eta_j \;=\; q_j - T(r_j),
\end{equation}
where $q_j$ is the empirical $(1-\bar\alpha_j)$ quantile formed with
\begin{align*}
  \bar\alpha_j \;=\; \alpha \;-\; \sqrt{\frac{\ln(1/\delta_j)}{2\,n_j}} \in (0,\alpha).
\end{align*}
%Here $\delta\in(0,1)$ is the {level-tightening} parameter appearing inside $\bar\alpha_j$; see the main text. We will use a {separate} symbol $\delta_j$ below for per-episode tail probabilities in high-probability bounds.

\medskip
\noindent\textbf{Step 2: Branch-wise identities induced by the explicit update.}
Write $q_j=T(r_j)+\eta_j$ by \eqref{eq:e-eta-def-app}. The explicit update from Theorem~\ref{thm:explicit_update_rule} gives us  the following expressions:
\begin{align}
  \text{if } q_j\le r_j\ \text{(shrinkage):}\quad
  (1+\kappa)\,\big(r_{j+1}-r^\star\big)
  &= \big(T(r_j)-T(r^\star)\big) + \kappa\,(r_j-r^\star) + \eta_j,
  \label{eq:branch-shrink-balance-app}
  \\
  \text{if } q_j> r_j\ \text{(expansion):}\quad
  (1-\kappa)\,\big(r_{j+1}-r^\star\big)
  &= \big(T(r_j)-T(r^\star)\big) - \kappa\,(r_j-r^\star) + \eta_j.
  \label{eq:branch-expand-balance-app}
\end{align}
\emph{Derivation.} In the shrinkage case, in order to obtain \eqref{eq:branch-shrink-balance-app}, we compute $r_{j+1}-r^\star$ while substituting $r_{j+1}=\frac{q_j+\kappa r_j}{1+\kappa}$ and $q_j=T(r_j)+\eta_j$ so that
\[
(1+\kappa)(r_{j+1}-r^\star)=q_j+\kappa r_j-(1+\kappa)r^\star
=T(r_j)-T(r^\star)+\kappa(r_j-r^\star)+\eta_j.
\]
The expansion case is identical with $r_{j+1}=\frac{q_j-\kappa r_j}{1-\kappa}$.

\medskip
\noindent\textbf{Step 3: One-step bounds and contraction threshold.}
Taking absolute values in \eqref{eq:branch-shrink-balance-app}–\eqref{eq:branch-expand-balance-app} and using the Lipschitz property
\begin{align*}
  \big|T(r_j)-T(r^\star)\big|\;\le\;\kappa\,|r_j-r^\star|\;=\;\kappa\,e_j,
\end{align*}
we obtain
\begin{align*}
  |r_{j+1}-r^\star|
  &\le \frac{1}{1+\kappa}\Big(\,|T(r_j)-T(r^\star)|+\kappa e_j+|\eta_j|\,\Big)
   \;\le\; \frac{2\kappa}{1+\kappa}\,e_j + \frac{1}{1+\kappa}\,|\eta_j|,\\
  |r_{j+1}-r^\star|
  &\le \frac{1}{1-\kappa}\Big(\,|T(r_j)-T(r^\star)|+\kappa e_j+|\eta_j|\,\Big)
   \;\le\; \frac{2\kappa}{1-\kappa}\,e_j + \frac{1}{1-\kappa}\,|\eta_j|.
\end{align*}
The coefficients in the expansion case dominate those coefficients in the shrinkage case so for that for {all} episodes $j$ we obtain the one-step recursion
\begin{equation}\label{eq:uniform-recursion-app}
  e_{j+1} \;\le\; \gamma_\kappa\,e_j \;+\; B_\kappa\,|\eta_j|,
  \qquad
  \gamma_\kappa \;=\; \frac{2\kappa}{1-\kappa},\quad
  B_\kappa \;=\; \frac{1}{1-\kappa}.
\end{equation}
 We specifically note that
\begin{equation}\label{eq:gamma-cond-app}
  \gamma_\kappa < 1
  \iff \frac{2\kappa}{1-\kappa}<1
  \iff 3\kappa<1
  \iff \kappa<\tfrac13.
\end{equation}

\medskip
\noindent\textbf{Step 4: Explicit finite-horizon error bound.}
We now unroll \eqref{eq:uniform-recursion-app} over multiple episodes to obtain~\ref{thm:conv:P1} of Theorem~\ref{thm:convergence-main}. For convenience, we recall this bound as
\begin{equation}\label{eq:explicit-sum-app}
  e_{j+1}
  \;\le\;
  \gamma_\kappa^{\,j+1}\,e_0
  \;+\;
  B_\kappa \sum_{m=0}^{j} \gamma_\kappa^{\,j-m}\,|\eta_m|.
\end{equation}

\underline{Base case $j=0$.} From \eqref{eq:uniform-recursion-app}, $e_1\le \gamma_\kappa e_0+B_\kappa|\eta_0|$, which matches \eqref{eq:explicit-sum-app} with $j=0$.

\underline{Induction step.} Assume \eqref{eq:explicit-sum-app} holds for some $j\ge 0$. Then
\[
e_{j+2} \;\le\; \gamma_\kappa e_{j+1} + B_\kappa|\eta_{j+1}|
\;\le\; \gamma_\kappa\!\left(\gamma_\kappa^{\,j+1}e_0 + B_\kappa\sum_{m=0}^{j}\gamma_\kappa^{\,j-m}|\eta_m|\right)
      + B_\kappa|\eta_{j+1}|,
\]
which simplifies to
\[
e_{j+2} \;\le\; \gamma_\kappa^{\,j+2}e_0 \;+\; B_\kappa\sum_{m=0}^{j+1}\gamma_\kappa^{\,j+1-m}|\eta_m|,
\]
i.e., \eqref{eq:explicit-sum-app} with $j\leftarrow j+1$. This completes the induction.

\medskip
\noindent\textbf{Step 5: Closed-form bounds.}
If $|\eta_m|\le C$ for all $m\in\{0,1,\dots,j\}$, then applying \eqref{eq:explicit-sum-app} and summing the geometric series gives
\begin{align}
  e_{j+1}
  &\le \gamma_\kappa^{\,j+1}e_0
       + B_\kappa C\sum_{m=0}^{j}\gamma_\kappa^{\,j-m}
   \;=\; \gamma_\kappa^{\,j+1}e_0 + B_\kappa C\sum_{\ell=0}^{j}\gamma_\kappa^{\,\ell}
   \nonumber\\
  &= \gamma_\kappa^{\,j+1}e_0 \;+\; B_\kappa C\,\frac{1-\gamma_\kappa^{\,j+1}}{1-\gamma_\kappa}.
  \label{eq:closed-form-const-eta-app}
\end{align}
If moreover $\kappa<\tfrac13$ (so that $\gamma_\kappa<1$), then letting $j\to\infty$ in \eqref{eq:closed-form-const-eta-app} yields the steady-state bound
\begin{align*}
  \limsup_{j\to\infty} e_j
  \;\le\; \frac{B_\kappa}{1-\gamma_\kappa}\,C
  \;=\; \frac{1}{1-3\kappa}\,C.
\end{align*}
This proves~\ref{thm:conv:P2} of Theorem~\ref{thm:convergence-main}.

\medskip
\noindent\textbf{Step 6: High-probability control of \(\eta_j\) (level shift + empirical error).}
First, decompose
\begin{align*}
  \eta_j
  \;=\;
  \underbrace{Q_{1-\bar\alpha_j}\!\big(\pi^\star(r_j)\big) - Q_{1-\alpha}\!\big(\pi^\star(r_j)\big)}_{\Delta^{\mathrm{lvl}}_j}
  \;+\;
  \underbrace{q_j - Q_{1-\bar\alpha_j}\!\big(\pi^\star(r_j)\big)}_{\varepsilon^{\mathrm{est}}_j}
\end{align*}
into a level shift and empirical error component.

By assumption, the CDF $F_{r_j}(s)$ is differentiable and has density no less than $f_\star>0$  in a sufficiently large neighborhood of its $(1-\alpha)$ quantile $Q_{1-\alpha}\!\big(\pi^\star(r_j)\big)$. Then, by Lemma \ref{lem:level-lip} which we separately present in Appendix \ref{app:quant-lemmas}, the inverse-CDF is $1/f_\star$-Lipschitz in the probability level, and therefore 
\begin{equation}\label{eq:level-shift-app}
  \big|\Delta^{\mathrm{lvl}}_j\big| \;\le\; \frac{|\alpha-\bar\alpha_j|}{f_\star}.
\end{equation}
Next, Lemma \ref{lem:emp-quant} which we separately present in Appendix \ref{app:quant-lemmas} and the inverse-CDF Lipschitz property  imply
\begin{equation}\label{eq:emp-to-quantile-app}
  \mathbb{P}_{n_j}\!\left\{\big|\varepsilon^{\mathrm{est}}_j\big|
  \;\le\; \frac{1}{f_\star}\,\sqrt{\frac{\ln(2/\delta_j)}{2 n_j}}\right\}
  \;\ge\; 1-\delta_j.
\end{equation}
where $\delta_j\in(0,1)$ is any prescribed failure probability.

Combining \eqref{eq:level-shift-app}–\eqref{eq:emp-to-quantile-app} gives the per-episode high-probability bound
\begin{equation}\label{eq:eta-bound-per-episode}
  \mathbb{P}_{n_j}\!\left\{|\eta_j|
  \;\le\;
  \frac{|\alpha-\bar\alpha_j|}{f_\star}
  \;+\;
  \frac{1}{f_\star}\sqrt{\frac{\ln(2/\delta_j)}{2 n_j}}\right\}
  \;\ge\; 1-\delta_j,
\end{equation}
which proves the per-episode bound in~\ref{thm:conv:P3} of Theorem~\ref{thm:convergence-main}.

Finally, applying equation \eqref{eq:explicit-sum-app} and a union bounding argument over the first $j{+}1$ episodes, we obtain the explicit finite-horizon control
\begin{align*}
  \mathbb{P}_{\sum_{m=0}^{j}n_m}\!\left\{e_{j+1}
  \;\le\;
  \gamma_\kappa^{\,j+1} e_0
  \;+\;
  \frac{B_\kappa}{f_\star}\sum_{m=0}^{j}\gamma_\kappa^{\,j-m}\,|\alpha-\bar\alpha_m|
  \;+\;
  \frac{B_\kappa}{f_\star}\sum_{m=0}^{j}\gamma_\kappa^{\,j-m}\sqrt{\frac{\ln(2/\delta_m)}{2 n_m}}\right\}
  \;\ge\; 1-\sum_{m=0}^{j}\delta_m.
\end{align*}
This shows the joint bound in~\ref{thm:conv:P3} of Theorem~\ref{thm:convergence-main}.
% Here the sequence $\{\delta_m\}$ is {independent of} the level-tightening parameter $\delta$ used in $\bar\alpha_m$; the former controls empirical quantile fluctuations, the latter sets the calibration level.

\noindent\textbf{Step 7: Asymptotic conclusions.} By \eqref{eq:eta-bound-per-episode}, $|\eta_j|\le A_j$ holds with $\mathbb{P}_{n_j}$-probability at least $1-\delta_j$. If $n_j\to\infty$, $\bar\alpha_j\to\alpha$, and $\delta_j\to 0$, then $A_j\to 0$, so $\mathbb{P}_{n_j}\{|\eta_j|\to 0\}\ge 1-\delta_j$. Since calibration samples at distinct episodes are independent, a union bound under the product measure $\mathbb{P}_{\sum_{m=0}^{\infty}n_m}$ gives $\mathbb{P}_{\sum_{m=0}^{\infty}n_m}\{|\eta_j|\to 0\}\ge 1-\sum_{m=0}^{\infty}\delta_m$. If additionally $\kappa<\tfrac{1}{3}$, then $\gamma_\kappa<1$ by \eqref{eq:gamma-cond-app} and \eqref{eq:explicit-sum-app} gives $e_j\to 0$ on the same event, so $\mathbb{P}_{\sum_{m=0}^{\infty}n_m}\{e_j\to 0\}\ge 1-\sum_{m=0}^{\infty}\delta_m$.

If instead $\bar\alpha_j\equiv\bar\alpha<\alpha$, $n_j\to\infty$, and $\delta_j\to 0$, then $A_j\to(\alpha-\bar\alpha)/f_\star$, so $\mathbb{P}_{n_j}\{|\eta_j|\to(\alpha-\bar\alpha)/f_\star\}\ge 1-\delta_j$. Using \eqref{eq:closed-form-const-eta-app} with $C=(\alpha-\bar\alpha)/f_\star$ and the same union bound yields
\begin{align*}
\mathbb{P}_{\sum_{m=0}^{\infty}n_m}\!\Big\{\limsup_{j\to\infty} e_j\le \tfrac{1}{1-3\kappa}\cdot\tfrac{\alpha-\bar\alpha}{f_\star}\Big\}\;\ge\;1-\sum_{m=0}^{\infty}\delta_m.
\end{align*}
This completes the proof of Theorem~\ref{thm:convergence-main}.

\subsection{Quantile perturbation lemmas (expanded and annotated)}\label{app:quant-lemmas}

% ---------- Preliminaries (notation and definitions) ----------
\paragraph{Setup and notation.}
Recall that $F_r:\mathbb R\to[0,1]$ denotes the cumulative distribution function (CDF) of the nonconformity score
$s(\hat y_{0:T},Y_{0:T})$ for the trajectory $y_{0:T}\sim\mathcal D(\pi^\star(r))$. We have also defined the right-continuous quantile map as
\[
Q_{p}(\pi^\star(r)) \;\coloneqq\; \inf\{\,s\in\mathbb R:\ F_r(s)\ge p\,\},
\qquad \alpha\in(0,1).
\]

Assume that the CDF $F_{r}(s)$ is differentiable and has density $f_{r}(s)$ no less than $f_\star>0$ for all $s$ in a sufficiently large neighborhood of its $p$ quantile $Q_{p}\!\big(\pi^\star(r)\big)$, i.e., 
%Assume that in a neighborhood of the target quantile $Q_{p}(\pi^\star(r))$ the CDF $F_r$ is differentiable with density $f_r=\frac{dF_r}{dt}$ and that this density is bounded below by a uniform constant
\begin{equation}\label{eq:uniform-density}
  f_r(s)\ \ge\ f_\star\ >\ 0
  \qquad
  \text{for all $s$ in a sufficiently large neighborhood of $Q_{p}\!\big(\pi^\star(r)\big)$.}
\end{equation}
Condition \eqref{eq:uniform-density} guarantees that the inverse-CDF is {locally} Lipschitz; geometrically, the CDF has a slope bounded away from $0$ near the quantile, so small changes in {probability level} produce controlled changes in the {quantile value}.

% ---------- Lemma 1 ----------
\begin{lemma}[Level-to-quantile Lipschitzness]\label{lem:level-lip}
For any two probability levels $p,p'\in (0,1)$ and for any $Q_{p'}(\pi^\star(r))$ in the neighborhood of  $Q_{p}(\pi^\star(r))$ from \eqref{eq:uniform-density}, we have that
\begin{equation}\label{eq:level-lip}
  \big|\,Q_{p}(\pi^\star(r)) - Q_{p'}(\pi^\star(r))\,\big|
  \;\le\; \frac{|p-p'|}{f_\star}.
\end{equation}
\begin{proof}
\textbf{1) Fix the two quantiles.} Let
\[
q \;=\; Q_p(\pi^\star(r)),\qquad q' \;=\; Q_{p'}(\pi^\star(r)).
\]
By definition of $Q_p$, we have $F_r(q)\ge p$ and, by right continuity and strict monotonicity in our neighborhood, we can take $F_r(q)=p$ and $F_r(q')=p'$.

\smallskip
\textbf{2) Apply the Mean Value Theorem (MVT).} Since $F_r$ is differentiable on the open interval between $q$ and $q'$, there exists a point $\xi$ between $q$ and $q'$ such that
\[
F_r(q')-F_r(q) \;=\; f_r(\xi)\,(q'-q).
\]

\smallskip
\textbf{3) Use the density lower bound.} Because $f_r(\xi)\ge f_\star$ by \eqref{eq:uniform-density}, we get
\[
|p'-p| \;=\; |F_r(q')-F_r(q)| \;=\; f_r(\xi)\,|q'-q|
\;\ge\; f_\star\,|q'-q|.
\]

\smallskip
\textbf{4) Rearrange.} Hence
\[
|q'-q| \;\le\; \frac{|p'-p|}{f_\star},
\]
which is \eqref{eq:level-lip}.
\end{proof}
\end{lemma}

\paragraph{Why we need Lemma~\ref{lem:level-lip}.}
In the convergence analysis (Theorem~\ref{thm:convergence-main}), we compare the population quantiles at two \emph{levels}, $1-\bar\alpha_j$ (used for calibration) and $1-\alpha$ (the target). Lemma~\ref{lem:level-lip} gives the clean bound
\[
\big|\,Q_{1-\bar\alpha_j}-Q_{1-\alpha}\,\big|
\;\le\; \frac{|\alpha-\bar\alpha_j|}{f_\star},
\]
which is the {level-shift} term in the perturbation $\eta_j$.

% ---------- Lemma 2 ----------
\begin{lemma}[Empirical quantile error via DKW Inequality]\label{lem:emp-quant}
Let $F_{r,n}$ be the empirical CDF from $n$ i.i.d.\ samples drawn from $F_r$, and let $q_{n,p}$ be the empirical $p$-quantile
\(
q_{n,p}=\inf\{t:\ F_{r,n}(t)\ge p\}.
\)
Then, for any outer tail level $\delta\in(0,1)$, we have
\begin{equation}\label{eq:DKW}
  \mathbb P_n\!\left\{\,\sup_{t\in\mathbb R}\big|F_{r,n}(t)-F_r(t)\big|\le \varepsilon_n(\delta)\,\right\}
  \;\ge\; 1-\delta,
  \qquad \varepsilon_n(\delta)\coloneqq \sqrt{\frac{\ln(2/\delta)}{2n}}.
\end{equation}
Under the assumption in equation \eqref{eq:uniform-density}, the empirical quantile satisfies
\begin{equation}\label{eq:empq-bound}
  \mathbb P_n\!\left\{\big|\,q_{n,p} - Q_p(\pi^\star(r))\,\big| \;\le\; \frac{\varepsilon_n(\delta)}{f_\star}\right\}
  \;\ge\; 1-\delta.
\end{equation}
\begin{proof}
\textbf{1) DKW event (uniform CDF control).} By the DKW inequality, the statement in equation \eqref{eq:DKW} immediately follows. 

\smallskip
\textbf{2) Pin the target quantile and a local window.} Let
\[
q^\star \;=\; Q_p(\pi^\star(r)),
\]
and pick a small $\Delta>0$ that keeps $[q^\star-\Delta,\,q^\star+\Delta]$ inside the neighborhood where $f_r\ge f_\star$.

\smallskip
\textbf{3) One-sided controls for the \emph{true} CDF using the density lower bound.}  
By the Mean Value Theorem applied to $F_r$ on $[q^\star,\,q^\star+\Delta]$ and $[q^\star-\Delta,\,q^\star]$ there exist points $\xi_+,\xi_-$ in those intervals with
\[
F_r(q^\star+\Delta)-F_r(q^\star) \;=\; f_r(\xi_+)\,\Delta \;\ge\; f_\star \Delta,
\qquad
F_r(q^\star)-F_r(q^\star-\Delta) \;=\; f_r(\xi_-)\,\Delta \;\ge\; f_\star \Delta.
\]
Since $F_r(q^\star)=p$, we get
\begin{equation}\label{eq:true-CDF-one-sided}
  F_r(q^\star+\Delta) \;\ge\; p + f_\star \Delta,
  \qquad
  F_r(q^\star-\Delta) \;\le\; p - f_\star \Delta.
\end{equation}

\smallskip
\textbf{4) Transfer these inequalities to the \emph{empirical} CDF on the DKW event.}  
Using \eqref{eq:DKW}, we have
\begin{align}
\mathbb P_n\{&F_{r,n}(q^\star+\Delta) \;\ge\; F_r(q^\star+\Delta) - \varepsilon \;\ge\; p + f_\star\Delta - \varepsilon\\
&\text{and} \; F_{r,n}(q^\star-\Delta) \;\le\; F_r(q^\star-\Delta) + \varepsilon \;\le\; p - f_\star\Delta + \varepsilon\}
  \;\ge\; 1-\delta.
\end{align}

\smallskip
\textbf{5) Choose $\Delta$ to make the inequalities straddle level $p$.}  
Set $\Delta \;=\; \frac{\varepsilon}{f_\star}$ so that
\[
\mathbb P_n\!\left\{F_{r,n}(q^\star+\Delta) \;\ge\; p \text{ and } F_{r,n}(q^\star-\Delta) \;\le\; p\right\}
  \;\ge\; 1-\delta.
\]

\smallskip
\textbf{6) Use the empirical quantile definition to trap $q_{n,p}$.}  
By definition, $q_{n,p}=\inf\{t:\ F_{r,n}(t)\ge p\}$.  
Since $F_{r,n}(q^\star-\Delta)\le p$ and $F_{r,n}(q^\star+\Delta)\ge p$, the monotonicity of $F_{r,n}$ implies
\[
\mathbb P_n\!\left\{q_{n,p} \;\in\; [\,q^\star-\Delta,\ q^\star+\Delta\,]\right\}
  \;\ge\; 1-\delta,
\]
where we used the union bound. Therefore,
\[
\mathbb P_n\!\left\{|q_{n,p}-q^\star| \;\le\; \Delta \;=\; \frac{\varepsilon}{f_\star}
\;=\; \frac{\varepsilon_n(\delta)}{f_\star}\right\}
  \;\ge\; 1-\delta.
\]
This is exactly \eqref{eq:empq-bound}.
\end{proof}
\end{lemma}

\paragraph{Why we need Lemma~\ref{lem:emp-quant}.}
In the convergence proof, the empirical quantile $q_j$ is compared to the population quantile at the \emph{same level} $1-\bar\alpha_j$. Lemma~\ref{lem:emp-quant} converts the uniform CDF error (obtained using the DKW inequality) into an error on the quantile value, with the sharp factor \(1/f_\star\). Concretely, we thereby get
\[
\mathbb P_{n_j}\!\left\{\big|\,q_j - Q_{1-\bar\alpha_j}(\pi^\star(r_j))\,\big|
\;\le\;
\frac{1}{f_\star}\sqrt{\frac{\ln(2/\delta_j)}{2\,n_j}}\right\}
  \;\ge\; 1-\delta_j.
\]

\fi
\ifincludeAppPerf   \subsection{Proof of Theorem~\ref{thm:performance_convergence}}
\label{app:proof-performance-convergence}
\subeqgentleon
The first result $|J_{j+1}-J^\star|=|V(r_{j+1})-V(r^\star)|\le L_V\,e_{j+1}$ follows directly from the Lipschitz constant $L_V$ of $V(r)$ and $e_{j+1}=|r_{j+1}-r^\star|$ from Theorem~\ref{thm:convergence-main}.

We next show that $V$ is nondecreasing, which is used in all three parts below. For $r\in\calR$, define the feasible policy set
\begin{align}\label{eq:app:feasible_set}
\mathcal{F}(r):=\big\{\pi:\ H(x_{0:T}(\pi),\zeta)\le 0\ \ \forall\,\zeta\in\mathcal{C}_r(\hat y_{0:T})\big\}.
\end{align}
For $r,r'\in\calR$ with $r'\le r$, we have that $\mathcal{C}_{r'}(\hat y_{0:T})\subseteq\mathcal{C}_{r}(\hat y_{0:T})$, so $\mathcal{F}(r)\subseteq\mathcal{F}(r')$. Taking infima of the same objective function over nested sets yields
\begin{align}\label{eq:app:V_monotone}
V(r')=\inf_{\pi\in\mathcal{F}(r')}J\big(x_{0:T}(\pi),\pi\big)\;\le\;\inf_{\pi\in\mathcal{F}(r)}J\big(x_{0:T}(\pi),\pi\big)=V(r).
\end{align}

\paragraph{Part~(\ref{thm:perf:P1}): Improvement over the initial policy.}
Fix any $r_j\in[r^\star,r_0)$ with $r_0>r^\star$. Since $r_j<r_0$, the monotonicity~\eqref{eq:app:V_monotone} gives $J_j=V(r_j)\le V(r_0)=J_0$, with $J_j=V(r_j)< V(r_0)=J_0$ when the function $V(r)$ is strictly increasing on $[r^\star,r_0]$.

Next, fix any $r_j\le r_0$. Assume additionally that there exists $m_V>0$ such that $V(r)-V(r')\ge m_V(r-r')$ for  all $r, r'\in[r^\star,r_0]$.  If $r_j\ge r^\star$, then $r_j\in[r^\star,r_0]$  so that we directly obtain
\begin{align*}
J_0-J_j\;=\;V(r_0)-V(r_j)\;\ge\;m_V(r_0-r_j)\;=\;m_V(r_0-r^\star-e_j)
\end{align*}
where we recall that $e_j=r_j-r^\star$. If $r_j<r^\star$, then $e_j=r^\star-r_j\ge 0$ and the monotonicity property~\eqref{eq:app:V_monotone} implies $V(r_j)\le V(r^\star)$, so that we can derive
\begin{align*}
J_0-J_j\;=\;V(r_0)-V(r_j)\;\ge\;V(r_0)-V(r^\star)\;\ge\;m_V(r_0-r^\star)\;\ge\;m_V(r_0-r^\star-e_j).
\end{align*}
In both cases $J_0-J_j\ge m_V(r_0-r^\star-e_j)$.

\paragraph{Part~(\ref{thm:perf:P2}): One-step improvement.}
Recall from Theorem~\ref{thm:convergence-main} that $\eta_j=q_j-T(r_j)$ where $T(r):=Q_{1-\alpha}(\pi^\star(r))$ satisfies $T(r^\star)=r^\star$ and has Lipschitz constant $\kappa$. For $r_j\ge r^\star$, this Lipschitz property $T(r)$ gives us that
\begin{align}\label{eq:app:Trj_upper}
T(r_j)\;\le\; T(r^\star)+\kappa\,(r_j-r^\star)\;=\;r^\star+\kappa\,(r_j-r^\star).
\end{align}
By the definition of $\eta_j$ and~\eqref{eq:app:Trj_upper}, we have $q_j=T(r_j)+\eta_j\le r^\star+\kappa(r_j-r^\star)+|\eta_j|$. Under the condition $|\eta_j|<(1-\kappa)(r_j-r^\star)$ assumed in~(\ref{thm:perf:P2}), this yields
\begin{align*}
q_j\;<\; r^\star+\kappa(r_j-r^\star)+(1-\kappa)(r_j-r^\star)\;=\;r_j.
\end{align*}
Since $q_j<r_j$, the update rule in~\eqref{r_update_explicit} is on the shrinkage branch, and Theorem~\ref{thm:stability_shrinkage} (Episode-to-Episode Stability and Shrinkage) gives us that $r_{j+1}<r_j$. The monotonicity property~\eqref{eq:app:V_monotone} then yields $J_{j+1}=V(r_{j+1})\le V(r_j)=J_j$.

For the last statement in \ref{thm:perf:P2}, let $\mathcal{E}_j:=\{J_{j+1}\le J_j\}$. By Theorem~\ref{thm:convergence-main}(\ref{thm:conv:P3}), $\mathbb{P}_{n_j}\{|\eta_j|\le A_j\}\ge 1-\delta_j$. On this event, the sufficient condition $A_j<(1-\kappa)(r_j-r^\star)$ from~(\ref{thm:perf:P2}) implies the deterministic condition above, so $\mathbb{P}_{n_j}(\mathcal{E}_j)\ge 1-\delta_j$. Since calibration samples across episodes are independent, a union bound argument yields
\begin{align*}
\mathbb{P}_{\sum_{j=J}^{K}n_j}\!\Big\{\textstyle\bigcap_{j=J}^{K}\mathcal{E}_j\Big\}\;\ge\;1-\sum_{j=J}^{K}\delta_j.
\end{align*}

\paragraph{Part~(\ref{thm:perf:P3}): Asymptotics.}
By Theorem~\ref{thm:convergence-main}(\ref{thm:conv:P4}), $e_j\to 0$ with probability at least $1-\sum_{m=0}^{\infty}\delta_m$. On this event, $|J_j-J^\star|\le L_V\,e_j\to 0$ by the Lipschitz constant $L_V$ of $V$, hence
\begin{align*}
\mathbb{P}_{\sum_{m=0}^{\infty}n_m}\{J_j\to J^\star\}\;\ge\;1-\sum_{m=0}^{\infty}\delta_m.
\end{align*}
For the second claim, fix $r_0>r^\star$. On the same event, $|r_j-r^\star|\to 0$, so $r_j\to r^\star$. Since $r^\star<r_0$, there exists $j_0$ such that $r_j<r_0$ for all $j\ge j_0$. The monotonicity~\eqref{eq:app:V_monotone} then gives $J_j=V(r_j)\le V(r_0)=J_0$ for all $j\ge j_0$, hence
\begin{align*}
\mathbb{P}_{\sum_{m=0}^{\infty}n_m}\{\exists\,j_0\ \text{s.t.}\ J_j\le J_0\ ,\forall\,j\ge j_0\}\;\ge\;1-\sum_{m=0}^{\infty}\delta_m,
\end{align*}
which completes the proof.
          \fi
\ifincludeAppPolicy 
\section{The Iterative Policy Update Algorithm}
\label{sec:iterative_algorithm}
We summarize the explicit and implicit solver in Algorithm \ref{alg:iterative}.
\begin{algorithm}[h!]
\caption{Iterative Safe Policy Improvement (Explicit \& Implicit Forms)}
\label{alg:iterative}
\begin{algorithmic}[1]
\STATE \textbf{Input:} confidence levels $1-\alpha$, $1-\delta$; initial safe policy $\pi_0$.
\STATE \textbf{Input:} closed-loop gain $\kappa=\beta_TL_U$; fixed predictor $\hat y_{0:T}$.
\STATE \textbf{Input:} solver choice: $\texttt{SOLVER\_TYPE}\in\{\text{'EXPLICIT'},\text{'IMPLICIT'}\}$.
\STATE \emph{Initialize:} choose $r_0\in\calR$; set $\pi_0\gets\pi^\star(r_0)$.
\FOR{$j=0,1,2,\ldots$}
\STATE Execute $\pi_j$; collect rollouts $\{y^{(i)}_{j,0:T}\}_{i=1}^{n_j}$.
\STATE Calibration level with $\delta$ inside: $\bar\alpha_j\gets\alpha-\sqrt{\ln(1/\delta)/(2n_j)}$, so that $\bar\alpha_j\in(0,\alpha)$.
\STATE Empirical quantile (scores at episode $j$): $q_j\gets q_{1-\bar\alpha_j}\!\big(\{s(\hat y_{0:T},y^{(i)}_{j,0:T})\}_{i=1}^{n_j}\big)$.
\IF{\texttt{SOLVER\_TYPE}=\text{'EXPLICIT'}}
\STATE \{use analytical solution from Lemma~\ref{thm:explicit_update_rule}\}
\IF{$q_j\le r_j$}
\STATE $r_{j+1}\gets(q_j+\kappa r_j)/(1+\kappa)$
\ELSE
\STATE $r_{j+1}\gets(q_j-\kappa r_j)/(1-\kappa)$
\ENDIF
\ELSIF{\blue{\texttt{SOLVER\_TYPE}=\text{'IMPLICIT'}}}
\STATE \blue{\{use iterative solver for equation \eqref{eq:implicit_problem}, see Appendix~\ref{sec:implicit_solver}\}}
\STATE \blue{$r_{j+1}\gets\mathrm{FindRoot}\!\big(r\mapsto r-(q_j+\beta_T\|\pi^\star(r)-\pi_j\|_\infty)\big)$}
\ENDIF
\STATE \emph{Certify and deploy:} solve $\mathbf{P}[j{+}1;r_{j+1}]$ to get policy $\pi_{j+1}\gets\pi^\star(r_{j+1})$.
\STATE \emph{Monitor:} record $|r_{j+1}-r_j|$ and $\|\pi_{j+1}-\pi_j\|_\infty$.
\STATE \textbf{if} changes are below threshold \textbf{then break.}
\ENDFOR
\end{algorithmic}
\end{algorithm}
 \fi
\ifincludeAppSolver \section{The Implicit Solver Algorithm}
\label{sec:implicit_solver}

We here detail the implicit solver of Approach~1, as presented in Section~\ref{sec:iterative_planning}. At episode $j$, the quantities $q_j$, $\pi_j$, and $\beta_T$ are fixed, so the implicit safety requirement in equation \eqref{eq:implicit_problem} reduces to a scalar one-dimensional program with residual
\[
g_j(r)\ :=\ r\ -\ \Big(q_j + \beta_T \big\|\pi^\star(r)-\pi_j\big\|_{\infty}\Big),
\]
and we seek the smallest $r\in[q_j,r_{\max}]$ such that $g_j(r)\ge 0$. Each evaluation of $g_j(r)$ requires one call to the planner $\pi^\star(\cdot)$, i.e., one solve of $\mathbf{P}[j{+}1;\,r]$ in equation \eqref{eq:episodic} (warm-started at $\pi_j$ in our implementation). In MATLAB, we handle this one-dimensional constrained problem using either (i) a bracketed line search on $r$ to find a feasible point by increasing $r$ until $g_j(r)\ge 0$ or $r=r_{\max}$, followed by bisection to a tolerance $\varepsilon$ when the feasible set is observed to be an interval, or (ii) a generic constrained optimizer such as \texttt{fmincon} with objective $\min_r r$ and constraint $g_j(r)\ge 0$. %The theoretical guarantees of Theorem~\ref{thm:per_episode_safety} correspond to the global bound $\beta_T$ from Lemma~\ref{prop:coupling_sensitivity}, while the data-driven variant substitutes $\beta_T\leftarrow\widehat\beta_T$ as described in the practical instantiations paragraph in Section~\ref{sec:iterative_planning}.

\section{Practical Instantiations of Sensitivity Constants}
\label{app:practical_instantiations}

Our analysis treats $\beta_T$ from Lemma~\ref{prop:coupling_sensitivity} and $L_U$ from Assumption~\ref{ass:planner_sensitivity} as fixed global upper bounds so that the closed-loop gain $\kappa=\beta_TL_U$ is constant. The closed-form expressions for $\beta_T$  in Appendix~\ref{app:proof_beta_t}  results from recursive unfolding of the coupled dynamics~\eqref{eq:pf:dyn}; our argument shows that $\beta_T$ grows like $(1+cL)^T$ in the horizon $T$, where $c$ is a constant $L$ is the largest Lipschitz constant in Assumption~\ref{ass:lipschitz_dynamics}. For even moderate $T$, this bound can be overly conservative.

A practical alternative is to replace $\beta_T$ and $L_U$ by data-driven estimates: the maximum pairwise slopes of the maps $r \mapsto \pi^\star(r)$ and $r \mapsto y_{0:T}(\pi^\star(r))$ evaluated on a finite grid in $\calR$. These are the kinky-inference-style Lipschitz estimators of~\cite{calliess2017lazily,huang2023sample}. Under mild density assumptions,~\cite{huang2023sample} shows these are certified upper bounds with confidence $1-\delta_\beta$. Since the robustification term $M_{j+1}$ then depends on estimation data, this substitution degrades the outer confidence in Theorem~\ref{thm:per_episode_safety} from $1-\delta$ to $1-\delta-\delta_\beta$ by a union bounding argument. The same applies to the implicit solver and to Lemma~\ref{thm:explicit_update_rule}.

% \section{The Implicit Solver Algorithm}
% \label{sec:implicit_solver}

% \begingroup\color{blue}
% To solve the implicit requirement \eqref{eq:implicit_problem} (or its local-constant variant with $\beta_T$ replaced by $\beta_{T,j}$), we define the scalar residual
% \[
% g_j(r)\ :=\ r\ -\ \Big(q_j + \beta_{T,j} \big\|\pi^\star(r)-\pi_j\big\|_{\infty}\Big),
% \]
% and seek the \emph{smallest} $r\in[q_j,r_{\max}]$ such that $g_j(r)\ge 0$.
% Each evaluation of $g_j(r)$ requires one call to the planner $\pi^\star(\cdot)$, i.e., one solve of $\mathbf{P}[j{+}1;\,r]$ in \eqref{eq:episodic} (warm-started at $\pi_j$ in our implementation).

% In MATLAB, we handle this one-dimensional constrained problem using either (i) a bracketed line-search on $r$ to find a feasible point (increase $r$ until $g_j(r)\ge 0$ or $r=r_{\max}$), followed by bisection to a tolerance $\varepsilon$ when the feasible set is observed to be an interval, or (ii) a generic constrained optimizer (\texttt{fmincon}) with objective $\min_r r$ and constraint $g_j(r)\ge 0$.
% Theoretical guarantees correspond to $\beta_{T,j}\equiv \beta_T$ from Proposition~\ref{prop:coupling_sensitivity}, whereas the fully perturbation-based variant sets $\beta_{T,j}\equiv \widehat\beta_{T,j}$ as described in Section~\ref{sec:iterative_planning}.
% \endgroup
 \fi       
\ifshortcontent\input{sections/app_quadswarm}\fi
\end{document}